\documentclass[a4paper,USenglish,cleveref, autoref, thm-restate]{lipics-v2021}
\usepackage{listings, fancyvrb}
\usepackage{amsmath, amsthm}
\usepackage{xcolor}
\usepackage{graphicx}

\hideLIPIcs 

\def\draft{1}  
\if\draft 0
\newcommand{\mycomment}[3]{{\color{#2}{[\textbf{#1: #3}]}}}
\newcommand{\myedit}[2]{{\color{#1}{#2}}\normalcolor}
\newcommand{\rebuttaledit}[2]{{\color{#1}{#2}}\normalcolor}
\usepackage{showkeys}
\else
\newcommand{\mycomment}[3]{}
\newcommand{\rebuttaledit}[2]{#2}
\newcommand{\myedit}[2]{#2}
\fi

\newcommand{\Eric}[1]{\mycomment{Eric}{magenta}{#1}}

\newcommand{\Hao}[1]{\mycomment{Hao}{brown}{#1}}

\newcommand{\Yihan}[1]{\mycomment{Yihan}{purple}{#1}}
\newcommand{\Youla}[1]{\mycomment{Youla}{cyan}{#1}}

\newcommand{\er}[1]{\myedit{magenta}{#1}}
\newcommand{\y}[1]{\myedit{cyan}{#1}}
\newcommand{\hedit}[1]{\myedit{brown}{#1}}
\newcommand{\rebuttal}[1]{\rebuttaledit{cyan}{#1}}

\newcommand{\dnode}{DNode}
\newcommand{\vnode}{VNode}
\newcommand{\rcptr}{\texttt{rc\_ptr}}
\newcommand{\arcptr}{\texttt{arc\_ptr}}
\newcommand{\rangetracker}{range-tracking}

\newcommand{\downptr}{down-pointer}
\newcommand{\lrunreachable}{lr-unreachable}

\newcommand{\lrreachable}{lr-reachable}
\newcommand{\upptr}{up-pointer}

\newcommand{\future}[1]{} 

\newcommand{\hil}[1]{\color{blue}{#1}}

\newcommand{\myparagraph}[1]{\subparagraph*{#1}} 

\newcommand{\op}[1]{\co{#1}} 

\newcommand{\retireInterval}{\op{deprecate}}
\newcommand{\Announce}{\op{announce}}
\newcommand{\Unannounce}{\op{unannounce}}

\newcommand{\RT}{{\sc RangeTracker}}

\newcommand{\vCAS}{versioned CAS}

\newcommand{\announce}{{\tt Ann}}
\newcommand{\retired}{\co{LDPool}}

\newcommand{\Needed}{\op{Needed}}
\newcommand{\Redundant}{\op{Redundant}}
\newcommand{\queue}{{\tt Q}}
\newcommand{\Empty}{\mbox{$\bot$}}

\newcommand{\merge}{\op{merge}}
\newcommand{\sortAnnouncements}{\op{sortAnnouncements}}

\newcommand{\intersect}{\op{intersect}}
\newcommand{\Split}{\op{split}}

\makeatletter
\lst@Key{countblanklines}{true}[t]%
{\lstKV@SetIf{#1}\lst@ifcountblanklines}

\def\StartLineAt#1{\lstset{firstnumber=#1}}

\lst@AddToHook{OnEmptyLine}{%
	\lst@ifnumberblanklines\else%
	\lst@ifcountblanklines\else%
	\advance\c@lstnumber-\@ne\relax%
	\fi%
	\fi}
\makeatother

\newcommand{\here}[1]{{\bf [[[ #1 ]]]}}

\newcommand{\ignore}[1]{}

\ignore{
\newtheorem{theorem}{Theorem}
\newtheorem{lemma}[theorem]{Lemma}
\newtheorem{definition}[theorem]{Definition}
\newtheorem{observation}[theorem]{Observation}
\newtheorem{corollary}[theorem]{Corollary}
}
\newtheorem{invariant}[theorem]{Invariant}

\newcommand{\MRS}{MRS}

\newcommand{\co}[1]{\mbox{\tt #1}}
\newcommand{\A}{\co{A}}
\newcommand{\B}{\co{B}}
\newcommand{\C}{\co{C}}
\newcommand{\D}{\co{D}}
\newcommand{\E}{\co{E}}
\newcommand{\U}{\co{U}}
\newcommand{\V}{\co{V}}
\newcommand{\W}{\co{W}}
\newcommand{\X}{\co{X}}
\newcommand{\Y}{\co{Y}}
\newcommand{\Z}{\co{Z}}
\newcommand{\ts}{\co{ts}}

\newcommand{\sul}{\co{spliceUnmarkedLeft}}
\newcommand{\sur}{\co{spliceUnmarkedRight}}
\newcommand{\spl}{\co{splice}}
\newcommand{\frozen}{\co{frozen}}

\theoremstyle{definition}
\newtheorem{assumption}[theorem]{Assumption}

\title{Space and Time Bounded Multiversion Garbage Collection}

\author{Naama Ben-David}{VMware Research, USA}{bendavidn@vmware.com}{}{}

\author{Guy E. Blelloch}{Carnegie Mellon University, USA}{guyb@cs.cmu.edu}{}{}

\author{Panagiota Fatourou}{FORTH ICS and University of Crete, Greece}{faturu@csd.uoc.gr}{}{}

\author{Eric Ruppert}{York University, Canada}{ruppert@eecs.yorku.ca}{}{}

\author{Yihan Sun}{University of California, Riverside, USA}{yihans@cs.ucr.edu}{}{}

\author{Yuanhao Wei}{Carnegie Mellon University, USA}{yuanhao1@cs.cmu.edu}{}{}

\funding{\it{Guy E. Blelloch and Yuanhao Wei}: NSF CCF-1901381, CCF-1910030, and CCF-1919223. \newline Eric Ruppert: NSERC Discovery Grant. Yihan Sun: NSF grant CCF-2103483. \newline {\it Panagiota Fatourou}: EU Horizon 2020, Marie Sklodowska-Curie GA No 101031688.}

\keywords{Lock-free, data structures, memory management, snapshot, version lists}
\ccsdesc{Theory of computation~Concurrent algorithms}
\ccsdesc{Theory of computation~Data structures design and analysis}

\authorrunning{N. Ben-David, G.\,E. Blelloch, P. Fatourou, E. Ruppert, Y. Sun, and Y. Wei}

\Copyright{Naama Ben-David, Guy E. Blelloch, Panagiota Fatourou, Eric Ruppert, Yihan Sun, Yuanhao Wei}

\acknowledgements{We thank the anonymous DISC 2021 referees for their helpful comments and suggestions.}

\EventEditors{Seth Gilbert}
\EventNoEds{1}
\EventLongTitle{35th International Symposium on Distributed Computing (DISC 2021)}
\EventShortTitle{DISC 2021}
\EventAcronym{DISC}
\EventYear{2021}
\EventDate{October 4--8, 2021}
\EventLocation{Freiburg, Germany (Virtual Conference)}
\EventLogo{}
\SeriesVolume{209}
\ArticleNo{39}

\begin{document}
\nolinenumbers
  \maketitle

\begin{abstract}
We present a general technique for garbage collecting old versions for multiversion concurrency control that simultaneously achieves good time and space complexity. Our technique takes only $O(1)$ time on average to reclaim each version and maintains only a constant factor more versions than needed (plus an additive term).
It is designed for multiversion schemes using version lists, which are the most common.

Our approach uses two components that are of independent interest.
First, we define a novel {\it range-tracking data structure} which stores a set of old versions and efficiently finds those that are no longer needed.
We provide a wait-free implementation in which all operations take amortized constant time.
Second, we represent version lists using a new lock-free doubly-linked list algorithm that supports efficient (amortized constant time) removals given a pointer to any node in the list.
These two components naturally fit together to solve the multiversion garbage collection problem--the range-tracker identifies which versions to remove and our list algorithm can then be used to remove them from their  version lists.
We apply our garbage collection technique to generate end-to-end time and space bounds for the multiversioning system of Wei et al.\ (PPoPP 2021).

\end{abstract}

\newcommand{\interval}{interval}

\section{Introduction}

Supporting multiple ``historical'' versions of data, often called multiversioning
or multiversion concurrency control, is a powerful technique widely used in
database systems \cite{Reed78,BG83,papadimitriou1984concurrency,HANAgc,neumann2015fast,Wu17},
transactional memory~\cite{perelman2010maintaining,FC11,Perelman11,Kumar14,Keidar2015},
and shared data structures~\cite{BBB+20,FPR19,NHP21,WBBFRS21a}.
This approach allows complex queries (read-only transactions)
to proceed concurrently with updates while still appearing
atomic because they get data views that are consistent with a single point in time.
If implemented carefully, queries do not interfere with one another or with updates.
The most common approach for multiversioning uses version lists~\cite{Reed78} (also called version chains):
the system maintains a global timestamp that increases over time,
and each object maintains a history of its updates as a list of value-timestamp pairs,
each corresponding to a value written and an update time.
Each node in the list has an associated \interval{} of time from that node's timestamp until the next (later) node's timestamp.
A query can first read a timestamp value $t$ and then, for each object
it wishes to read, traverse the object's version list to find the version whose \interval{} contains $t$.

Memory usage is a key concern for
multiversioning, since
multiple versions can consume huge amounts of memory.
Thus, most previous work on multiversioning
discusses how to reclaim the memory of old versions. \hedit{We refer to this as the multiversion garbage collection (MVGC) problem}.
A widely-used approach is to keep track of the earliest active query and reclaim the memory of any versions overwritten before \hedit{the start of this query} \cite{FC11,neumann2015fast,kim2019mvrlu,NHP21,WBBFRS21a}.
However, a query that runs for a long time, either because it is complicated or because it has
been delayed, will force the system to retain many unneeded intermediate
versions \hedit{between the oldest required version and the current one}.
This has been observed to be a major bottleneck for database systems with
Hybrid Transaction and Analytical Processing (HTAP) workloads~\cite{bottcher2019scalable}
(i.e., many small updates concurrent with some large analytical queries).
\rebuttal{To address this problem in the context of software transactional memory, Lu and Scott \cite{lu2013generic} proposed a non-blocking
algorithm that can reclaim intermediate versions.  Blocking techniques were later proposed by the database community~\cite{bottcher2019scalable,HANAgc}.
However, these techniques add significant time overhead in worst-case executions.
}

We present a wait-free MVGC scheme that achieves good time and space bounds, using $O(1)$ time\footnote{\hedit{For time/space complexity, we count both local and shared memory operations/objects.}} on average per allocated version and maintaining only a constant factor more versions than needed (plus an additive term).
The scheme is very flexible and it can be used in a variety of multiversioning implementations.
It uses a three-step approach that involves 1) identifying versions that can be reclaimed, including intermediate versions, 2) unlinking them from the version lists, and 3) reclaiming their memory.
To  implement these three steps efficiently, we develop
two general components---a \emph{range-tracking} data structure and a \emph{version-list}
data structure---that could be of independent interest beyond MVGC.


%
%

The range-tracking data structure is used to identify version list nodes
that are no longer needed.
It supports an \co{announce} operation that
is used by a query to acquire the current timestamp $t$ as well as 
protect any versions that were current at $t$ from being reclaimed.
A corresponding \co{unannounce} is used \y{to} indicate when the query is finished.
The data structure also supports a \co{\retireInterval} operation that is given
a version and its time \interval, and indicates that the version is no
longer the most recent---i.e., is safe to reclaim once its interval
no longer includes any announced timestamp.
When a value is updated with a new version, the previous version is
deprecated. 
A call to \co{\retireInterval} also returns a list of versions that had
previously been \co{deprecated} and \hedit{are} no longer cover any announced
timestamp---i.e., are now safe to reclaim.
%
We provide a novel implementation of the range-tracking data structure for which the amortized
number of steps per operation is $O(1)$.  
We also bound the number of versions on which \retireInterval\ has been called, but have not yet been returned.
If $H$ is the maximum, over all configurations, of the number of needed deprecated versions, then the number of deprecated versions that have not yet been returned is at most $2H + O(P^2 \log P)$, where $P$ is the number of processes.
\hedit{To achieve these time and space bounds,}
we borrow some ideas from real-time garbage collection~\cite{Baker78,BC99}, 
and add several new ideas such as batching and using a shared queue.


The second main component of our scheme is a wait-free version-list data structure
that supports efficient \hedit{(amortized constant time)} removals of nodes \hedit{from anywhere in the list.
When the \retireInterval\ operation} \er{identifies an unneeded version, we must splice it out of its version list, without knowing its current predecessor in the list, so we need a doubly-linked version list.}
\hedit{Our} doubly-linked list implementation has certain restrictions that are naturally satisfied when
maintaining version lists, for example nodes may be appended only at one end.
The challenge
is in achieving constant amortized time per \hedit{remove}, and bounded space.
Previously known concurrent doubly-linked lists~\cite{ST08,Sha15} do not meet these requirements, \hedit{requiring at least $\Omega(P)$ amortized time per \hedit{remove}.}
We first describe the implementation of our version list assuming
\hedit{a garbage collector and then we show how to manually reclaim removed nodes while maintaining our desired overall time and space bounds.}


To delete elements from the list efficiently, we leverage some recent ideas from randomized parallel list contraction~\cite{BFGS20}, which asynchronously removes elements from a list.  To avoid concurrently splicing out
adjacent elements in the list, \er{which can cause problems}, the approach defines an implicit binary tree \er{so that the list is an in-order traversal of the tree.  Only nodes corresponding to leaves of the tree, which cannot be adjacent in the list, may be spliced out.}
\hedit{Directly applying this technique, however, is not efficient in our setting.
To reduce space overhead, we had to develop intricate helping mechanisms for splicing out internal
nodes rather than just leaves.
To achieve wait-freedom, we had to skew the implicit tree so that it is right-heavy.}
The final algorithm ensures that at most $2(L-R)+O(P\log L_{max})$
nodes remain reachable in an execution with $L$ appends and $R$
removes across an arbitrary number of version lists\hedit{, and at most $L_{max}$ appends on a single version list.}
\hedit{This means the version lists store at most a constant factor more than the $L-R$ required nodes plus an additive term shared across all the version lists.}
\hedit{Combining this with the bounds from the range tracker, our MVGC scheme ensures that at most
$O(V + H + P^2 \log P + P\log L_{max})$ versions are reachable from the $V$ version lists.
This includes the current version for each list, $H$ needed versions, plus additive terms from the range tracking and list building blocks.
}

After a node has been spliced out of the doubly-linked list, its memory must be reclaimed.
This step may be handled automatically by the garbage collector in languages such as Java, but in non-garbage-collected languages, additional mechanisms are needed to safely reclaim memory.
The difficulty in this step is that while a node is being spliced out, \er{other processes traversing the list might be visiting that node.}
We use a reference counting reclamation scheme \hedit{and this requires modifying our doubly-linked list algorithm slightly to maintain the desired space bounds.}
\er{We apply an existing concurrent reference counting implementation \cite{ABW21} that}
\hedit{employs a local hash table per process which causes the time bounds of our reclamation to become amortized} \er{$O(1)$ \emph{in expectation}.}
\hedit{It also requires an additional fetch-and-add instruction, whereas the rest of our algorithms require only read and CAS.}

We apply \hedit{our MVGC scheme} to a specific multiversioning scheme~\cite{WBBFRS21a} \hedit{to generate end-to-end bounds for a full multiversioning system}.
This multiversioning scheme takes a given CAS-based concurrent
data structure and transforms it to support \hedit{complex} queries \hedit{(e.g., range queries)}
by replacing each CAS object with one that maintains a  version list.
\hedit{Overall, we ensure that the memory usage of the multiversion data structure is within a constant factor of the needed space, plus $O(P^2\log P + P^2 \log L_{max})$.}
\rebuttal{In terms of time complexity, our garbage collection scheme takes only $O(1)$ time on average for each allocated version.}

Detailed proofs of correctness and of our complexity bounds appear in the full version~\cite{arxiv}.


\section{Related Work}
\label{sec:related}
\myparagraph{Garbage Collection.}
One of the simplest, oldest techniques for garbage collection
is reference counting (RC)~\cite{correia2021orcgc,DM+01,HL+05}.
In its basic form, RC attaches to each object a counter of the number
of references to it. 
An object is reclaimed when its counter reaches zero.
Some variants of RC are wait-free~\cite{ABW21,S05}.
In Section~\ref{sec:smr-list}, we apply the RC scheme of~\cite{ABW21}
to manage version list nodes \hedit{as it adds only constant time overhead (in expectation) and it is the only concurrent RC scheme that maintains our desired time bounds.}

Epoch-based reclamation (EBR)~\cite{Fra04,B15} employs a counter that is incremented periodically
and
is used to divide the execution into epochs.
Processes read and announce the counter value at the beginning of an operation.
An object can be reclaimed only if it was retired
in an epoch preceding the oldest announced.
EBR is often the preferred choice in practice,
as it is simple and exhibits good performance.
However, 
a slow or crashed process with timestamp $t$ can prevent the reclamation
of {\it all} retired objects with timestamps larger than $t$.
EBR, or variants, are used in a variety of MVGC
schemes~\cite{FC11,neumann2015fast,WBBFRS21a}
to identify versions that are older than any query.  
\hedit{An advantage of these schemes is that identified versions can be immediately reclaimed without first being unlinked from the version lists because the section of the version list they belong to is old enough to never be traversed.}
\hedit{However, they} inherit the same problem as EBR and are not able to reclaim intermediate
versions between the oldest needed version and the current version when a long-running query holds on to an old epoch.   This
can be serious for multiversioned systems since EBR works best when
operations are short, but a key motivation for multiversioning is to support lengthy queries.


Hazard pointers (HP)~\cite{HL+05,M04} can be used to track which objects
are currently being accessed by each process and are therefore more
precise.
\hedit{Combinations of HP and EBR have been proposed 
(e.g.~\cite{RC17,WI+18-I}) 
with the goal of preserving the practical efficiency of EBR while lowering its memory usage.
However, unlike EBR, none of these techniques directly solve the MVGC problem.
Other memory reclamation schemes have been studied that require hardware support~\cite{AE+14,DH+11}
or rely on the signaling mechanism of the operating system~\cite{B15,SBM21}.
Hyaline~\cite{nikolaev2021snapshot} implements a similar interface to EBR and can be used for MVGC, but like EBR, it cannot reclaim intermediate versions.
}

\hedit{
  We are aware of \rebuttal{three} multiversioning systems based on version lists that reclaim
  intermediate versions: \rebuttal{GMV~\cite{lu2013generic}}, HANA~\cite{HANAgc} and
  Steam~\cite{bottcher2019scalable}.
  To determine which versions are
  safe to reclaim, \rebuttal{all three} systems merge the current version list for an object{}
  with the list of active timestamps to check for overlap.
  \rebuttal{The three schemes differ based on when they decide to perform this merging step and how they remove and reclaim version list nodes.}
  \rebuttal{In GMV, when an update operation sees that memory usage has passed a certain threshold, it iterates through all the version lists to reclaim versions.
  Before reclaiming a version, it has to 
  help other processes traverse the version list to ensure traversals remain wait-free.
  }%
  HANA uses a background thread to
  \rebuttal{identify and reclaim} obsolete versions while Steam scans the entire version list
  whenever a new version is added to it.
  In \rebuttal{HANA and Steam}, nodes are
  removed by locking the entire version list, \rebuttal{whereas in GMV, nodes are removed in a lock-free manner by first logically marking a node for deletion, as in Harris's linked list \cite{harris2001pragmatic}. If a \op{remove} operation in GMV experiences contention (i.e., fails a CAS), it restarts from the head of the version list. None of these three techniques ensure constant-time removal from a version list.}
  \rebuttal{Both Steam and GMV ensure $O(PM)$ space where $M$ is the amount of space required in an equivalent sequential execution.
  In comparison, we use a constant factor more than the required space plus an additive term of $O(P^2 \log P + P^2 \log L_{max})$, where $L_{max}$ is the maximum number of versions added to a single version list. This can be significantly less than $O(PM)$ in many workloads.
  
  }
  }

\myparagraph{Lock-Free Data Structures and Query Support.}
We use doubly-linked lists to store old versions.
Singly-linked lists had lock-free implementations as early as 1995 \cite{Val95}.
Several implementations of doubly-linked lists were developed later 
from multi-word CAS instructions~\cite{AH13,Gre02}, which are not widely available in hardware
\hedit{but can be simulated in software~\cite{harris2002practical,guerraoui2020efficient}.}
Sundell and Tsigas \cite{ST08} gave the first implementation from single-word CAS,
although it lacks a full proof of correctness.
Shafiei~\cite{Sha15} gave an implementation with a proof of correctness and
amortized analysis.
\hedit{Existing doubly-linked lists are not}
\er{efficient enough for our application,} \er{so we} 
give a new implementation \hedit{with better time bounds.}

Fatourou, Papavasileiou and Ruppert \cite{FPR19} used multiversioning to add
range queries to a search tree~\cite{EFR10}.
Wei et al.~\cite{WBBFRS21a} generalized this approach (and made it more efficient)
to support wait-free queries on a large class of lock-free data structures.
Nelson, Hassan and Palmieri \cite{NHP21} sketched a similar scheme, but
it is not non-blocking.
In Appendix~\ref{sec:application}, we apply our garbage collection scheme to the multiversion system of~\cite{WBBFRS21a}.
}


\section{Preliminaries}

We use a standard model with asynchronous, crash-prone processes that access shared memory
using CAS, read and write instructions.  
For our implementations of data structures, we bound the number of steps needed to perform operations,
and the number of shared objects that are allocated but not yet reclaimed. 

\label{dest}
We also use {\em destination objects}~\cite{BW20a}, which are single-writer objects that store a value and support 
 \op{swcopy} operations in addition to standard reads and writes.
A \op{swcopy(ptr)}  atomically reads the value pointed to by \co{ptr}, and
copies the value into the destination object.
Only the owner of a destination object can perform \op{swcopy} and \op{write}; any process may \op{read} it.
Destination objects can be implemented from CAS  so that all three
operations take $O(1)$ steps~\cite{BW20a}.
They are used to implement our range-tracking objects in Section~\ref{sec:identify}.

\myparagraph{Pseudocode Conventions.}
We use syntax similar to C++.
The type \co{T*} is a pointer to an object of type \co{T}.
\co{List<T>} is a \co{List} of objects of type \co{T}.
If \co{x} stores a pointer to an object, then \co{x->f} is that object's member \co{f}.
If \co{y} stores an object, \co{y.f} is that object's member \co{f}.


\section{Identifying Which Nodes to Disconnect from the Version List}
\label{sec:identify}

We present the \emph{range-tracking} object, which we use
to identify version nodes that are safe to disconnect from version lists because
they are no longer needed.
To answer a query, a slow process
may have to traverse an entire version list when searching for a
very old version. 
However, 
we need only maintain list nodes that are the potential target nodes of such queries.
The rest 
may be spliced out of the list \er{to improve space usage and traversal times}.


\er{We assign  to each version node \X\ an interval that represents
the period of time when \X\ was the current version.
When the next version \Y\ is appended to the version list, 
\X\ ceases to be the current version and becomes a potential candidate for
removal from the version list (if no query needs it).
Thus, the left endpoint of \X's interval is the timestamp assigned to \X\ by the multiversioning system,
and the right endpoint is the timestamp assigned to~\Y.}

We assume that a query starts by announcing a timestamp $t$,
and then proceeds to access, for each relevant object $o$, its corresponding version at time $t$,
by finding the first node in the version list with timestamp at most $t$ (starting from the most recent version).
Therefore, an announcement of $t$ means it is unsafe to disconnect any nodes whose intervals contain $t$.

\er{As many previous multiversioning systems~\cite{FC11,HANAgc,NHP21,neumann2015fast,WBBFRS21a}
align with the general scheme discussed above,
we define the range-tracking object to abstract the problem of identifying  versions that are not needed.  We believe this abstraction is of general interest.}

\begin{definition}[Range-Tracking Object]
\label{def:range tracker}
	A range-tracking object maintains a multiset $A$ of integers, and a set $O$
	of triples of the form \co{(o,low,high)} where \co{o} is an object of some type $T$ and
	$\co{low}\leq \co{high}$ are integers.
	Elements of $A$ are called {\em active announcements}.
	If $\co{(o,low,high)}\in O$ then \co{o} is a {\em deprecated object} with associated
	half-open interval $[\co{low}, \co{high})$.
	The range-tracking object supports the following operations.
\begin{itemize}
	\item \co{\Announce(int* ptr)} atomically reads the integer pointed to by \co{ptr}, adds the value read to $A$, and returns the value read.
	\item \co{\Unannounce(int i)} removes one copy of \co{i} from $A$, rendering the announcement inactive.
	\item \co{\retireInterval(T* o, int low, int high)}, where $\co{low} \leq \co{high}$,
	adds the triple \co{(o,low,high)} to $O$
        and returns a set $S$, which contains the deprecated objects of a set $O'\subseteq O$ such that for any $o\in O'$, the interval of $o$
        does not intersect $A$, and removes $O'$ from~$O$.
\end{itemize}
\end{definition}

The specification of Definition~\ref{def:range tracker} should be paired with a progress property
\er{that rules out the trivial implementation in which}
\retireInterval\ always returns an empty set.
We do this by bounding the number of deprecated objects that have not been returned by \retireInterval.

\begin{assumption}
\label{ass}
To implement the range-tracking object, we assume the following.
\hspace{1mm}
\begin{enumerate}
	\item
	\label{ass:ascending}
	A process's calls to \retireInterval\ have non-decreasing values of parameter \co{high}.
	\item
	\label{ass:previously}
	If, in some configuration $G$, there is a pending \Announce\ whose argument is a pointer to an integer variable \co{x}, then the value of \co{x} at $G$ is greater than or equal to the \co{high} argument of every \retireInterval\ that has been invoked before $G$.
%
	\item
	\label{ass:announce}
	For every process $p$, the sequence of invocations to \Announce\ and \Unannounce\
		performed by $p$ should have the following properties:
		a) it should start with \Announce;
		b) it should alternate between invocations of \Announce\ and invocations of \Unannounce;
		c) each \Unannounce\ should have as its argument the integer returned by the preceding \Announce.
	\item
	\label{ass:distinct}
	Objects passed as the first parameter to \retireInterval\ operations are distinct.
\end{enumerate}
\end{assumption}

In the context we are working on, we have a non-decreasing
integer variable that works as a global timestamp,
and is passed as the argument to every \Announce\ operation.
Moreover, the \co{high} value passed to each \retireInterval\ operation is a value that has been read from
this variable.  This ensures that parts \ref{ass:ascending} and \ref{ass:previously}
of Assumption \ref{ass} are satisfied. The other parts of the assumption are also
satisfied quite naturally for our use of the range-tracking object, and
we believe that the assumption is reasonably general. 
Under this assumption, we present and analyze a linearizable implementation of the range-tracking object
in Section~\ref{sec:rt-implem}. 

\subsection{A Linearizable Implementation of the Range-Tracking Object}
\label{sec:rt-implem}

\begin{figure*}[t]\small
	\begin{minipage}[t]{.48\textwidth}
		\begin{lstlisting}[linewidth=.99\columnwidth, numbers=left,frame=none, language=C++, morekeywords=class]
class Range { T* t, int low, int high; };	@\label{range-object}@
class RangeTracker {
	// global variables
	Destination Ann[P];
	Queue<List<Range>> Q; //initially empty
	// thread local variables
	List<Range> LDPool; // initially empty
	Array<int> @\sortAnnouncements @() {
		List<int> result;
		for(int i = 0; i < P; i++) {
			int num = Ann[i].read();      @\label{line:Ann_read}@
			if(num != @\Empty @) result.append(num); }
		return sort(toArray(result)); }
	
	List<T*>, List<Range> intersect(
	    List<Range> MQ, Array<int> ar) {
		Range r; int i = 0;
		List<T*> Redundant; 
		List<Range> Needed; 
		for(r in MQ) {     @\label{line:for}@
			while(i < ar.size() &&
						ar[i] < r.high) i++;
			if(i == 0 || ar[i-1] < r.low)  @\label{line:intersect-test}@
				Redundant.append(r.t);   @\label{line:RedAppend}@
			else Needed.append(r); }   @\label{line:NeedAppend}@
		return <Redundant, Needed>; }	
		\end{lstlisting}
	\end{minipage}\hspace{3mm}
	\begin{minipage}[t]{.55\textwidth}
		\StartLineAt{27}
		\begin{lstlisting}[linewidth=.99\textwidth, numbers=left, frame=none]
int Announce(int* ptr) {
	Ann[p].swcopy(ptr);			@\label{line:swcopy}@
	return Ann[p].read(); }     @\label{line:swread}@

void @\Unannounce@() { Ann[p].write(@\Empty @); }			@\label{line:write-empty}@

List<T*> @\retireInterval @(T* o, int low, int high) {
	List<T*> Redundant;
	List<Range> Needed, Needed1, Needed2;
	// local lists are initially empty
 	LDPool.append(Range(o, low, high)); @\label{line:retint-lin}@
 	if(LDPool.size() == B) { @\label{line:ifFreeNeeded}@
 		List<Range> MQ = merge(Q.deq(),Q.deq()); @\label{line:deq}@
 		Array<int> ar = @\sortAnnouncements @(); @\label{line:scan}@
		Redundant, Needed = intersect(MQ, ar); @\label{line:intersect}@		
		if(Needed.size() > 2*B) {    @\label{line:checksize}@
			Needed1, Needed2 = split(Needed);   @\label{line:split}@
			Q.enq(Needed1); @\label{line:enq1}@
			Q.enq(Needed2); } @\label{line:enq2}@
		else if(Needed.size() > B) {    @\label{line:ifNeeded}@
			Q.enq(Needed); }	@\label{line:enq3}@
		else {				@\label{line:else}@
			LDPool = merge(LDPool,Needed); }	@\label{line:mergeLDP}@
	 	Q.enq(LDPool); @\label{line:enq-retired}@
	 	LDPool = empty list; } @\label{line:list-clear}@
 	return Redundant; } };
		\end{lstlisting}
	\end{minipage}
	\vspace{-.1in}
	\caption{Code for process \co{p} for our linearizable implementation of a range-tracking object.
	}
	\label{fig:range-tracker}
\end{figure*}

Our implementation, \RT, is shown in Figure~\ref{fig:range-tracker}.
Assumption \ref{ass}.\ref{ass:announce} means that each process
can have at most one active announcement at a time.
So, \RT\ maintains a shared array \announce\ of length $P$
to store active announcements.
\co{\announce[p]} is a destination object (defined in Section \ref{dest}) that is owned by process \co{p}.
Initially, \co{Ann[p]} stores a special value \Empty.
To announce a value, a process \co{p} calls \op{swcopy} (line~\ref{line:swcopy}) to copy the current timestamp into \co{\announce[p]} and returns the announced value (line~\ref{line:swread}).
To deactivate an active announcement, $p$ 
writes \Empty\ into \co{\announce[p]}\ (line~\ref{line:write-empty}).
Under Assumption \ref{ass}.\ref{ass:announce},
the argument to \Unannounce{} must match the argument of the process's previous \Announce,
so we suppress \Unannounce's argument in our code.
An \Announce\ or \Unannounce\ performs O$(1)$ steps.

A Range object (line \ref{range-object})
stores the triple \co{(o,low,high)} for a deprecated object \co{o}.
It is created (at line \ref{line:retint-lin}) during a \retireInterval\ of \co{o}.
\RT\ maintains the deprecated objects as {\em pools} of Range objects.
 Each pool is sorted by its elements' \co{high} values.
Each process maintains a local pool of deprecated objects, called \retired.
To deprecate an object,
a process simply appends its Range
to the process's local \retired\ (line~\ref{line:retint-lin}).  Assumption~\ref{ass}.\ref{ass:ascending}
implies that objects are appended to \retired\ in non-decreasing order 
of their \co{high} values.

We wish to ensure that most deprecated objects are eventually returned by 
a \retireInterval\ operation so that they can be freed.
If a process $p$ with a large \retired\ ceases to take steps, it can
cause all of those objects to remain unreturned.
Thus, when the size of $p$'s \retired\ hits a threshold $B$, they
are flushed to a shared queue, \queue, so that other processes can also return them.
The elements of \queue\ are pools that each contain
$B$ to $2B$ deprecated objects.
For the sake of our analysis, we choose $B=P\log P$.
When a flush is triggered, 
$p$ dequeues two pools from \queue\ and processes them as a batch
to identify the deprecated objects whose intervals do not intersect with the values in \co{Ann},
and return them. 
The rest of the dequeued objects, together with those in \retired, 
are stored back into~\co{Q}.
We call these actions (lines~\ref{line:ifFreeNeeded}--\ref{line:list-clear}), the {\em flush phase} of \retireInterval.
A \retireInterval\ without a flush phase returns an empty set.

During a flush phase, 
a process $p$ dequeues two pools from \queue\  and
merges them  (line~\ref{line:deq}) into a new pool, \co{MQ}. 
Next, $p$ makes a local copy of \announce\ and sorts it  (line~\ref{line:scan}).
It then uses the \co{intersect} function  (line~\ref{line:intersect}) to partition \co{MQ} into two sorted lists:
\Redundant\ contains objects whose intervals do not intersect the local copy
of \announce, and \Needed\ contains the rest.
Intuitively, a deprecated object in \co{MQ} is put in \Redundant\
if the \co{low} value of its interval is larger than the announcement value immediately before its \co{high} value.
Finally, $p$ enqueues the \Needed\ pool with its \retired\ into \queue\ (lines~\ref{line:enq1}--\ref{line:enq3} and line~\ref{line:enq-retired}).
To ensure that the size of each pool in \queue\ is between $B$ and $2B$,
the \Needed\ pool is split into two halves if it is too large (line~\ref{line:split}), or is merged with \retired\ if it is too small (line~\ref{line:mergeLDP}).
A flush phase is performed once every $P\log P$ calls to \retireInterval,
and the phase executes $O(P\log P)$ steps. Therefore, the amortized number of steps for
\retireInterval\ is $O(1)$.

The implementation of the concurrent queue \co{Q} should ensure that an element can be enqueued or dequeued in $O(P \log P)$ steps.
The concurrent queue presented in~\cite{FK14} has step complexity $O(P)$ and thus ensures these bounds.
To maintain our space bounds, the queue nodes must be reclaimed.
This can be achieved if we apply hazard-pointers on top of the implementation in~\cite{FK14}.
If \co{Q} is empty, then \co{Q}.\op{deq()} returns an empty list.

We sketch the proofs of the following three theorems.  For detailed proofs, see~\cite{arxiv}.


\begin{restatable}{theorem}{RTcorrect}
\label{thm:rt-lin}
If Assumption~\ref{ass} holds, then \RT\ is a linearizable implementation of a range-tracking object.
\end{restatable}

The linearization points used in the proof of Theorem \ref{thm:rt-lin} are defined as follows.
An \Announce\ is linearized at its \op{swcopy} on line~\ref{line:swcopy}.
An \Unannounce\ is linearized at its \op{write} on line~\ref{line:write-empty}.
A \retireInterval\ is linearized at line \ref{line:enq-retired} if it executes that line, or at line \ref{line:retint-lin} otherwise.


The most interesting part of the proof concerns
a \retireInterval\ operation $I$ with a flush phase. 
$I$ dequeues two pools from \queue\ as \co{MQ} and decides which objects in \co{MQ} 
to return based on the local copy of \co{Ann} array.  To show linearizability, we must also show
that intervals of the objects returned by $I$ do not intersect the \co{Ann} array at the linearization point of $I$.
Because of Assumption \ref{ass}.\ref{ass:previously}, values written into \co{Ann} after the pools are dequeued cannot be contained in the intervals in \co{MQ}.
Thus, if an object's interval does not contain the value $I$ read from \co{Ann[i]},
it will not contain the value in \co{Ann[i]} at $I$'s linearization point.

\begin{restatable}{theorem}{RTtime}
\label{step1-time}
In the worst case, \Announce\ and \Unannounce\ take $O(1)$ steps, while \retireInterval\ takes $O(P\log P)$ steps.
The amortized number of steps performed by each operation is $O(1)$.
\end{restatable}

Let $H$ be the maximum, over all configurations in the execution, of the number of {\it needed} deprecated objects, i.e., those whose intervals contain an active announcement. 

\begin{restatable}{theorem}{RTspace}
\label{step1-space}
At any configuration, the number of deprecated objects
that have not yet been returned by any instance of \retireInterval\ is at most $2H + 25 P^2\log P$.
\end{restatable}

At any time, each process holds at most $P\log P$ deprecated objects in \retired\ and at most
$4P\log P$ that have been dequeued from \queue\ as part of a flush phase.
We prove by induction that the number of deprecated objects in \queue\ at a configuration $G$ is at most
$2H+O(P^2\log P)$. 
Let $G'$ be the latest configuration before $G$ such that all pools in \queue\ at $G'$
are dequeued between $G'$ and $G$.  Among the dequeued pools, only the objects that were needed at $G'$
are re-enqueued into \queue, and there are at most $H$ such objects.
Since we dequeue two pools (containing at least $B$ elements each) each time we enqueue $B$ {\it new} objects
between $G'$ and $G$, this implies that the number of such new objects is at most half the number of objects
in \queue\ at $G'$ (plus $O(P^2\log P)$ objects from flushes already in progress at $G'$).  Assuming the bound on the size of \queue\ holds at $G'$, this allows us to prove the bound at $G$.  

The constant multiplier of $H$ in Theorem \ref{step1-space} can be made arbitrarily close to 1
by dequeuing and processing $k$ pools of \queue\ in each flush phase instead of two.
The resulting space bound would be $\frac{k}{k-1}\cdot H + \frac{(2k+1)(3k-1)}{k-1}\cdot P^2\log P$.
This would, of course, increase the constant factor in the amortized number of steps performed by \retireInterval\ (Theorem \ref{step1-time}).

\section{Maintaining Version Lists}
\label{sec:vlists}

We use a restricted version of a doubly-linked list to maintain each version list so that
we can more easily remove nodes from the list when they are no longer needed.
We assume each node has a timestamp field.
The list is initially empty and provides the following operations.
\begin{itemize}
	\item \op{tryAppend(Node* old, Node* new)}: Adds \op{new} to the head of the list and returns true if the current head is \op{old}. Otherwise returns false. Assumes \op{new} is not null.
	\item \op{getHead()}: Returns a pointer to the \op{Node} at the head of the list (or \co{null} if list is empty).
	\item \op{find(Node* start, int ts)}: Returns a pointer to the first \op{Node}, starting from \op{start} and moving away from the head of the list, whose timestamp is at most \op{ts} (or \co{null} if no such node exists). 
	\item \op{remove(Node* n)}: Given a previously appended \op{Node}, removes it from the list.
\end{itemize}

\er{To obtain an efficient implementation, we assume several preconditions, summarized in Assumption \ref{list-assumptions} (and stated more formally in the full version \cite{arxiv}).  
A version should be removed from the object's version list only if it is not current: either it has been 
superseded by another version (\ref{list-assumptions}.\ref{append-before-remove-inf}) 
or, if it is the very last version, the entire list is 
no longer needed (\ref{list-assumptions}.\ref{remove-pre-inf}).
Likewise, a version should not be removed if a \co{find} is looking for it (\ref{list-assumptions}.\ref{removals-safe-inf}), which can be guaranteed using our range-tracking object.
We allow flexibility in the way timestamps are assigned to versions.
For example, a timestamp can be assigned to a version after appending it to the list.
However, some assumptions on the behaviour of timestamps are needed to ensure that responses to \co{find} operations are properly defined (\ref{list-assumptions}.\ref{append-pre-inf}, \ref{list-assumptions}.\ref{find-append-inf}).
}

\begin{assumption}\mbox{ }
\label{list-assumptions}
\begin{enumerate}
\er{\item
\label{append-before-remove-inf}
Each node (except the very last node) is removed only after
the next node is appended.  
\item
\label{remove-pre-inf}
No \co{tryAppend}, \op{getHead} or \co{find} is called after a \co{remove} on the very last node.
\item
\label{removals-safe-inf}
After \co{remove(X)} is invoked, no pending or future \co{find} operation should be seeking a timestamp in the interval between \X's timestamp and its successor's.
\item
\label{append-pre-inf}
Before trying to append a node after} \y{a node} \B\ \er{or using \B\ as the starting point for a \co{find}, \B\ has been the head of the list and its
 timestamp has been set.
A node's timestamp does not change after it is set.
Timestamps assigned to nodes are non-decreasing. 
\item
\label{find-append-inf}
If a \co{find(X,t)} is invoked, any node appended after \X\ has a higher timestamp than \co{t}.
\item
\label{no-duplicates-inf}
Processes never attempt to append the same node to a list twice, or to remove it twice.
}
\end{enumerate}
\end{assumption}

\Youla{Maybe move assumptions later. Perhaps before the sketch of proofs if there's no references in 5.1.}

\subsection{Version List Implementation}

Pseudocode for our list implementation is in Figure~\ref{fig:list-alg-full}.
A \co{remove(X)} operation first marks the node \X\ to be deleted by setting a \co{status} field of \X\ to \co{marked}.
We refer to the subsequent physical removal of \X\ as {\it splicing}
\X\ out of the list.
\ignore{The \co{remove(X)} operation may or may not splice \X\ out of the list.
In the latter case, other operations may help to splice \X\ later.}

\begin{figure}[t]
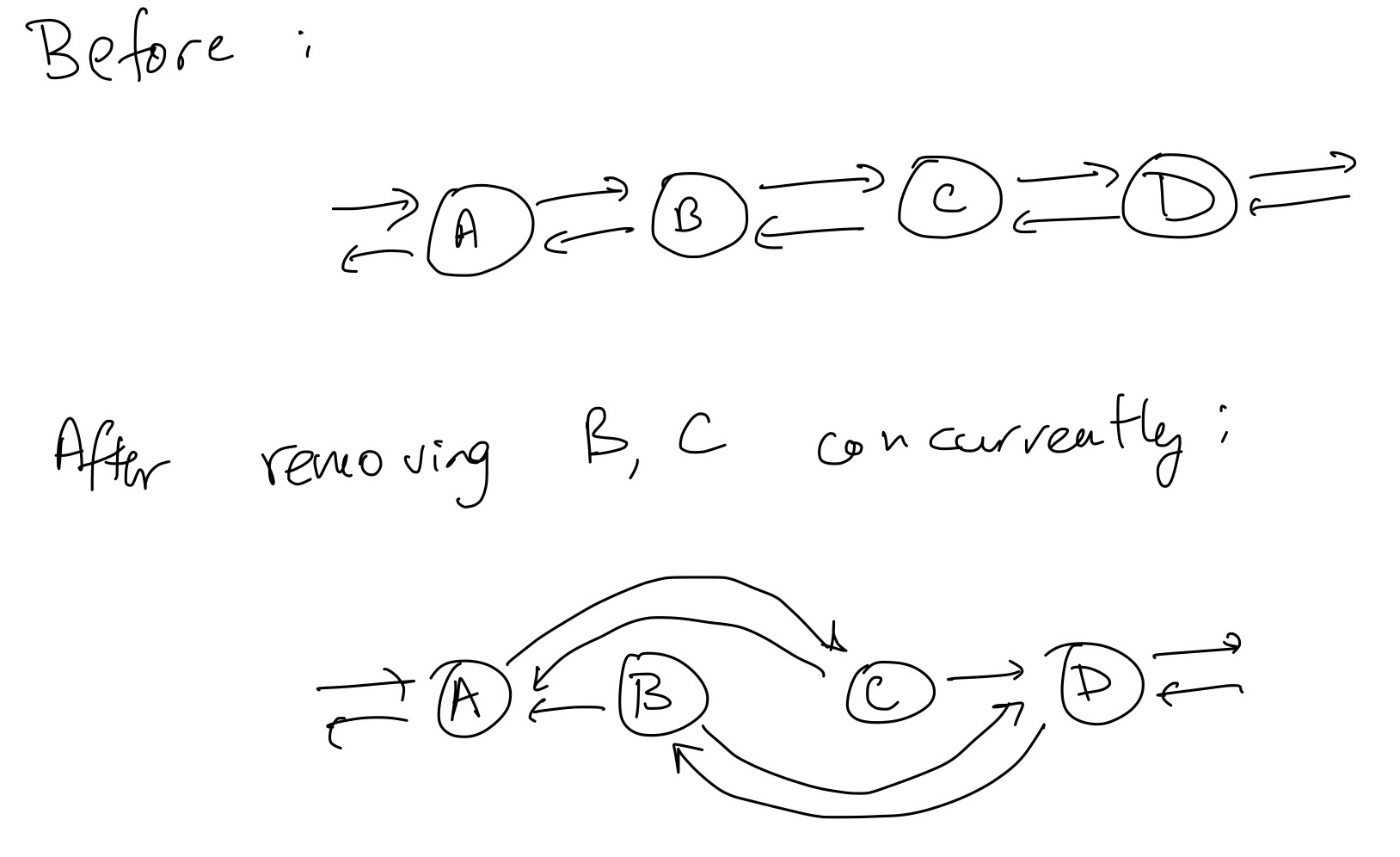
\caption{An example of incorrect removals.}
\label{bad-example}
\end{figure}

Splicing a node \B\ from a doubly-linked list requires  finding
its left and right neighbours, \A\ and \C,
and then updating the pointers in \A\ and \C\ to point  to each other.
Figure~\ref{bad-example} illustrates the problem that could arise if  adjacent nodes \B\ and \C\ are
spliced out concurrently.
The structure of the doubly-linked list becomes corrupted:  \C\ is still reachable when traversing the
list towards the left, and \B\ is still reachable when traversing towards the right.
The challenge of designing our list implementation is to coordinate splices
to avoid this situation.

We begin with an idea that has been used for parallel list contraction \cite{SGBFG15}.
We assign each node a priority value and splice  a node out only if its priority
is greater than both of its neighbours' priorities.
This ensures that two adjacent nodes cannot be spliced concurrently.

Conceptually, we can define a {\it priority tree} corresponding to a list
of nodes with priorities as follows.
Choose the node with minimum priority as the root.
Then, recursively define the left and right subtrees of the root by applying the same procedure to the
sublists to the left and right of the root node.
The original list is an in-order traversal of the priority tree.
See Figure \ref{priority-tree} for an example.
We describe below how we choose priorities to ensure that \y{(1) there is always a unique minimum in a sublist corresponding to a subtree
(to be chosen as the subtree's root), and (2)} if $L$ nodes are appended to the list, the height of the priority
tree is $O(\log L)$.
We emphasize that the priority tree is not actually represented in memory; it is simply
an aid to understanding the design of our implementation.

The requirement that a node is spliced out of the list only if its priority is greater than its neighbours
corresponds to requiring that we splice only nodes whose descendants in the priority tree have all already been spliced out of the list.
To remove a node that still has unspliced descendants, we simply mark it as logically
deleted and leave it in the list.
If \X's descendants have all been spliced out, then \X's parent \Y\ in the priority tree is the neighbour of \X\
in the list with the \er{larger} priority.
An operation that splices \X\ from the list then attempts to {\it help} splice \X's parent \Y\
(if \Y\ is marked for deletion and \Y\ is larger than its two neighbours),
and this process continues up the tree.
Conceptually, this means that if a node \Z\ is marked  but not spliced,
the last descendant of \Z\ to be spliced is also responsible for splicing \Z.

\begin{figure}
\hspace*{22mm}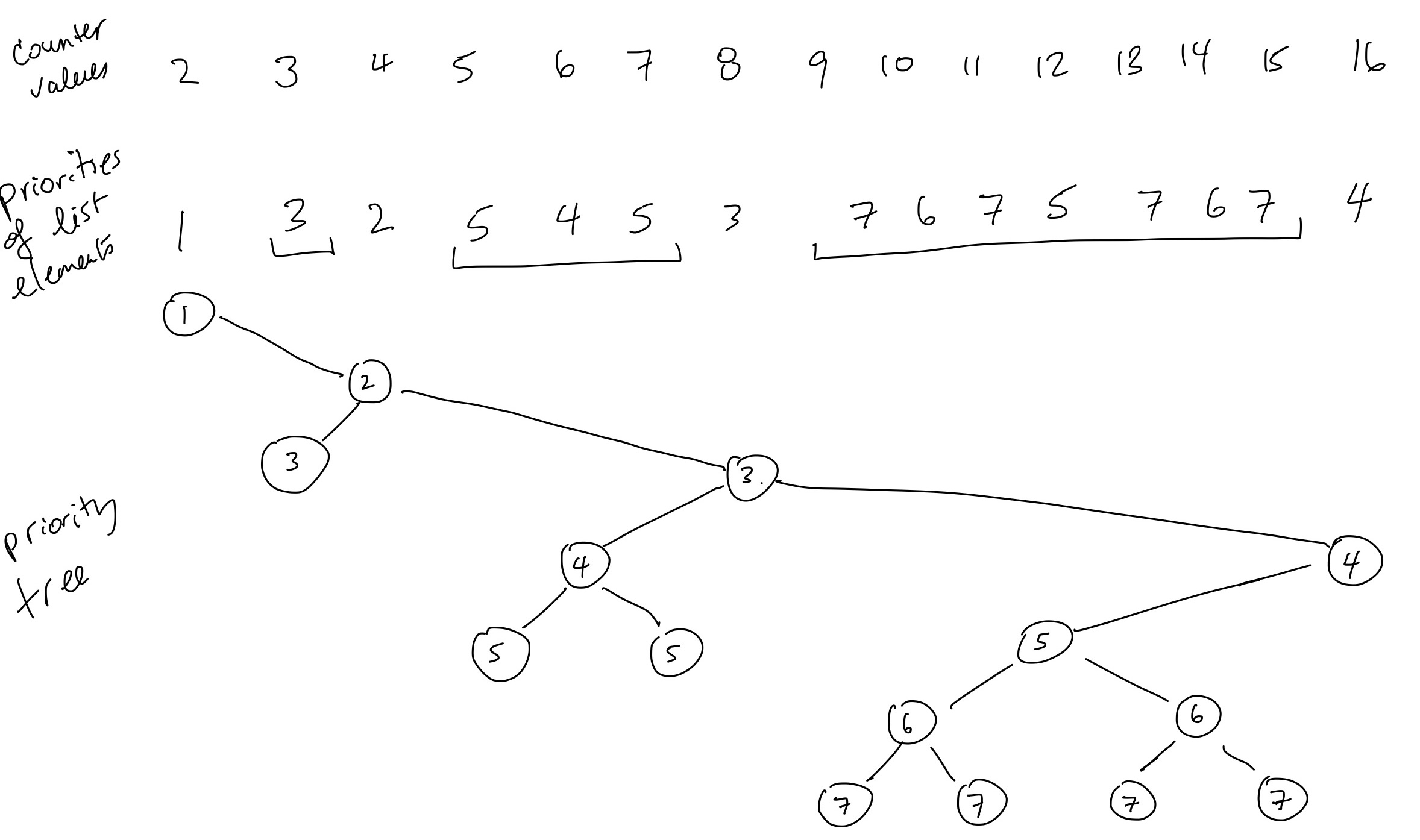
\caption{A list and its priority tree.
}
\label{priority-tree}
\end{figure}

In this scheme, an unmarked node can block  its ancestors in the priority tree
from being spliced out of the list.
For example, in Figure \ref{priority-tree}, if the nodes with counter values 10 to 16
are all marked for deletion, nodes 11, 13 and 15 could be spliced out immediately.  After
13 and 15 are spliced, node 14 could be too.  The unmarked node 9
prevents the remaining nodes 10, 12 and 16 from being spliced,
since each has a neighbour with higher priority.
Thus, an unmarked node  could prevent up to
$\Theta(\log L)$ marked nodes from being spliced out of the list.

Improving this space overhead factor to $O(1)$ requires an additional, novel mechanism.
If an attempt to remove node \B\ observes that  \B's left neighbour \A\ is unmarked
and \B's priority is greater than \B's right neighbour \C's priority, we allow \B\ to be spliced
out of the list using a special-purpose routine called \sul, even if \A's priority is greater than \B's.
In the example of the previous paragraph, this would allow node 10 to be spliced out
after~11.  Then, node 12 can be spliced out after 10 and 14, \er{again using \sul}, and finally node 16 can be spliced out.
A symmetric routine \sur\ applies if \C\ is unmarked and \B's priority is greater than \A's.
This additional mechanism of splicing out nodes when one neighbour is unmarked
allows us to splice out all nodes in a string of consecutive marked nodes,
except possibly one of them,
which might remain in the list if both its neighbours are unmarked and have higher priority.
However, 
during the \sul\ routine that is splicing out~\B, \A\ could become marked.  If \A's priority
is greater than its two neighbours' priorities, there could then be simultaneous splices
of \A\ and \B.
To avoid this, 
instead of splicing out \B\ directly, the \sul\ installs a
pointer to a {\it Descriptor} object into node \A,
which describes the splice of \B.  If \A\ becomes marked, the information in the
Descriptor is used to {\it help} complete the splice of \B\ before \A\ itself is spliced.
\er{Symmetrically, a \sur\ of \B\ installs a Descriptor in~\C.}

Multiple processes may attempt to splice the same node \B, either because of the helping
coordinated by Descriptor objects or because the process that spliced \B's last descendant
in the priority tree will also try to splice \B\ itself.
To avoid unnecessary 
work,
processes use a CAS to 
change the status of \B\ 
from \co{marked} to \co{finalized}.  Only the process that succeeds in this CAS
has the responsibility to recursively  splice \B's ancestors.
(In the case of the \sul\ and \sur\ routines, only the process that successfully
installs the Descriptor  recurses.)
If one process responsible for removing a node (and its ancestors) 
stalls,
it could leave $O(\log L)$ marked nodes in the list; this is the source of \y{an} {\it additive}
$P\log L$ term in \y{the bound we prove} on the number of unnecessary nodes in the list.

\begin{figure*}[!thp]\small
	\hspace{-10mm}
	\begin{minipage}[t]{.55\textwidth}
		\begin{lstlisting}[linewidth=.99\columnwidth, numbers=left,frame=none]
class Node {
	Node *left, *right; // initially null
	enum status {unmarked,marked,finalized};
			// initially unmarked
	int counter; // used to define priority
	int priority; // defines implicit tree
	int ts; // timestamp
	Descriptor *leftDesc, *rightDesc;
			// initially null
};

class Descriptor { Node *A, *B, *C; };
Descriptor* frozen = new Descriptor();

class VersionList {
	Node* Head;
	// public member functions:
	Node* getHead() {return Head;}
	
	Node* find(Node* start, int ts) {
		VNode* cur = start;                   @\label{find-init}@
		while(cur != null && cur->ts > ts) @\label{find-read}@
			cur = cur->left;                @\label{find-advance}@
		return cur; }
	
	bool tryAppend(Node* B, Node* C) {
		// B can be null iff C is the initial node
		if(B != null) {
			C->counter = B->counter+1;
			Node* A = B->left;							@\label{mr-read-left}@
			// Help tryAppend(A, B)
			if(A != null) CAS(&(A->right), null, B); @\label{mr-cas1}@
		} else C->counter = 2;
		C->priority = p(C->counter);					@\label{set-priority}@
		C->left = B; 									@\label{set-left}@
		if(CAS(&Head, B, C)) {							@\label{head-cas}@
			if(B != null) CAS(&(B->right), null, C); @\label{mr-cas2}@
			return true;
		} else return false; }
	
	// public static functions:
	void remove(Node* B) {
		// B cannot be null
		B->status = marked;					@\label{mark}@
		for F in [leftDesc, rightDesc] {
			repeat twice {					@\label{repeat-twice}@
				Descriptor* desc = B->F;		@\label{read-desc1}@
				help(desc);						@\label{remove-help}@
				CAS(&(B->F), desc, frozen); } } @\label{freeze-desc}@
		removeRec(B); }						@\label{call-rr1}@
	
	// private helper functions:
	bool validAndFrozen(Node* D) {
		// rightDesc is frozen second
		return D != null && D->rightDesc == frozen; }
	
	void help(Descriptor* desc) {
		if(desc != null && desc != frozen) 			@\label{test-frozen}@
			splice(desc->A, desc->B, desc->C); }	@\label{help-splice}@
	
	int p(int c) {
		k = floor(log@$_2$@(c));
		if(c == 2^k) return k;
		else return 2k + 1 - lowestSetBit(c); }
		\end{lstlisting}
\end{minipage}\hspace{4mm}
\begin{minipage}[t]{.62\textwidth}
	\StartLineAt{65}
	\begin{lstlisting}[linewidth=.99\textwidth, numbers=left, frame=none]
// private helper functions continued:
void removeRec(Node* B) {
	// B cannot be null
	Node* A = B->left;							@\label{read-left}@
	Node* C = B->right;							@\label{read-right}@
	if(B->status == finalized) return;          @\label{test-unfinalized}@
	int a, b, c;
	if(A != null) a = A->priority;
	else a = 0;
	if(C != null) c = C->priority;
	else c = 0;
	b = B->priority;
	if(a < b > c) { @\label{test-priority}@
		if(splice(A, B, C)) {               @\label{call-splice}@
			if(validAndFrozen(A)) {
				if(validAndFrozen(C) && c > a) removeRec(C);		@\label{recurse1}@
				else removeRec(A); }	@\label{recurse2}@
			else if(validAndFrozen(C)) {
				if(validAndFrozen(A) && a > c) removeRec(A);		@\label{recurse3}@
				else removeRec(C); } } }		@\label{recurse4}@
	else if(a > b > c) {
		if(spliceUnmarkedLeft(A, B, C) && @\label{call-sul}@
		validAndFrozen(C)) {
			removeRec(C); } }				@\label{recurse5}@
	else if(a < b < c) {
		if(spliceUnmarkedRight(A, B, C) && @\label{call-sur}@
		validAndFrozen(A)) {
			removeRec(A); } } } }			@\label{recurse6}@

bool splice(Node* A, Node* B, Node* C) {
	// B cannot be null
	if(A != null && A->right != B) return false;							@\label{test-right}@
	bool result = CAS(&(B->status), marked, finalized);    @\label{finalize-cas}@
	if(C != null) CAS(&(C->left), B, A);    								@\label{left-cas}@
	if(A != null) CAS(&(A->right), B, C);   								@\label{right-cas}@
	return result; }

bool spliceUnmarkedLeft(Node* A, Node* B, Node* C) {
	// A, B cannot be null
	Descriptor* oldDesc = A->rightDesc;			  @\label{read-desc2}@
	if(A->status != unmarked) return false;      @\label{sul-test-marked}@
	help(oldDesc);								@\label{sul-help-old}@
	if(A->right != B) return false;				 @\label{sul-test-right}@
	Descriptor* newDesc = new Descriptor(A, B, C); @\label{create-right-desc}@
	if(CAS(&(A->rightDesc), oldDesc, newDesc)) {   @\label{store-right-desc}@
		// oldDesc != frozen
		help(newDesc);								 @\label{sul-help-new}@
		return true;
	} else return false; }

bool spliceUnmarkedRight(Node* A, Node* B, Node* C) {
	// B, C cannot be null
	Descriptor* oldDesc = C->leftDesc;			@\label{read-desc3}@
	if(C->status != unmarked) return false;     @\label{sur-test-marked}@
	help(oldDesc);								@\label{sur-help-old}@
	if(C->left != B || (A != null && A->right != B)) 	@\label{sur-test} @
		return false;
	Descriptor* newDesc = new Descriptor(A, B, C); @\label{create-left-desc}@
	if(CAS(&(C->leftDesc), oldDesc, newDesc)) {	 @\label{store-left-desc}@
		// oldDesc != frozen
		help(newDesc);								 @\label{sur-help-new}@
		return true;
	} else return false; } };
	\end{lstlisting}
\end{minipage}
\vspace{-.1in}
\caption{Linearizable implementation of our doubly-linked list.
}
\label{fig:list-alg-full}
\end{figure*}

We now look at the code in more detail. 
Each node \X\ in the doubly-linked list has \co{right} and \co{left} pointers
\er{that point toward the list's head and away from it, respectively}.
\X\ also has a \co{status} field that is initially \co{unmarked} and \co{leftDesc} and \co{rightDesc} fields to hold pointers to Descriptors
for splices happening to the left and to the right of \X, respectively.
\X's \co{counter} field is filled in  when \X\ is appended to the right end of the list
with a value that is one greater than the preceding node. 
To ensure that the height of the priority tree is $O(\log L)$, we use the \co{counter} value $c$
to define the \co{priority} of \X\ as $p(c)$, where
$p(c)$ is either $k$ if $c$ is of the form~$2^k$, \linebreak
or $2k+1-(\mbox{number of consecutive 0's at the right end of the binary representation of }c)$, if 
$2^k<c<2^{k+1}$.
The resulting priority tree has a sequence of nodes with priorities $1, 2, 3, \ldots$
along the rightmost path in the tree, where the left subtree of the $i$th node along this rightmost path is a complete binary tree of height $i-1$,
as illustrated in Figure \ref{priority-tree}.
(Trees of this shape have been used to describe 
search trees \cite{BH76} and in concurrent data structures \cite{AAC09,AF01}.)
This assignment of priorities ensures that between any two nodes with the same
priority, there is another node with lower priority.
Moreover, the depth of a node with counter value $c$ is 
$O(\log L)$.
This construction also ensures that \co{remove} operations are
wait-free, since the priority of a node is a bound on the number of recursive calls that
a \co{remove} 
performs.

A Descriptor of a splice of node \B\ out from between 
\A\ and \C\ is 
an object that stores pointers to the three nodes \A, \B\ and \C.
After \B\ is marked, we set its Descriptor pointers to a special  Descriptor 
\co{frozen}
to indicate that  no further updates should occur on them.

To append a new node \C\ after the head node \B, the \co{tryAppend(B,C)} operation simply fills in the fields of \C, and then attempts to swing the \co{Head} pointer to \C\ at line \ref{head-cas}.
\B's \co{right} pointer is then updated at line \ref{mr-cas2}.
If the \co{tryAppend} stalls before executing line \ref{mr-cas2}, any attempt to append
another  node after \C\ will first help complete the append of \C\ (line \ref{mr-cas1}).
The boolean value returned by \co{tryAppend} indicates whether the append was successful.

A \co{remove(B)} first
sets \B's \co{status} to \co{marked} at line \ref{mark}.
It then stores the \co{frozen} Descriptor in both \co{B->leftDesc} and \co{B->rightDesc}.
The first attempt to store \co{frozen} in one of these fields may fail,
but we prove that the second will succeed because of some handshaking, described below.
\B\ is {\it frozen} once \co{frozen} is stored in both of its Descriptor fields.
Finally, \co{remove(B)} calls \co{removeRec(B)} 
to attempt
the real work of splicing \B.

The \co{removeRec(B)} routine
manages the recursive splicing of nodes.
It first calls \co{splice}, \sul\ or \sur, as appropriate,
to splice \B.  If the splice of \B\ was successful, it then recurses (if needed)
on the neighbour of \B\ with the larger priority. 

The actual updates to pointers are done inside the
\co{splice(A,B,C)} routine, which is called after reading \A\ in \co{B->left} and \C\ in \co{B->right}.
The routine first tests that $\co{A->right} = \B$ at line \ref{test-right}.
This could fail for two reasons:  \B\ has already been spliced out,
so there is no need to proceed,
or there is a \co{splice(A,D,B)} that has been partially completed;  \co{B->left} has 
been updated to \A, but \co{A->right} has not yet been updated to \B.
In the latter case, the \co{remove} that is splicing out \D\ will also splice \B\
after 
\D, 
so again there is no need to proceed with the splice of \B.
If  $\co{A->right} = \B$, \B's \co{status} is updated to \co{finalized} at line \ref{finalize-cas},
and the pointers in \C\ and \A\ are updated to splice \B\ out of the list at line \ref{left-cas} and \ref{right-cas}.

The \co{\sul(A,B,C)} handles the splicing of a node \B\ when \B's left neighbour \A\ has higher priority
but is unmarked,
and \B's right neighbour \C\ has lower priority.
The operation attempts to CAS a Descriptor of the splice into \co{A->rightDesc} at line \ref{store-right-desc}.
If there was already an old Descriptor there, it is first helped to complete at line \ref{sul-help-old}.
If the new Descriptor is successfully installed,
the \co{help} routine is called at line \ref{sul-help-new},
which in turn calls \co{splice} to complete the splicing out of \B.
The \sul\ operation can fail in several ways.
First, it can observe that \A\ has become marked, in which case \A\ should be spliced out before
\B\ since \A\ has higher priority.
(This test is also a kind of handshaking:  once a node is marked, at most one more Descriptor can be installed in it, and this ensures that one of the two attempts to install \co{frozen} in a node's Descriptor field during the \co{remove} routine succeeds.)
Second, it can observe at line \ref{sul-test-right} that $\co{A->right} \neq \B$.
As described above for the \co{splice} routine, it is safe to abort the splice in this case.
Finally, the CAS at line \ref{store-right-desc} can fail, either because
\co{A->rightDesc} has been changed to \co{frozen} (indicating that \A\ should be spliced before \B)
or another process has already stored a new Descriptor in \co{A->rightDesc} (indicating either
that \B\ has already been spliced or will be by another process).

The \sur\ routine is symmetric to \sul, aside from a slight difference in line \ref{sur-test} because \co{splice} changes the \co{left} pointer before the \co{right}~pointer.
The return values of \co{splice}, \sul\ and \sur\ 
say
whether 
the calling process should continue recursing up the priority tree to splice out more nodes.

\ignore{
\Eric{The finalize-cas may be unnecessary; processes could instead return false if the left pointer cas failed.  This would require an overhaul of the proof.  Instead of talking about finalized nodes \X,
the notion would be nodes where $\co{X->right->left} \neq \X$.
Also, at some point, we should check whether it is really necessary to help the oldDesc in spliceUnmarkedLeft/Right.}

\Eric{Might be possible to simplify lines \ref{recurse3}--\ref{recurse4} to a simple call to removeRec(C) but that may complicate the progress proof a little, because with current version, we know that if we recurse to the higher of the two neighbours, the lower neighbour is not frozen at the time we recurse.}
}


\Youla{I believe that the above part of this section can become shorter by removing several sentences here and there.
However, after some point, I left this for later on.}

\subsection{Properties of the Implementation}

Detailed proofs of the following  results appear in the full version \cite{arxiv}.  We sketch them here.
\begin{restatable}{theorem}{listcorrectness}
\label{list-correctness}
Under Assumption \ref{list-assumptions}, the implementation in Figure \ref{fig:list-alg-full} is linearizable.
\end{restatable}
Since the implementation 
is fairly complex, the correctness proof
is necessarily quite intricate.
We say that $\X <_c \Y$ if node \X\ is appended to the list before node~\Y.
We prove that \co{left} and \co{right} pointers in the list always respect this ordering.
Removing a node has several key steps:  marking it (line \ref{mark}),
freezing it 
(second iteration of line \ref{freeze-desc}),
finalizing it 
(successful CAS at line \ref{finalize-cas}) 
and then making it
unreachable 
(successful CAS at line \ref{right-cas}).
We prove several lemmas showing that these steps take place in an orderly way.
We also show that the steps make progress.  
Finally, we show that the coordination between \co{remove} operations
guarantees that the structure of the list remains a doubly-linked list in which nodes are
ordered by $<_c$, except
for a temporary situation while a node is being spliced out, \er{during which its left neighbour
may still point to it after its right neighbour's pointer has been updated to skip past it.}
To facilitate the inductive proof of this invariant, it is wrapped up with several others,
including an assertion that overlapping calls to \co{splice} of the form \co{splice(W,X,Y)} and \co{splice(X,Y,Z)} never occur.
The invariant also asserts that  unmarked
nodes remain in the doubly-linked list; no \co{left} or \co{right} pointer can
jump past a node that has not been finalized.
Together with Assumption \ref{list-assumptions}.\ref{removals-safe-inf}, this
ensures a \co{find} cannot miss the node that it is supposed to return, \er{regardless of how \co{find} and \co{remove} operations are linearized.}
We linearize \co{getHead} and \co{tryAppend}  when they access the
\co{Head} pointer.

\begin{theorem}
\label{remove-wait-free}
The number of steps a \co{remove(X)} operation performs is $O(\co{X->priority})$ and the
\co{remove} operation is therefore wait-free.
\end{theorem}
\begin{proof}
Aside from the call to \co{removeRec(X)},  \co{remove(X)} performs $O(1)$ steps.
Aside from doing at most one recursive call to \co{removeRec}, a \co{removeRec} operation
performs $O(1)$ steps.
Each time \co{removeRec} is called recursively, the node on which it is called has a smaller priority.
Since priorities are non-negative integers, the claim follows.
\end{proof}
\begin{restatable}{theorem}{removeAmortized}
\label{remove-amortized-time}
The \co{tryAppend} and \co{getHead} operations take $O(1)$ steps.  The amortized number of steps for  \co{remove} is $O(1)$.
\end{restatable}
Consider an execution with $R$ \co{remove} operations. Using the argument for Theorem \ref{remove-wait-free}, it suffices to bound the number of calls to \co{removeRec}.
There are at most $R$ calls to \co{removeRec} directly from \co{remove}.
For each of the $R$ nodes \X\ that are removed, we show that at most one call to \co{removeRec(X)}
succeeds either in finalizing \X\ 
or installing a Descriptor to remove \X, 
and only this \co{removeRec(X)} can
call \co{removeRec} recursively.

We say a node is {\it\lrreachable} if it is reachable from the head of the list by following \co{left} or \co{right} pointers.  A node is {\it \lrunreachable} if it is not \lrreachable.
%
%
\begin{restatable}{theorem}{removeSpace}
\label{space-bound}
At the end of any execution by $P$ processes that contains $L$ successful \co{tryAppend}
operations and $R$ \co{remove} operations \hedit{on a set of version lists, and a maximum of $L_{max}$ successful \co{tryAppend}s on a single version list}, the total number of lr-reachable nodes across all the version lists \hedit{in the set}
is at most $2(L-R)+O(P\log L_{max})$. 
\Youla{The phrase "across all version lists" is confusing. Until now we were talking about the implementation of a version list. 
We probably need to explain why we refer to many version lists here.}
\Hao{Added a sentence to address Youla's previous comment.}
\end{restatable}

\hedit{Theorem \ref{space-bound} considers a set of version lists to indicate that the $O(P \log L_{max})$ additive space overhead is shared across all the version lists in the system.}
A node \X\ is {\it removable} if \co{remove(X)} has been invoked.
We must show  at most $(L-R)+O(P\log L_{max})$ removable nodes are still lr-reachable.
We count the number of nodes that are in each of the various phases (freezing, finalizing, making unreachable) of the removal.
There are at most $P$ removable nodes that are not yet frozen, since each has a pending \co{remove} operation on it.
There are at most $P$ finalized nodes that are still lr-reachable, since each has a pending \co{splice}
operation on it.
To bound the number of nodes that are frozen but not finalized,
we classify an unfinalized node as Type 0, 1, or 2, depending on the number of its subtrees
that contain an unfinalized node.
We show that each frozen, unfinalized node \X\ of type 0 or 1 has a pending
\co{remove} or \co{removeRec} at one of its descendants.
So, there are $O(P\log L_{max})$ such nodes.
We show that at most half of the unfinalized nodes are of type 2, so there are at most $L-R+O(P\log L_{max})$ type-2 
nodes.
Summing up yields the 
bound.

\future{Disentangle two meanings of L in following theorem:  total length of all lists in L-R term and length of one list in log L term.}


\section{Memory Reclamation for Version Lists}
\label{sec:smr-list}

We now describe how to safely reclaim the nodes 
spliced out of version lists and
\y{the} Descriptor objects that are no longer needed.
We apply an implementation of
Reference Counting (RC)~\cite{ABW21} with amortized expected $O(1)$ time overhead to a slightly modified version of our list.
\y{To apply RC in Figure~\ref{fig:list-alg-full},} we add a reference count field to each Node or Descriptor
and replace raw pointers to \y{Node}s or Descriptors with reference-counted pointers.
Reclaiming an object clears all \y{its reference-counted pointers,} 
which may lead to recursive reclamations if any reference count hits zero.
This reclamation scheme is simple, but
not sufficient by itself because 
a single pointer to a spliced out node may prevent a long chain of spliced out nodes from being reclaimed
(see Figure~\ref{fig:up-pointers}, discussed later).
To avoid this, we modify the \co{splice} routine so that 
whenever the \co{left} or \co{right} pointer of an node
\co{Y} points to a descendant in the implicit tree, we set the pointer to~$\top$
after \co{Y} is spliced out. 
Thus, only \y{\co{left} and \co{right}} pointers from spliced out nodes to their ancestors
in the implicit tree 
remain \y{valid}.
This ensures that there are only $O(\log L)$ spliced out nodes reachable from any 
spliced out node. 

\begin{figure}[t]
\centering{	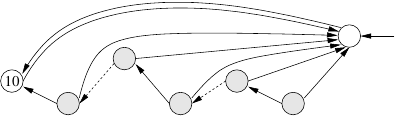}
	\caption{\hedit{A portion of a version list where shaded nodes 15, 14, ..., 11 have been removed, in that order. Dotted pointers represent left and right pointers set to $\top$ by our modified \op{splice} routine. Node labels are counter values and vertical positioning represents nodes' priorities (cf.\ Figure \ref{priority-tree}).}}
	\label{fig:up-pointers}
\end{figure}

This modification requires some changes to \co{find}.
When a \op{find} reaches a node whose
\co{left} pointer is $\top$, the traversal moves right instead;
this results in following a valid pointer because
whenever \op{splice(A, B, C)} is called, it is guaranteed that either \A\ or \C\ is an ancestor of \B.
For example in Figure \ref{fig:up-pointers}, a process $p_1$, paused on node 15, will next traverse nodes 14, 16, and 10.
Breaking up chains of removed nodes (e.g., from node 15 to 11 in Figure~\ref{fig:up-pointers})
by setting some pointers to $\top$ is important because otherwise, such chains can become arbitrarily long and 
a process paused at the head of a chain can prevent all of its nodes from being reclaimed.
In the full version of the paper, we prove that traversing backwards does not have any significant impact on the time complexity of \op{find}.
Intuitively, this is because backwards traversals only happen when the \op{find} is poised to read a node that has already been spliced out and each backwards traversal brings it closer to a non-removed node.
\Eric{Previous sentence doesn't sound convincing.  Isn't there more to it than this?  Should we just say "we prove that ..." in order to indicate that it is not meant to be obvious.}
\Hao{Applied Eric's suggestion}

Using the memory reclamation scheme described above, 
 we prove Theorems~\ref{thm:step3-time} and~\ref{thm:step3-space}  
that provide bounds similar to Theorems \ref{remove-amortized-time} and~\ref{space-bound} in~\cite{arxiv}. 
\hedit{Both theorems include the resources needed by the RC algorithm, 
such as incrementing reference counts, maintaining retired lists, etc.}
\hedit{Since the RC algorithm uses process-local hash tables, 
the amortized time bounds in Theorem~\ref{remove-amortized-time} become amortized \emph{in expectation} in Theorem~\ref{thm:step3-time}.}
Using this scheme requires that \op{getHead} and \op{find} return reference counted pointers rather than raw pointers.
Holding on to these reference counted pointers prevents the nodes that they point to from being reclaimed. 
For the space bounds in Theorem~\ref{thm:step3-space}, we consider the number of reference counted pointers $K$, 
returned by version list operations that are still used by the application code.
In most multiversioning systems (including the one in Appendix~\ref{sec:application}), 
each process holds on to a constant number of such pointers, so $K \in O(P)$. 

\begin{restatable}{theorem}{memoryReclamationTime}
	\label{thm:step3-time}
	The amortized expected time complexity of \op{tryAppend}, \op{getHead}, \op{remove}, and creating a new version list is $O(1)$.
	The amortized expected time complexity of \op{find(V, ts)} is $O(n + \min(d, \log c))$, where $n$ is the number of version nodes with timestamp greater than $\ts$ that are reachable from \op{V} by following \co{left} pointers (measured at the start of the \op{find}), $d$ is the depth of the \vnode{} \op{V} in the implicit tree and $c$ is the number of successful \op{tryAppend} from the time \V\ was the list head until the end of the \op{find}.
	All  operations are wait-free.
\end{restatable}
\vspace{-.1in}
\begin{restatable}{theorem}{memoryReclamationSpace}
	\label{thm:step3-space}
	Assuming there are at most $K$ reference-counted pointers to \vnode{}s from the application code, at the end of any execution that contains $L$ successful \op{tryAppend} operations, $R$ \op{remove} operations and a maximum of $L_{max}$ successful \op{tryAppend}s on a single version list, the number of \vnode{}s and Descriptors that have been allocated but not reclaimed is $O((L-R) + (P^2 + K) \log L_{max})$.
\end{restatable}

In RC, cycles must be broken before a node can be reclaimed.
While there are cycles in our
version lists, we show that \vnode s that have been spliced out are not part of
any cycle.

\section{Application to Snapshottable Data Structures}
\label{sec:application}


We present a summary of the multiversioning scheme of Wei et al.~\cite{WBBFRS21a}, 
and describe how the techniques in this paper can be applied to achieve good complexity bounds.


\myparagraph{The Multiversioning Scheme.}
Wei et al.~\cite{WBBFRS21a} apply multiversioning to
a concurrent data structure (DS) implemented from CAS objects
to make it {\em snapshottable}. It does so
by replacing each CAS object by a VersionedCAS object which
stores a version list of all earlier values of the object.
VersionedCAS objects support \co{vRead} and \co{vCAS} operations, which behave
like ordinary read and CAS. 
They also support a \co{readVersion} operation which can be used to
read earlier values of the object.
\hedit{Wei et al. present an optimization for avoiding the level of indirection introduced by version lists. For simplicity, we apply our MVGC technique to the version without this optimization.}

Wei et al.\ also introduce a \emph{Camera} object
which is associated with these \emph{VersionedCAS} objects.
The Camera object simply stores a timestamp.
A \co{takeSnapshot} operation applied to the Camera object
attempts to increment the timestamp and returns the old value of the timestamp as a snapshot handle.
To support read-only query operations on the concurrent DS
(such as range-query, successor, filter, etc.), it suffices
to obtain a snapshot handle $s$, and then read the relevant objects
in the DS using $\co{readVersion}(s)$ to get their values
at the linearization point of the \co{takeSnapshot} that returned $s$. 
This approach can be used to add arbitrary queries to many standard data structures.

For multiversion garbage collection,
Wei et al.~\cite{WBBFRS21a} uses a variation of EBR~\cite{Fra04}, inheriting its drawbacks.
Applying our range-tracking and version-list data structures significantly reduces space usage,
resulting in bounded space without sacrificing time complexity.

\myparagraph{Applying Our \hedit{MVGC} Scheme.}
Operations on snapshottable data structures (obtained by applying the technique in~\cite{WBBFRS21a})
are divided into \emph{snapshot queries},
which use a snapshot handle to answer queries,
and \emph{frontier operations},
which are inherited from the original non-snapshottable DS.
We use our doubly-linked list algorithm (with the memory reclamation scheme from Section~\ref{sec:smr-list}) for each VersionedCAS object's version list, and a range-tracking object \op{rt} to announce timestamps and keep track of required versions by ongoing snapshot queries.
We  distinguish between objects inherited from the original DS (\dnode{}s) and version list nodes (\vnode{}s).
\Eric{This sounds a bit weird now that you renamed dnodes to dobjects}
\Hao{I changed it back to dnodes to address Eric's comment and save space.}
For example, if the original DS is a search tree, the \dnode{}s would be the nodes of the search tree.
See~\cite{arxiv} for the enhanced code of~\cite{WBBFRS21a} with our MVGC scheme.

At the beginning of each snapshot query, the taken snapshot is announced using \op{rt.announce()}.
At the end of the query, \op{rt.unannounce()} is called to indicate that the snapshot that it reserved is no longer needed.
Whenever a \op{vCAS} operation adds a new \vnode{} \C{} to the head of a version list,
we \op{deprecate} the previous head VNode \B{} by calling \op{rt.deprecate(B, B.timestamp, C.timestamp)}.
Our announcement scheme prevents \vnode s that are part of any ongoing snapshot from being returned by \op{deprecate}.


Once a \vnode{} is returned by a \op{deprecate}, it is removed from its version list and the reclamation of this \vnode{} and the Descriptors that it points to
is handled automatically by the reference-counting scheme of Section~\ref{sec:smr-list}.
Thus, we turn our attention to \dnode{}s.
A \dnode{} can be reclaimed when neither frontier operations nor snapshot queries can access it.

We assume that the original, non-snapshottable DS comes with a memory reclamation scheme, \MRS,
which we use to determine if a \dnode{} is needed by any frontier operation.
We assume that this scheme
calls \op{retire} on a node $X$ when it becomes unreachable from the roots of the DS,
and \op{free} on $X$ when no frontier operations need it any longer.
This assumption is naturally satisfied by many well-known reclamation schemes (e.g.,~\cite{HL+05,RC17,Fra04}). 

Even when \MRS\ \op{free}s a \dnode{},
it may not be safe to reclaim it, as it may still be needed by ongoing snapshot queries.
To solve this problem, we tag each \dnode{} with a birth timestamp and a retire timestamp.
A \dnode's birth timestamp is set after a \dnode\ is allocated but before it is attached to the data structure.
Similarly, a \dnode's retire timestamp is set when \MRS\ calls \op{retire} on it.
We say that a \dnode{} is \emph{necessary} if it is not yet freed by  \MRS, 
or if there exists an announced timestamp in between its birth and retire timestamp.
We track this using the same \rangetracker{} data structure \op{rt} that was used for \vnode{}s.
Whenever \MRS\ \op{frees} a \dnode{} \op{N},
we instead call \op{rt.deprecate(N, N.birthTS, N.retireTS)}.
When a \dnode{} gets returned by a \op{deprecate}, it is no longer needed so we 
\hedit{reclaim its storage space}.

We say that a \vnode{} is \emph{necessary} if it is pointed to by a \dnode{} that has not yet been deprecated (i.e. freed by \MRS) or if its interval contains an announced timestamp.
Let $D$ and $V$ be the maximum, over all configurations in the execution, of the number of necessary \dnode{}s and \vnode{}s, respectively.
Theorem~\ref{thm:overall-space} bounds the overall memory usage of our memory-managed snapshottable data structure. 
Theorem~\ref{thm:overall-time} is an amortized version of the time bounds proven in~\cite{WBBFRS21a}.


\begin{restatable}{theorem}{applicationSpace}
\label{thm:overall-space}
	Assuming each \vnode{} and \dnode{} takes $O(1)$ space, the overall space usage of our memory-managed snapshottable data structure is $O(D + V + P^2 \log P + P^2 \log L_{max})$, where $L_{max}$ is the maximum number of successful \op{vCAS} operations on a single \op{VCAS} object.
\end{restatable}

%

\begin{restatable}{theorem}{applicationTime}
\label{thm:overall-time}
	A snapshot query takes amortized expected time proportional to its sequential complexity plus the number of \op{vCAS} instructions concurrent with it.
	The amortized expected time complexity of frontier operations is the same as in the non-snapshottable DS.
\end{restatable}







\newpage

\bibliographystyle{abbrv}
\bibliography{gc}

\appendix
 

\section{Model}
\label{app:model}

We consider a standard concurrent system where $P$ {\em processes}
execute asynchronously and may fail by crashing.
The system provides a set of {\em base objects},  which
have a state and support
atomic operations to read or update their state, such as
\op{read}, \op{write}, and/or \op{Compare\&Swap} (\op{CAS}).
An {\em implementation} of a data structure (also known as {\em concurrent data structure})
provides, for each process $p$, an algorithm for implementing the operations supported by the data structure
(using base objects).
More complex objects (or data structures) can be implemented using the base objects.

We follow a standard formalism to model the system.
A {\em configuration} is a vector that represents an instantaneous snapshot of the system
(i.e., it provides the state of each process and the value stored in each base object) at some point in time.
In an {\em initial configuration}, all processes are in an initial state
and each object stores an initial value.
A {\em step} by a process $p$ is comprised of a single operation applied by $p$ on a base object
(and may also contain some local computation performed by $p$).
The sequence of steps that each process performs
depends on its program.
An {\em execution} is an alternating sequence of configurations and steps,
starting from an initial configuration.

Fix any execution $E$.
The {\em execution interval} of an operation starts with
the step that invokes the operation and finishes with
the step performing its last instuction.
We use a standard definition of linearizability:
$E$ is linearizable if there is a sequential execution $\sigma_E$,
which contains all complete operations in $E$ (and some of the incomplete ones)
and the following hold:
$\sigma_E$ respects the real-time ordering imposed by the execution intervals of operations
in $E$ and each operation in $E$ that appears in $\sigma_E$ returns
the same response in both executions.

An implementation is {\em wait-free} is each process completes the execution of
every operation it invokes within a finite number of steps.
It is {\em lock-free} if in every infinite execution, an infinite number
of operations complete; thus, a lock-free implementation ensures that the system as whole makes progress,
but individual processes may starve.
The {\em step complexity}
of an operation is the maximum number of steps, over all executions, that a thread performs to
complete any instance of the operation; we sometimes use
the term {\em time complexity} (or simply {\em time}) to refer to
the step complexity of an operation.
\er{Saying that the expected amortized number of steps for an operation is $O(f(n))$ means
that the expected total number of steps performed by $m$ invocations of the operation
is $O(m\cdot f(n))$.}

\ignore{
\Youla{Do we need also something like the following? I was thiknig that this couls be c or better remove it?}
In this paper, we focus on {\em memory reclamation} techniques for data structures
that support complex queries using multi-versioning (i.e., that can become {\em snapshotable}).
We call the original concurrent data structure on which multi-versioning is applied to support complex
queries {\em frontier} data structure, and all operations
that it supports, {\em frontier operations}.
}


\section{Correctness and Analysis of \RT}
\label{RT-proofs}

Here, we provide the proofs for Theorems \ref{thm:rt-lin}, \ref{step1-time} and \ref{step1-space} from Section \ref{sec:identify}.  
These show that \RT\ is linearizable and give  time and space bounds for it.

\subsection{Correctness of \RT}
\label{sec:rt-correctness}

In this section, we prove Theorem \ref{thm:rt-lin}, which says that \RT\ (Figure~\ref{fig:range-tracker}) is a linearizable implementation of the range-tracking object. 
Consider any execution $E$. We first prove that the pools maintained by the algorithm are disjoint and have the right size.


\begin{invariant}\hspace{1mm}
\label{inv:distinct}
\begin{enumerate}
\item
Among all the local \retired s and pools in \co{Q}, no object appears more than once
(with one temporary exception:  if a process is poised at line \ref{line:list-clear}, 
its \retired\ may contain elements that also appear in pools of \co{Q}). 
\item
Each \retired\ contains at most $B$ objects (with one temporary exception:  if a process has
executed line \ref{line:mergeLDP}, but not yet executed line \ref{line:list-clear}, then
its \retired\ may contain up to $2B$ objects).
\item
Each pool in \co{Q} contains between $B$ and $2B$ objects.  
\item
Each pool is sorted by \co{high}.
\end{enumerate}
\end{invariant}
\begin{proof}
\begin{enumerate}
\item
Assumption \ref{ass}.\ref{ass:distinct} ensures that each object added to a local \retired\ is
distinct from objects that are in other processes' \retired s or in one of \co{Q}'s pools.
The flush phase of a \retireInterval\ rearranges the objects of a few pools, and may remove
some of the objects to be returned, but does not add or duplicate any objects
(with the temporary exception that a copy of \retired\ is added to \co{Q} at line 
\ref{line:enq-retired} before \retired\ is emptied on the next line). 
\item
The test at line \ref{line:ifFreeNeeded} ensures that \retired\ is emptied (at line \ref{line:list-clear})
whenever it reaches size $B$.
\item
We show that the modifications to \co{Q} made by \retireInterval\ preserve this claim.
Since the two pools dequeued at line \ref{line:deq} have at most $2B$ objects each, \Needed\ has
at most $4B$ objects.  If \Needed\ has more than $2B$ objects, it is split into two pools whose sizes
are between $B$ and $2B$ before those pools are enqueued.  \retired\ is of size at least $B$ when it is enqueued at line \ref{line:enq-retired}.  If two pools are merged at line \ref{line:mergeLDP},
then each of the two is of size at most $B$, so the resulting size is at most $2B$.
\item
Assumption \ref{ass}.\ref{ass:ascending} ensures that elements are appended to \retired\ in sorted order.
All steps that split or merge pools maintain their sorted order.\qedhere
\end{enumerate}
\end{proof}

Recall that we assign linearization points to operations as follows.
\begin{itemize}
\item An \Announce\ is linearized at its \op{swcopy} on line~\ref{line:swcopy}.
\item An \Unannounce\ is linearized at its \op{write} on line~\ref{line:write-empty}.
\item A \retireInterval\ is linearized at line \ref{line:enq-retired} if it executes that line, or at line \ref{line:retint-lin} otherwise.  
\end{itemize}
Operations that do not execute the line specified above are not included in the linearization.
We say that an object \co{o} is {\em deprecated} at some configuration of an execution
if a \retireInterval\ of object \co{o} is linearized before that configuration.

\begin{invariant}
\label{inv:Q-elements}
All objects in the following pools are deprecated:
\begin{enumerate}
\item all pools in \co{Q},
\label{Q-deprecated}
\item 
\label{local-deprecated}
the \Needed\ and \Redundant\ pools of any \intersect\ or \retireInterval\ operation, and
\item 
\label{ldp-inv}
each process $p$'s \retired, except for object \co{o} if $p$ is executing \retireInterval\co{(o,low,high)}.
\end{enumerate}
\end{invariant}

\begin{proof}
\co{Q} and each process's \retired\ are initially empty and no calls to \retireInterval\ or \intersect\
	are initially active, so the claim holds trivially in the initial state.
	We proceed inductively by showing every step $s$ preserves the invariant.
Since \co{Q} is a linearizable  queue,
we treat the execution of \co{Q}.\op{enq()} and \co{Q}.\op{deq()} as atomic steps. 

We first show that if $s$ is a \co{Q}.\op{enq()} performed by a \retireInterval\ on some object \co{o}, it preserves
claim \ref{Q-deprecated}. 
The pools enqueued at line~\ref{line:enq1},~\ref{line:enq2} or~\ref{line:enq3}, 
contain objects stored in \Needed, which are deprecated by induction hypothesis \ref{local-deprecated}.
The pool enqueued at line~\ref{line:enq-retired}
results from merging \retired\ with \Needed. 
By the induction hypotheses \ref{local-deprecated} and \ref{ldp-inv}, all objects in \Needed\ and all objects in \retired\ other than \co{o} are deprecated.
Since $s$ is the linearization point of the \retireInterval\ on \co{o}, \co{o} is also deprecated in the configuration after $s$. 
So, claim \ref{Q-deprecated} holds after~$s$.

We next show that if $s$ causes a \retireInterval\ by process $p$ on object \co{o} to terminate,
it preserves claim \ref{ldp-inv}.
By induction hypothesis \ref{ldp-inv}, we must just show
that if \co{o} is in $p$'s \retired, then \co{o} is deprecated.
If \co{o} is in \retired, then line \ref{line:list-clear} was not executed, so the \retireInterval\ of \co{o} was
linearized at line \ref{line:retint-lin} before $s$.  Thus, \co{o} is deprecated,
by the definition of deprecated objects.
\Youla{This paragraph does not connect nicely with the rest of the proof.
Maybe move it after the next paragraph, as the next paragraph seems to be
also necessary to prove claim 3?}
\Eric{I did some rewriting to make structure of proof clearer.  Is it okay now?}

We next show that if $s$ is the assignment of the result of \merge\ into \retired\ on line~\ref{line:mergeLDP} of a \retireInterval\ on \co{o}, it preserves claim \ref{ldp-inv}. 
By induction hypotheses \ref{local-deprecated} and \ref{ldp-inv}, 
all objects  in \Needed\  and all objects
in \retired\ other than \co{o} are  deprecated. Since $s$ merges the objects
in \retired\ with those in \Needed\ and assigns the result to \retired, it follows that 
all nodes in \retired\ other than \co{o} are deprecated after $s$.

Finally, assume that $s$ is an execution of  line~\ref{line:RedAppend} or line~\ref{line:NeedAppend} of \co{intersect},
which appends an object drawn from \co{MQ} to \co{Redundant} or \co{Needed}.
We show that $s$ preserves claim \ref{local-deprecated}.
Before \co{intersect} was called at line \ref{line:intersect}, 
\co{MQ} was created on line \ref{line:deq} by merging two pools that were dequeued from \co{Q}.
By induction hypothesis  \ref{Q-deprecated}, the object $s$ appends to \co{Redundant} or \co{Needed} is deprecated. 
\end{proof}

To show that \retireInterval\ returns a correct response,
we first prove some lemmas.

\begin{lemma}
\label{intersect-works}
If \co{intersect(MQ, ar)} is called with sorted arguments, the \co{Redundant} and \co{Needed} sets
it returns partition the elements of \co{MQ}, where \co{Needed} contains exactly those elements whose
intervals intersect \co{ar}.
\end{lemma}
\begin{proof}
We consider the test on line \ref{line:intersect-test}.
If $\co{i}=0$, then by the exit condition of the while loop, we have 
$\co{r.high} \leq \co{ar[0]}\leq \co{ar[p]}$ for all \co{p}, and \co{r} is placed in \co{Redundant}.
Assume now that $\co{i}> 0$.
Since \co{i} was last incremented during an iteration of the for loop on object $\co{r}'$, 
we have $\co{ar[i-1]} < \co{r}'.\co{high} \leq \co{r.high}$.
Using the exit condition of the while loop, we therefore have $\co{ar[i-1]} < \co{r.high} \leq \co{ar[i]}$.
Since \co{ar} is sorted, \co{r}'s interval $[\co{r.low},\co{r.high})$ 
intersects $\co{ar}$ if and only if $\co{ar[i-1]} \geq \co{r.low}$ and this is exactly the
test used to determine whether \co{r} is added to \co{Needed} or \co{Redundant}.
\end{proof}

We say that an object \co{o} is \emph{redundant} at some configuration $G$ if it is deprecated at $G$
and its interval does not intersect \co{Ann}. 
If a deprecated object is not redundant at configuration~$G$, 
we say that it is \emph{needed} at $G$. 

The following observation is a direct consequence of Assumption \ref{ass}.\ref{ass:previously}.
\begin{observation}
\label{Ann-bound}
Any non-$\bot$ value written into \co{Ann} after a \co{\retireInterval(o,low,high)} is invoked
 is greater than or equal to \co{high}.
\end{observation}

\begin{lemma}
\label{lem:deprecateRetVal}
All objects returned by a \retireInterval\ are redundant in all configurations after the \retireInterval\ is linearized.  The size of the set returned by \retireInterval\ is at most $4P\log P$.
\end{lemma}

\begin{proof}
Consider any instance $I$ of \retireInterval.
If  $I$ returns the empty set, then the claim holds trivially.
So, assume $I$ performs a flush phase and returns a non-empty set \Redundant.
During the flush phase,
$I$ dequeues two pools from \co{Q} and puts their elements in \co{MQ}.
Then, $I$ calls \co{sortAnnouncements} at line \ref{line:scan} to make a local copy of the announcement array 
\co{Ann} in the array \co{ar}, and  calls \co{intersect} at line \ref{line:intersect} to compute
\Redundant.

Let \co{r} be an element of \Redundant\ and \co{p} be any process.
By Invariant~\ref{inv:Q-elements}, all elements of \co{MQ} are deprecated.
So, \co{r} is deprecated before \co{sortAnnouncements} is called.
By Observation \ref{Ann-bound}, no value written into \co{Announce[p]} after that will intersect
\co{r}'s interval.
By Lemma \ref{intersect-works}, \co{ar[p]} is not in \co{r}'s interval either.
Thus at all times after \co{sortAnnouncements} reads the value \co{ar[p]} from \co{Announce[p]}, 
the value in \co{Announce[p]} is not in \co{r}'s interval.
Since $I$ is linearized after this read for every \co{p}, the first claim follows.

All objects returned by the \retireInterval\ are drawn from two pools dequeued from \queue.
It follows from Invariant \ref{inv:distinct}, that at most $4P\log P$ objects are returned.
\end{proof}


The flush phase of a \retireInterval\ removes all the objects it returns from the pools maintained by \RT.
Thus, the sets of nodes returned by any pair of calls 
to \retireInterval\ are disjoint. 
Given the way linearization points are assigned, 
Invariants~\ref{inv:distinct} and \ref{inv:Q-elements}, and Lemma~\ref{lem:deprecateRetVal},
then imply the following theorem. 

\RTcorrect*
	
\subsection{Time and Space Bounds for \RT}
\label{sec:rt-bounds}

In this section, we prove Theorems \ref{step1-time} and \ref{step1-space}.
In particular, we show that, in any execution $E$, \retireInterval\ takes $O(P\log P)$ steps in the worst case, 
and the amortized number of steps per \retireInterval\ is constant. We also show that at any  time, 
at most $2H_{max} + 5P^2\log P$ redundant intervals have not yet been returned by any instance of \retireInterval,
where $H_{max}$ is the highest number of objects needed at any point in time during $E$; 
this will help us provide space bounds on the overall garbage collection mechanism we construct in this paper.

\RTtime*

\begin{proof}
	Recall that \op{swcopy}, \op{read}, and \op{write} on destination objects take a constant number of steps. 
Thus, \Announce\ and \Unannounce\ take a constant number of steps in the worst case. Aside from the flush phase,
a \retireInterval\ takes $O(1)$ steps.
By Invariant \ref{inv:distinct}, all pools have size $O(P\log P)$.
Enqueueing or dequeueing a pool from \co{Q} takes $O(P\log P)$ steps, including the time to copy elements in the pool.
Thus, as discussed in Section 4.1, the flush phase of a \retireInterval\ takes $O(P\log P)$ steps.
Since each process performs the flush phase once every $P\log P$ calls to 
\retireInterval, the amortized bound follows.
\end{proof}

We next bound 
the number of deprecated objects that have not yet been returned by any instance of \retireInterval.


\begin{lemma}
\label{remain-redundant}
If an object \co{o} is redundant at some configuration, then it remains
redundant in every following configuration.
\end{lemma}

\begin{proof}
Consider any object \co{o} that is redundant at some configuration $G$. 
Since \co{o} is redundant at $G$, the elements of \co{Ann} at $G$
do not intersect \co{o}'s interval. 
By Observation \ref{Ann-bound}, no value in \co{o}'s interval can be written into \co{Ann} after $G$.
It follows that \co{o} is redundant in every configuration after $G$. 
\end{proof}

\begin{lemma}
\label{return-redundant}
Consider any instance $I$ of \retireInterval\ that executes the flush phase. 
Let $G$ be the configuration just after $I$ dequeues a list $L$ from \co{Q} (on line~\ref{line:deq}). 
Then, $I$ returns all objects in $L$ that were redundant at $G$. 
\end{lemma}

\begin{proof}
Consider any object \co{o} in $L$ that is redundant at $G$.
By Lemma \ref{remain-redundant}, \co{o} is still redundant when $I$'s call to \sortAnnouncements\ (line~\ref{line:scan})
reads each entry from \co{Announce} and stores the value in the local copy \co{ar}.
Thus each value copied into \co{ar} cannot be in \co{o}'s interval.
By Lemma \ref{intersect-works}, \co{o} is included in the \Redundant\ set returned by $I$.
%
\end{proof}

We denote by $H_G$ the number of needed objects at configuration $G$.
We denote by $H_{max}$ the maximum number of needed objects over all configurations of the execution $E$,
i.e., $H_{max} = \max \{H_G~ |~ G \in E\}$.

\begin{lemma}
\label{Q-bound}
In the configuration after a dequeue from \queue, \queue\ contains at most\linebreak $2H_{max}+15P^2\log P$ objects.
In the configuration after an enqueue, \queue\ contains at most $2H_{max}+20P^2\log P$ objects.
\end{lemma}
\begin{proof}
We prove the claim by induction.
The claim is clearly true after the first enqueue or dequeue in the execution (since an enqueue can
add at most $2P\log P$ objects to the initially empty queue, by Invariant \ref{inv:distinct}).
We assume the claim is true for all configurations prior to some operation on \queue, and prove
that it holds in the configuration $G$ after that operation.

First, consider a dequeue from \queue.
If \queue\ is empty, then the claim is clearly true.
So, suppose the dequeue removes some pool $L$ from \queue.
Let $G'$ be the configuration after $L$ was enqueued.

To bound the number of objects in \queue\ at $G$
we first bound the number of objects that are in \co{Q} at both of the configurations $G'$ and $G$.
Consider any such object \co{o}.
The pool $L_{\co{o}}$ containing \co{o} was dequeued from \co{Q} by some \retireInterval\ between $G'$ and $G$.
By Lemma \ref{return-redundant}, the \retireInterval\ re-enqueues into \co{Q} only objects that
are needed at the time $L_{\co{o}}$ was dequeued.  So, \co{o} was needed at that time.
By Lemma \ref{remain-redundant}, this means \co{o} was needed at the earlier configuration $G'$.
So, there are at most $H_{max}$ such objects \co{o}.
	
Next, we count the objects that were dequeued from \queue\ before $G'$ 
and were re-enqueued between $G'$ and $G$. 
All of these objects were held in the local memory of a \retireInterval\ at $G'$.
Since each process holds at most $4P\log P$ dequeued objects at any time, 
there are at most $4P^2\log P$ objects in this category. 

Finally, we bound the number of {\it new} objects that are enqueued into \queue\ 
(by line \ref{line:enq-retired}) for the first time between $G'$ and $G$.
Let $N'$ be the number of objects in \queue\ at $G'$.
By Invariant \ref{inv:distinct}, each pool in \queue\ contains at least
$B$ objects, so there are at most $\frac{N'}{B}$ pools in \queue\ at $G'$.
Thus, the number of dequeues from \queue\ between $G'$ and $G$ 
is at most $\frac{N'}{B}$.
Between any pair of enqueues by a process executing line \ref{line:enq-retired}, 
the process must do two dequeues at line \ref{line:deq}.
Thus, there are at most $P+\frac{N'}{2B}$ pools enqueued into \queue\ between $G'$ and $G$.
Each pool contains exactly $B$ new objects for a total of at most $BP+\frac{N'}{2}$ new objects.
By the induction hypothesis, $N'\leq 2H_{max}+20 P^2\log P$, so the number of new objects enqueued
is at most $P^2\log P + H_{max} + 10 P^2 \log P$.

In total, the number of objects in \queue\ at $G$ is at most 
$H_{max} + 4P^2\log P + (H_{max} + 11 P^2\log P)$, as required.

Next, consider the configuration $G$ after an enqueue into \queue.
Let $G'$ be the configuration after the preceding dequeue (or the initial configuration if
there is no such dequeue).  By the induction hypothesis, the number of objects in \queue\ at $G'$
is at most $2H_{max}+15 P^2\log P$.  Between $G$ and $G'$, each process can execute steps of at most 
one flush phase, and each flush phase can enqueue pools containing at most $5P\log P$ objects.
So, the number of objects in \queue\ at $G$ is at most $2H_{max}+15P^2\log P + 5P^2\log P$, as required.
\end{proof}


\RTspace*

\begin{proof}
By Lemma \ref{Q-bound}, there are at most $2H_{max}+20 P^2 \log P$ deprecated objects in \queue,
and a further $5P\log P$ in the local memory of any process, for a total of $2H_{max} + 25 P^2\log P$.
\end{proof}

\section{Correctness and Analysis of Version List Implementation}
\label{list-appendix}

In this section, we prove Theorems \ref{list-correctness}, \ref{remove-amortized-time} and \ref{space-bound} from Section \ref{sec:vlists}.
These show that our List implementation is linearizable and give time and space bounds for it.

Throughout this section, we fix an execution.
Let $\V_1, \V_2, \ldots$ be the nodes assigned to the \co{Head} pointer 
(in the order they were stored there), and let 
$\ts_1, \ts_2,\ldots$ be the timestamps assigned to those nodes.  
If a node $\V_i$ is never assigned a timestamp, then $\ts_i = \bot$.
(The last node in the sequence, if there is one, might not be assigned a timestamp).
We assume that the execution satisfies the following
statements, which are more formal versions of those stated in Assumption \ref{list-assumptions}.
\begin{enumerate}[\ref{list-assumptions}.1]
\item
\label{append-before-remove}
For any node \B, if there are calls to both \co{tryAppend(B,*)} and \co{remove(B)} then some call to \co{tryAppend(B,*)} returns true before the invocation of \co{remove(B)}.
\item
\label{remove-pre}
If \co{remove(B)} is invoked before any call to \co{tryAppend(B,*)} returns true, then
there is no \co{tryAppend} or \co{find} that is either pending or invoked after the \co{remove(B)} is invoked.
\item
\label{removals-safe}
If \co{remove($\V_i$)} is invoked and $\ts_{i+1} \neq \bot$, then there is no \co{find(*,ts)} with $\ts_i\leq \ts < \ts_{i+1}$ that is either pending or invoked after the \co{remove($\V_i$)} is invoked.
\item
\label{append-pre}
Before \co{tryAppend(B,C)} or \co{find(B,*)} is invoked, \B\ is returned by some \co{getHead} operation and \B's timestamp field has been set.  Once a node's timestamp is set, it never changes. 
If \co{tryAppend(B,C)} succeeds, the timestamp eventually assigned to \C\ is greater than or equal to \B's timestamp. 
\item
\label{find-append}
If \co{find($\V_i$,t)} is invoked and \co{tryAppend($\V_i$,$\V_{i+1}$)} succeeds, then $\ts_{i+1} > \co{t}$.
\item
\label{no-duplicates}
There are no two invocations of \co{remove(B)} or of \co{tryAppend(*,C)}.
\end{enumerate}

In our usage in version lists, version nodes are removed only after they have been superseded by a
later version (\ref{list-assumptions}.\ref{append-before-remove}).
The exception is the last version node, which is removed when the list is no longer needed,
and hence no more calls to \co{tryAppend} or \co{find} should be executing after this happens (\ref{list-assumptions}.\ref{remove-pre}).

\Eric{We need to explain in Step 3 why Figure \ref{fig:vcas-alg-youla} satisfies these preconditions}
\Hao{double check}

\subsection{Correctness of Version Lists}
\label{sec:list-correct}

We say a node is {\it active} if it has been successfully added to the list.
More precisely, a node becomes active when the \co{Head} pointer is changed to point to it at line
\ref{head-cas} and it remains active forever after that.

In our proof, we use the notation $\X \rightarrow \Y$ to mean $\co{X->right} = \Y$ and
$\X \not\rightarrow \Y$ to mean $\co{X->right} \neq \Y$.
Similarly, $\X \leftarrow \Y$ means $\co{Y->left} = \X$
and $\X \not\leftarrow \Y$ means $\co{Y->left} \neq \X$.
Finally, $\X \leftrightarrow\Y$ means $\X\leftarrow \Y$ and $\X\rightarrow \Y$.

We define two orderings on all the nodes that become active during the execution:
if \X\ and \Y\ are pointers to nodes, we write
$\X <_p \Y$ if $\co{X->priority} < \co{Y->priority}$ and $\X <_c \Y$ if $\co{X->counter} < \co{Y->counter}$.
The \co{counter} and \co{priority} fields of a node are initialized in the \co{tryAppend} routine
before the node becomes active at line \ref{head-cas}.
So, the \co{counter} and \co{priority} fields of an active node are well-defined and never change.
Hence, the orderings $<_c$ and $<_p$ are fixed.
\er{In particular, we have $\V_1 <_c \V_2 <_c \cdots$.}
For convenience, we also define $\X >_p \co{null}$ if \X\ is a pointer to any node.


We start with some simple lemmas that follow easily from the code.  The following straightforward lemma ensures that,
whenever the code reads \co{X->F} for some field \co{F} of a node pointer \X,  \X\ is not \co{null}.

\begin{lemma}
\label{splice-pre}
If \spl, \sul\ or \sur\ is called with arguments \A, \B, \C,
or if there is a Descriptor (other than \frozen) containing \A, \B, \C,
then \B\ is not \co{null} and there was a previous time when $\A\leftarrow \B$ and a previous time when $\B\rightarrow \C$.\\
Every call to \co{removeRec} has a non-\co{null} argument.\\
Every call to \co{\sul(A,B,C)} has a non-\co{null} argument for \A.\\
Every call to \co{\sur(A,B,C)} has a non-\co{null} argument for \C.
\end{lemma}
\begin{proof}
Assume the claim holds in some prefix of the execution.  We prove that it holds after one more step~$s$.

If $s$ creates a Descriptor at line \ref{create-right-desc} or \ref{create-left-desc}, the claim follows from the precondition of the \sul\ or \sur\ that performs $s$.

If $s$ is an invocation of \co{\sul(A,B,C)} at line \ref{call-sul} of \co{removeRec}, then \B\ is
not \co{null} by the precondition of \co{removeRec}.
The test on the preceding line ensures \co{a} is not 0, and hence \A\ is not \co{null}.
Furthermore, $\A\leftarrow \B$ was previously true at line \ref{read-left} and $\B\rightarrow \C$
was true at line \ref{read-right}.

If $s$ is an invocation of \sur{} at line \ref{call-sur} of \co{removeRec}, the proof is symmetric to the call of \sul.

If $s$ is an invocation of \spl\ in \co{removeRec}, then
the claim follows from the execution of lines \ref{read-left}--\ref{read-right} that precedes $s$, and the precondition of \co{removeRec}.

If $s$ is an invocation of \spl\ in the \co{help} routine, then the claim follows from the
induction hypothesis on the Descriptor passed as an argument to \co{help}, because the test on
line \ref{test-frozen} ensures that the Descriptor is neither \co{null} nor \frozen.

If $s$ is a recursive call to \co{removeRec} inside \co{removeRec}, it is preceded by a
call to \co{validAndFrozen} that ensures
the node used as the argument is not \co{null}.

If $s$ is a call to \co{removeRec} inside \co{remove}, the claim follows from the precondition that the argument of a \co{remove} is not \co{null}.
\end{proof}

We next observe that Descriptors are correctly created by \sul\ and \sur\ routines.

\begin{lemma}
\label{desc-properties}
If \co{X->rightDesc} is neither \co{null} nor \frozen, then the Descriptor it points to contains $(\X,\Y,\Z)$ for some \Y\ and \Z\ satisfying $\X >_p \Y >_p \Z$.\\
Similarly, if \co{Z->leftDesc} is neither \co{null} nor \frozen, then the Descriptor contains $(\X,\Y,\Z)$ for some \X\ and \Y\ satisfying $\X <_p \Y <_p \Z$.
\end{lemma}
\begin{proof}
We prove the first claim; the second is symmetric.
The Descriptor can only be stored in \co{X->right} at line \ref{store-right-desc}.
It was created at line \ref{create-right-desc} in a call to \co{\sul(X,Y,Z)} for some \Y\ and \Z.
The Descriptor contains $(\X,\Y,\Z)$ and by the test preceding the
call to \co{\sul(X,Y,Z)} on line \ref{call-sul}, we have $\X >_p \Y >_p \Z$.
\end{proof}

The next lemma shows that nodes of the list remain in the order they were appended, and
that whenever a \co{left} or \co{right} pointer is updated, it points further along the list than it did
before the change (i.e., nodes are spliced out).  It also ensures that there are no cycles of \co{left} pointers or cycles of \co{right} pointers
and that \co{left} and \co{right} pointers are not subject to ABA problems.


\begin{lemma}
\label{order}
Let \X\ be any node.\\
When \co{X->left} is first set to a non-\co{null} value \W, \W\ is an active node and $\W <_c \X$.\\
Whenever the value of \co{X->left} changes from a value \W\ to a non-\co{null} value \U, \U\ is an active node and $\U <_c \W <_c \X$.\\
Similarly, When \co{X->right} is first set to a non-\co{null} value \Y, \Y\ is an active node and $\X <_c \Y$.\\
Whenever the value of \co{X->right} changes from a non-\co{null} value \Y\ to \Z, \Z\ is an active node and $\X <_c \Y <_c \Z$.
\end{lemma}
\begin{proof}
A CAS on \co{X->left} at line \ref{left-cas} uses a non-\co{null} expected value, by Lemma \ref{splice-pre}.
So, \co{X->left} is first set to a non-\co{null} value \W\ at line \ref{set-left} of tryAppend.
By Assumption \ref{list-assumptions}.\ref{append-pre}, \W\ is active.  Moreover, $\co{X->counter} = \co{W->counter}+1$.

\co{X->left} is never changed by this line again because \co{tryAppend(*,X)} is called at most once, by Assumption \ref{list-assumptions}.\ref{no-duplicates}.
Consider an update of \co{X->left} from value \W\ to a non-\co{null} value \U, which can only happen at line \ref{left-cas} of a \co{splice(U,W,X)}.  Assume the claim holds in the prefix of the execution prior to this update.
By Lemma \ref{splice-pre}, \W\ is not \co{null} and there was an earlier time when $\U \leftarrow \W$.
By the induction hypothesis, if \U\ is not \co{null}, then \U\ is active and $\U <_c \W <_c \X$.  

A CAS on \co{X->right}  at line \ref{right-cas} uses a non-\co{null} expected value, by Lemma \ref{splice-pre}.
So, \co{X->right} is first set to a non-\co{null} value \Y\ at line \ref{mr-cas1} or \ref{mr-cas2} of \co{tryAppend}.
Prior to this step, we have $\X \leftarrow \Y$ at line \ref{mr-read-left} or \ref{set-left}, respectively.  As proved above, this means that $\X <_c \Y$.
By assumption \ref{list-assumptions}.\ref{append-pre}, the new value \Y\ used by the CAS at line \ref{mr-cas1} is an active node.
The new value \Y\ used by the CAS at line \ref{mr-cas2} is an active node, since it was successfully stored in \co{Head} at line \ref{head-cas}.

Consider an update to \co{X->right} from a non-\co{null} value \Y\ to a non-\co{null} value \Z, which can only happen at line \ref{right-cas} of \co{splice(X,Y,Z)}.  Assume the claim holds in the prefix of the execution prior to this update.
There was an earlier time when
$\Y\rightarrow \Z$, by Lemma \ref{splice-pre}.
Thus, by the induction hypothesis, \Z\ is active and $\X <_c \Y <_c \Z$.
\end{proof}

The following lemma and its corollary  are used to ensure that a node is
spliced out of the list only if a remove operation has been called on it.

\begin{lemma}
\label{splice-marked}
Before \co{removeRec(Y)}, \co{splice(X,Y,Z)}, \co{\sul(X,Y,Z)} or\linebreak
\co{\sur(X,Y,Z)} is called, \co{Y->status} is set to \co{marked}.
\end{lemma}
\begin{proof}
By Lemma \ref{splice-pre}, \Y\ is not \co{null}.
\co{Y->status} is initially \co{unmarked}.
The only way that it can change from \co{unmarked} to another value
is by setting it to \co{marked} at line \ref{mark} so the first execution of that line on the node will succeed.

If \co{removeRec(Y)} is called at line \ref{call-rr1} of \co{remove(Y)}, \co{Y->status} was previously set to \co{marked} at line \ref{mark}.
If \co{removeRec(Y)} is called recursively inside the \co{removeRec} routine,
the test \co{validAndFrozen(Y)} must first have returned true.
So, at some earlier time, \co{Y->rightDesc} was set to \co{frozen} at line \ref{freeze-desc}, which can
only happen after \co{Y->status} is set to \co{marked} at line~\ref{mark}.

The routines \co{\sul(X,Y,Z)} and \co{\sur(X,Y,Z)} can only be called by \co{removeRec(Y)},
and it has already been shown that \co{Y->status} is set to \co{marked} before \co{removeRec(Y)} is called.  The same applies if \co{splice(X,Y,Z)} is called in \co{removeRec(Y)}.

If \co{splice(X,Y,Z)} is called by \co{help}, then \co{help} was called on a Descriptor containing $(\X,\Y,\Z)$.
This Descriptor was created either by a \co{\sul(X,Y,Z)} or \co{\sur(X,Y,Z)}.  So, as shown above,
\co{Y->status} was already set to \co{marked}.
\end{proof}


\begin{corollary}
\label{remove-safe}
If a node \Y\ is finalized, some process has called \co{remove(Y)}.
\end{corollary}
\begin{proof}
\Y\ can only be finalized by a call to \co{splice(*,Y,*)}.  By Lemma \ref{splice-marked}
\co{Y->status} is set to \co{marked} before the \co{splice} is called, and this can only
be done at line \ref{mark} of a \co{remove(Y)}.
\end{proof}

Next, we show that the first CAS step that attempts to finalize a node succeeds, and permanently changes
the node's status to \co{finalized}.

\begin{lemma}
\label{finalize-works}
At all times after line \ref{finalize-cas} is performed on a node \X, \co{X->status} is \co{finalized}.
\end{lemma}
\begin{proof}
Consider the first time the CAS on line \ref{finalize-cas} is performed on \X.
By Lemma \ref{splice-marked}, \co{X->status} was set to \co{marked} prior to this CAS.
There is no way for \co{X->status} to change to any other value, except \co{finalized}.
So the CAS will succeed and change \co{X->status} to \co{finalized}.
After this occurs, no execution of line \ref{mark} can change the value of \co{X->status},
since \co{remove(X)} is called at most once by assumption \ref{list-assumptions}.\ref{no-duplicates}, and no other step can change the value of \co{X->status}.
\end{proof}

The next two lemmas describe freezing the fields of a node that can store Descriptors.
After a \co{remove(X)} marks \X, it performs two attempts to change each of its Descriptor fields
to \co{frozen}.  Although the first attempt may fail, we prove that no other Descriptor can be
written into these fields after that first attempt.  This, in turn, guarantees that the second
attempt will succeed, and the Descriptor fields remain \co{frozen} forever after the second attempt.
After this, it is safe to go on to try to splice \X\ out of the list.

\begin{lemma}
\label{desc-frozen}
Let $\co{F}\in\{\co{rightDesc}, \co{leftDesc}\}$.
No non-\frozen\ Descriptor can be stored in \co{X->F} after a \co{remove(X)} performs
line \ref{freeze-desc} on \co{X->F}.
\end{lemma}
\begin{proof}
We prove the claim for \co{X->rightDesc}.  The proof for \co{X->leftDesc} is symmetric.
Let $cas_1$ be the first execution of line \ref{freeze-desc} on \co{X->rightDesc} and let
$r_1$ be the preceding read of \co{X->rightDesc} at line \ref{read-desc1}.
To derive a contradiction, assume there is some step  that stores a non-\frozen\ value
in \co{X->rightDesc} after $cas_1$.  Let $cas_2$ be the first such step.
Then, $cas_2$ is an execution of line \ref{store-right-desc} in a call to \co{\sul(X,Y,Z)} for some \Y\ and \Z.
Let $r_2$ be the preceding read of \co{X->rightDesc} at line \ref{read-desc2}.

Prior to $r_1$, \co{X->status} is set to \co{marked} at line \ref{mark}, and \co{X->status} remains \co{marked} or \co{finalized} forever after.
So, $r_2$ must occur before $r_1$, since the line after $r_2$ observes that \co{X->status} is 
unmarked.
Since $r_2$ is before $cas_1$, which is the first attempt to CAS \frozen\ into \co{X->rightDesc},
it must read a value $D$ different from \co{frozen}.
Since $cas_2$ succeeds, \co{X->rightDesc} still has the value $D$ when $cas_2$ occurs.
Since at most one CAS attempts to store $D$ in \co{X->rightDesc}, this means \co{X->rightDesc}
has the value $D$ at all times between $r_2$ and $cas_2$.
But $r_1$ and $cas_1$ are within this interval of time, so $r_1$ would read $D$ and
$cas_1$ would succeed in changing the value from $D$ to \co{frozen}.  This contradicts the
fact that \co{X->rightDesc} does not change between $r_2$ and $cas_2$.
\end{proof}

\begin{lemma}
\label{frozen-before-rr}
If \co{removeRec(X)} has been called, then both \co{X->rightDesc} and \co{X->leftDesc} are \co{frozen}.
\end{lemma}
\begin{proof}
Let $\co{F} \in \{\co{leftDesc},\co{rightDesc}\}$.
Lemma \ref{desc-frozen} says that no non-\frozen\ Descriptor can be stored in \co{X->F} after
line \ref{freeze-desc} has been performed on \co{X->F}.
Thus, the second execution of line \ref{freeze-desc} on \co{X->F} stores \frozen\ in \co{X->F},
and \co{X->F} can never change thereafter.
So, if \co{removeRec(X)} is called at line \ref{call-rr1}, \frozen\ has already been stored in \co{X->F}.

If \co{removeRec(X)} is called recursively from the \co{removeRec} routine, then there
is a preceding test that ensures \co{X->rightDesc} is \frozen.  Since \frozen\ can be written
to \co{X->rightDesc} only after two executions of line \ref{freeze-desc} on \co{X->leftDesc},
it follows that \co{X->leftDesc} has also been frozen.
Thus, \frozen\ has already been stored in \co{X->F} before the call to \co{removeRec(X)}.

By Lemma \ref{desc-frozen}, once \frozen\ is stored in \co{X->F} (at line \ref{freeze-desc}),
\co{X->F} can never be changed again.
\end{proof}

Since the \co{Head} pointer can only change at line \ref{head-cas}, we have the following observation.
It implies that there is no ABA problem on the \co{Head} pointer.

\begin{observation}
\label{head-increases}
If \co{Head} changes from node \co{X} to node \co{Y},  $\co{Y->counter} = \co{X->counter} +1$.
\end{observation}

A {\it waning configuration} $C$ is one where, for some node \X, a \co{remove(X)} was called
before $C$ and before any \co{tryAppend(X,*)} had returned true.
A configuration is called a {\it waxing configuration} if it is not a waning configuration.
Once the execution switches from waxing to waning configurations, there can no longer
be any \co{tryAppend} operations, by precondition \ref{list-assumptions}.\ref{remove-pre}.
Thus, the list can grow while the execution is in waxing configurations, but can
only shrink once it enters a waning configuration.

\begin{lemma}
\label{head-unfinalized}
In any waxing configuration, the node that \co{Head} points to is not finalized.
\end{lemma}
\begin{proof}
Suppose \co{Head} points to \X\ in some waxing configuration $C$.
By Observation \ref{head-increases}, no \co{tryAppend(X,*)} has performed a  
successful CAS at line \ref{head-cas}.
Thus, no \co{tryAppend(X,*)} has returned true before $C$.
Since $C$ is a waxing configuration, this means that no  \co{remove(X)} has been invoked before $C$.
By Corollary \ref{remove-safe}, \X\ is not finalized in~$C$.
\end{proof}

\begin{figure}[t]
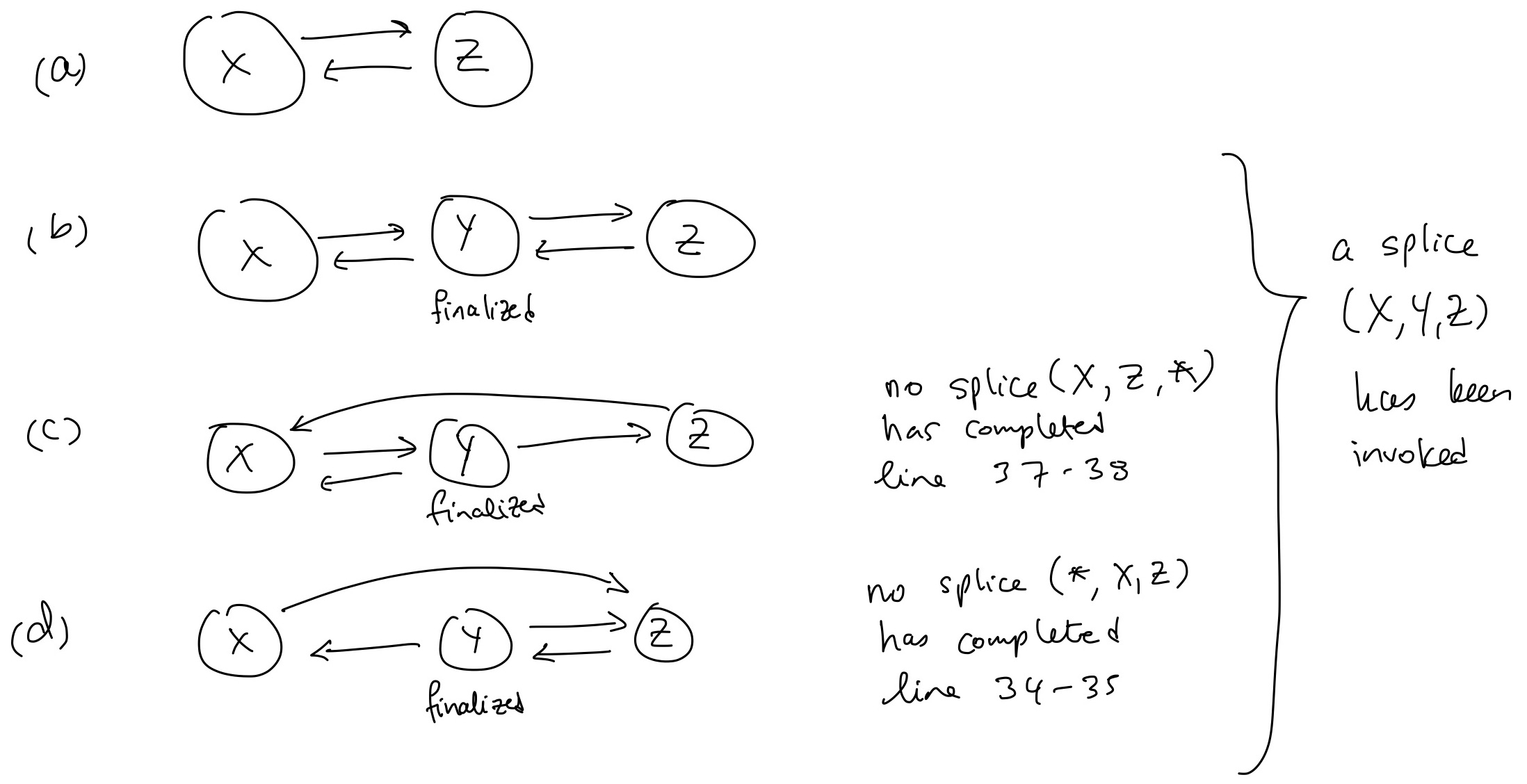
\caption{Possible configurations for list nodes, as described in Invariant \ref{inv}, claims \ref{well-formed}--\ref{well-formed-right}.  Finalized nodes are grey.
}
\label{possibilities}
\end{figure}

The following invariant gives a fairly complete description of the state of the doubly-linked
list at any time.
All of the unfinalized nodes are in the list.
There can be at most one finalized node between any two unfinalized nodes (or to the
left of the first unfinalized node or to the right of the last unfinalized node)
if there is an ongoing splice of that finalized node.
Moreover, the \co{left} and \co{right} pointers of any nodes that have already been spliced
out of the list cannot skip over any unfinalized node.
These properties are useful in showing that any traversal of the list cannot skip
any unfinalized nodes and hence must visit every node that has not yet been removed.

Intuitively, part \ref{no-overlap} of the invariant shows that we cannot
concurrently splice out two adjacent nodes, which is important for avoiding the problem illustrated in Figure \ref{bad-example}.
We say that \co{(X,Y,Z)} is a {\it splice triple} in an execution if there is a call to \co{splice(X,Y,Z)}
or if a Descriptor (other than \co{frozen}) that contains \co{(X,Y,Z)} is stored in some node's \co{leftDesc} or \co{rightDesc} field.

\begin{invariant}
\label{inv}
In any configuration $C$, the following statements hold.  Let $\U_1, \U_2, \ldots, \U_k$ be the sequence of active, non-finalized nodes in $C$, ordered by $<_c$.
\begin{enumerate}
\item
\label{well-formed}
For $1\leq i<k$, one of the following statements is true (refer to Figure \ref{possibilities}).
\begin{enumerate}[(a)]
\item
\label{not-done}
There is a finalized node \Y\ such that $\U_i\leftrightarrow \Y\leftrightarrow \U_{i+1}$ and\\
a \co{splice(U$_i$,Y,U$_{i+1}$)} has been invoked, but no \co{splice(U$_i$,Y,U$_{i+1}$)} has performed line \ref{left-cas}.
\item
\label{left-done}
There is a finalized node \Y\ such that $\U_i\leftrightarrow \Y \rightarrow \U_{i+1}$ and $\U_i \leftarrow \U_{i+1}$ and\\
a  \co{splice(U$_i$,Y,U$_{i+1}$)} has been invoked, but no \co{splice(U$_i$,Y,U$_{i+1}$)}  has performed line \ref{right-cas}.
\item
\label{done}
$\U_i\leftrightarrow \U_{i+1}$.
\item
\label{exception}
$i=k-1$ and $\U_{k-1}\rightarrow \co{null}$ and $\U_{k-1}\leftarrow \U_k$ and\\
a \co{tryAppend($\U_{k-1}$,$\U_k$)} has performed line \ref{head-cas} successfully, but has not performed line \ref{mr-cas2} and\\
no \co{tryAppend($\U_k$,*)} has performed line \ref{mr-cas1}.
\end{enumerate}
\item
\label{well-formed-left}
If $k\geq 1$ then either
\begin{enumerate}[(a)]
\item
\label{not-done-leftend}
There is a finalized node \Y\ such that $\co{null} \leftarrow \Y \leftrightarrow \U_1$ and a
\co{splice(null,Y,U$_1$)} has been invoked, but no \co{splice(null,Y,U$_1$)}  has performed line \ref{left-cas}, or
\item
\label{done-leftend}
$\co{null} \leftarrow \U_1$.
\end{enumerate}
\item
\label{well-formed-right}
If $k\geq 1$ then either
\begin{enumerate}[(a)]
\item
\label{not-done-rightend}
There is a finalized node \Y\ such that $\U_k \leftrightarrow \Y \rightarrow \co{null}$ and a
\co{splice(U$_k$,Y,null)} has been invoked, but no \co{splice(U$_k$,Y,null)}  has performed line \ref{right-cas}, or
\item
\label{done-rightend}
$\U_k\rightarrow \co{null}$.
\end{enumerate}
\item
\label{splice-args}
If, before this configuration, either a \co{splice(X,Y,Z)} has been invoked or a Descriptor containing \co{(X,Y,Z)} has been stored in a node, then $\X\leftarrow \Y \rightarrow \Z$.
\item
\label{no-skip}
\begin{enumerate}[(a)]
\item
\label{no-skip-right}
Suppose \X\ is finalized and $\X\rightarrow \Z$.  If $\Z \neq \co{null}$, then all nodes \Y\ such that $\X<_c \Y<_c \Z$ are finalized.  If \Z\ is \co{null}, then all nodes \Y\ such that $\X<_c \Y$ are finalized.
\Hao{I believe this statement would be useful in step 3 if we change it to say 'all nodes \Y\ such that $\X<_c \Y<_c \Z$ are \lrunreachable'. This would require updating the proof a bit.}
\Eric{Since the stronger statement is not needed for the induction to work, I'd prefer to state the stronger version as a corollary of the invariant (I think it is), since the induction proof is already complicated as it is.}
\item
\label{no-skip-left}
Suppose \Z\ is finalized and $\X\leftarrow \Z$.  If $\X \neq \co{null}$, then all nodes \Y\ such that $\X<_c \Y<_c \Z$ are finalized.  If \X\ is \co{null}, then all nodes \Y\ such that $\Y<_c \Z$ are finalized.
\Eric{This item will help show that a find won't skip any non-finalized nodes, and by Lemma \ref{remove-safe} that no non-removed nodes are skipped}
\end{enumerate}
\item
\label{descriptors-work}
For any Descriptor \co{(X,Y,Z)} that was in the \co{rightDesc} or \co{leftDesc} field of a node in a previous configuration but is no longer in that field,
\begin{enumerate}[(a)]
\item
\label{descriptors-work-right}
if $\X \neq \co{null}$ then $\co{X->right} >_c \Y$, and 
\item
\label{descriptors-work-left}
if $\Z \neq \co{null}$ then $\co{Z->left} <_c \Y$.
\end{enumerate}
\item
\label{no-overlap}
There are no two splice triples of the form \co{(W,X,Y)} and \co{(X,Y,Z)} in the prefix of the execution up to $C$.
\item
\label{unfinalized}
If a \co{splice(X,Y,Z)} has been called and \Y\ is not finalized, then neither \X\ nor \Z\ are finalized nodes.
\end{enumerate}
\end{invariant}


\Eric{Remark:  claim \ref{splice-args} implies that if \co{splice(A,B,C)} and \co{splice(A$'$,B,C$'$)} are called, then $\A=\A'$ and $\C=\C'$.}


\begin{proof}
First, we show that claim \ref{unfinalized} follows from the others.
Suppose claims \ref{well-formed}, \ref{well-formed-left}, \ref{well-formed-right}, \ref{splice-args} and \ref{no-overlap} hold for some configuration
in which \Y\ is not finalized and a \co{splice(X,Y,Z)} has been called.
By claim \ref{splice-args}, $\X\leftarrow \Y \rightarrow \Z$.
To derive a contradiction, suppose \X\ is a finalized node.  The proof for \Z\ is symmetric.
If \Y\ is $\U_1$ then by claim \ref{well-formed-left}, a \co{splice(null,\X,\Y)} has been invoked
since $\co{null} \not\leftarrow \Y$.
Otherwise, \Y\ is $\U_{i+1}$ for some $i\geq 1$.
Since $\X\leftarrow \Y$, the non-finalized nodes $\U_i$ and $\U_{i+1}$ satisfy case \ref{not-done} of claim \ref{well-formed},
and a \co{splice($\U_i$,\X,\Y)} has been invoked.
Either way, we have an invocation of \co{splice(X,Y,Z)} and \co{splice(*,X,Y)} contradicting claim \ref{no-overlap}.

\medskip

It remains to show that for any step $s$ if all the claims are satisfied in all configurations before $s$, then claims \ref{well-formed} to \ref{no-overlap} are satisfied in the configuration after $s$.
The only steps that could affect the truth of claims \ref{well-formed}--\ref{no-overlap} are
successful CAS steps,
an invocation of \co{splice},
executions of line  \ref{mr-cas1}, \ref{mr-cas2}, \ref{left-cas} or \ref{right-cas}, that do not perform a successful CAS,
and executions of line \ref{mark}.
So, we shall consider each of these types of steps in turn.

\Eric{Recent rearrangement of code makes the ordering of cases in the following now seem a bit random.  Could reorder the cases.}
\begin{enumerate}[{C{a}se 1}]
\item
\label{left-cas-proof}
Suppose $s$ is a successful CAS at line \ref{left-cas} by a call to \co{splice(A,B,C)}.
This step changes \co{C->left} from \B\ to \A.

To derive a contradiction, suppose \C\ is finalized in the configuration before $s$.
Consider the \co{splice(*,C,*)} operation that finalized \C.
By induction hypothesis \ref{splice-args}, since $\B \leftarrow \C$, it must have been a \co{splice(B,C,*)} operation.  This contradicts induction hypothesis~\ref{no-overlap}.
Thus, \C\ is not finalized when $s$ occurs.

By Lemma \ref{finalize-works}, \B\ is finalized when $s$ occurs.
In the configuration before $s$, we have $\A\leftarrow \B \rightarrow \C$ by induction hypothesis \ref{splice-args} and $\B \leftarrow \C$ since the CAS succeeds.

If \A\ is \co{null}, then by the induction hypothesis,
\C\ must be the leftmost non-finalized node and
case \ref{not-done-leftend} of induction hypothesis \ref{well-formed-left} is satisfied.
Thus, once $s$ updates \co{C->left} to \co{null}, case \ref{done-leftend} of claim
\ref{well-formed-left} is satisfied.

If \A\ is not \co{null}, then
by the  induction hypothesis,
\A\ is not finalized and the nodes \A, \B, \C\ satisfy
case \ref{not-done} of induction hypothesis \ref{well-formed}, so
$\A\rightarrow \B$ in the configuration prior to~$s$.
Thus, once $s$ updates \co{C->left} to \A, case \ref{left-done} of claim \ref{well-formed} will be satisfied for these three nodes.

Thus, claims \ref{well-formed}, \ref{well-formed-left} and \ref{well-formed-right} are preserved by $s$.

Next, we verify that $s$ preserves claim \ref{splice-args}.
Suppose \co{splice(X,Y,Z)} has been called before $s$.  The only way that $s$ can cause claim \ref{splice-args} to become false is if
it changes \co{Y->left} from \X\ to another value.  If $s$ did this, it would have to be part of a
\co{splice(*,X,Y)} operation, but this would contradict induction hypothesis \ref{no-overlap}.

Next, we verify that $s$ preserves claim \ref{no-skip}.
We consider two cases.

Suppose \A\ is \co{null}.  Since claim \ref{no-skip-left} holds before $s$,
it suffices to check that there are no finalized nodes less than \B\ in the $<_c$ order.
This follows from induction hypothesis \ref{no-skip-left}, because when $s$ occurs,
\B\ is finalized and $\co{null}\leftarrow \B$.

Suppose \A\ is not \co{null}.  Since claim \ref{no-skip-left} holds before $s$,
it suffices to check that there are no finalized nodes between \A\ and \B\ in the $<_c$ order.
This follows from induction hypothesis \ref{no-skip-left}, because when $s$ occurs,
\B\ is finalized and $\A \leftarrow \B$.

It follows from Lemma \ref{order} that step $s$ preserves claim \ref{descriptors-work}.

\item
\label{right-cas-proof}
Suppose $s$ is a successful CAS at line \ref{right-cas} by a call to \co{splice(A,B,C)}.
This step changes \co{A->right} from \B\ to \C.  (This case is nearly symmetric to \ref{left-cas-proof}.)

To derive a contradiction, suppose \A\ is finalized in the configuration before $s$.
Consider the \co{splice(*,A,*)} operation that finalized \A.
By induction hypothesis \ref{splice-args}, since $\A\rightarrow\B$, it must have been a \co{splice(*,A,B)} operation.  This contradicts induction hypothesis~\ref{no-overlap}.
Thus, \A\ is not finalized when $s$ occurs.

By Lemma \ref{finalize-works}, \B\ is finalized when $s$ occurs.
In the configuration before $s$, we have $\A\leftarrow \B \rightarrow \C$ by induction hypothesis \ref{splice-args} and $\A \rightarrow \B$ since the CAS succeeds.

If \C\ is \co{null}, then by the induction hypothesis,
\A\ must be the rightmost non-finalized node and
case \ref{not-done-rightend} of induction hypothesis \ref{well-formed-right} is satisfied.
Thus, once $s$ updates \co{A->right} to \co{null}, case \ref{done-rightend} of claim
\ref{well-formed-right} is satisfied.

If \C\ is not \co{null}, then
by the induction hypothesis,
\C\ is not finalized and the nodes \A, \B, \C\ satisfy
case \ref{left-done} of induction hypothesis \ref{well-formed} in the configuration before $s$.
(Case \ref{not-done} cannot be satisfied, since the \co{splice(A,B,C)} that performs $s$ has already
completed line \ref{left-cas}.)
So $\A\leftarrow \C$ in the configuration prior to $s$.
Thus, once $s$ updates \co{C->left} to \A, case \ref{done} of claim \ref{well-formed} will be satisfied with $\A\leftrightarrow \C$.

Thus, claims \ref{well-formed}, \ref{well-formed-left} and \ref{well-formed-right} are preserved by $s$.

Next, we verify that $s$ preserves claim \ref{splice-args}.  Suppose \co{splice(X,Y,Z)} has been called before $s$.  The only way that $s$ can cause claim \ref{splice-args} to become false is if
it changes \co{Y->right} from \Z\ to another value.  If $s$ did this, it would have to be part of a
\co{splice(Y,Z,*)} operation, but this would contradict induction hypothesis \ref{no-overlap}.

Finally, we verify that $s$ preserves claim \ref{no-skip}.
We consider two cases.

Suppose \C\ is \co{null}.  Since claim \ref{no-skip-right} holds before $s$,
it suffices to check that there are no finalized nodes greater than \B\ in the $<_c$ order.
This follows from induction hypothesis \ref{no-skip-right}, because when $s$ occurs,
\B\ is finalized and $ \B\rightarrow \co{null}$.

Suppose \C\ is not \co{null}.  Since claim \ref{no-skip-right} holds before $s$,
it suffices to check that there are no finalized nodes between \B\ and \C\ in the $<_c$ order.
This follows from induction hypothesis \ref{no-skip-right}, because when $s$ occurs,
\B\ is finalized and $\B \rightarrow \C$.

It follows from Lemma \ref{order} that step $s$ preserves claim \ref{descriptors-work}.

\item
\label{finalize-cas-proof}
Suppose $s$ is a successful CAS at line \ref{finalize-cas} of \co{splice(A,B,C)}.  Then, $s$ finalizes \B.

In the configuration before $s$, $\A\leftarrow \B \rightarrow \C$ by
induction hypothesis \ref{splice-args}
and neither \A\ nor \C\ is a finalized node
by induction hypothesis \ref{unfinalized}.
Moreover, since the CAS $s$ succeeds, it follows from Lemma \ref{finalize-works} that
no \co{splice(A,B,C)} has performed line \ref{finalize-cas}
before, and hence no \co{splice(A,B,C)} has performed line \ref{left-cas} or \ref{right-cas}.

If \C\ is not \co{null}, then
$\B \leftarrow \C$ by induction hypothesis \ref{well-formed} applied to the unfinalized nodes \B\ and \C.

If \A\ is not \co{null}, we apply induction hypothesis \ref{well-formed} to the
unfinalized nodes \A\ and \B.
Since $\A \leftarrow \B$, they either satisfy case \ref{left-done} or \ref{done}.
Suppose they satisfy case \ref{left-done}.  Then there is a finalized node \D\ such that
$\A\rightarrow \D \rightarrow \B$.  By Lemma \ref{order}, $\D <_c \B$.
Also by Lemma \ref{order}, $\D \geq_c \B$, since $\A\rightarrow \B$ at the earlier step
when line \ref{test-right} is executed.  This contradiction implies that if \A\ is non-\co{null},
then \A\ and \B\ satisfy
case \ref{done} of claim \ref{well-formed}, so $\A\leftrightarrow\B$ in the configuration before $s$.

Suppose both \A\ and \C\ are \co{null}.
Since $\co{null}\leftarrow \B$, \B\ is the smallest active, non-finalized node with respect to $<_c$.
Since $\B\rightarrow \co{null}$, \B\ satisfies either induction hypothesis \ref{exception} (with $k-1=1$) or \ref{done-rightend} (with $k=1$) in the configuration before $s$.
We next show that it cannot, in fact, be \ref{exception}.  Since $s$ finalizes \B, a \co{remove(B)} has been invoked before $s$.  If claim \ref{exception} were satisfied, then there would also be a pending
\co{tryAppend(B,*)} that performed line \ref{head-cas} successfully, so no other \co{tryAppend(B,*)} could have returned true.  This contradicts assumption \ref{list-assumptions}.\ref{append-before-remove}.
Thus, \B\ is the only active unfinalized node before $s$.  So, after $s$ finalizes \B, there are no active, non-finalized
nodes, so claims \ref{well-formed}, \ref{well-formed-left} and \ref{well-formed-right} are trivially
satisfied.

If \A\ is \co{null} but \C\ is not \co{null}, then
$\co{null} \leftarrow \B \leftrightarrow \C$ and \B\ becomes finalized by $s$, leaving
\C\ as the leftmost active, non-finalized node.
Thus, case \ref{not-done-leftend} of
claim \ref{well-formed-left} is satisfied after $s$.

Symmetrically, if \C\ is \co{null} but \A\ is not \co{null}, then
$\A \leftrightarrow \B \rightarrow \co{null}$ and \B\ becomes finalized by $s$, leaving
\A\ as the rightmost active, non-finalized node.  Thus, case \ref{not-done-rightend} of claim
\ref{well-formed-right} is satisfied after $s$.

If neither \A\ nor \C\ is \co{null}, $\A \leftrightarrow \B \leftrightarrow \C$ when $s$ occurs.
Thus, when $s$ finalizes \B, \A\ and \C\ become consecutive,
non-finalized nodes that satisfy case \ref{not-done} of claim \ref{well-formed}.

In all cases, claims \ref{well-formed}, \ref{well-formed-left} and \ref{well-formed-right} are preserved.

Step $s$ trivially preserves claim \ref{splice-args} since it does not change any
\co{left} or \co{right} pointers.

Finally, we verify that $s$ preserves claim \ref{no-skip}.
As argued above, \B's \co{left} and \co{right} pointers point to the next unfinalized node
on either side of \B\ in the order $<_c$, 
or, if they are \co{null} pointers, there are no such nodes.
So, the claim is satisfied.

\item
\label{append-proof1}
Suppose $s$ is a successful CAS at line \ref{mr-cas1}  in \co{tryAppend(B,C)}.
This CAS modifies \co{A->right} from \co{null} to \B.
By Assumption \ref{list-assumptions}.\ref{remove-pre}, the configuration before $s$ cannot be a waning configuration,
since no \co{tryAppend} can be pending in a waning configuration.

By Assumption \ref{list-assumptions}.\ref{append-pre}, \co{Head} was equal to \B\ prior to $s$.
Since \A\ was read from \co{B->left} at line \ref{mr-read-left}, we have $\A <_c \B$, by Lemma \ref{order}.
Thus, \A\ is not the last active node with respect to $<_c$.
By Observation \ref{head-increases}, $\A <_c \co{Head}$ when $s$ occurs.
To derive a contradiction, suppose \A\ is finalized when $s$ occurs.
Since $\A\rightarrow \co{null}$ in the configuration before $s$, 
induction hypothesis \ref{no-skip}.\ref{no-skip-right} implies
that there \co{Head} is finalized, contradicting Lemma \ref{head-unfinalized}.
So, \A\ is unfinalized when $s$ occurs.

Since $\A \rightarrow \co{null}$ in the configuration before $s$ and is not the last active, unfinalized node (with respect to $<_c$), the induction hypothesis
implies that \A\ must be $U_{k-1}$, satisfying case \ref{exception} of induction hypothesis \ref{well-formed}.
Thus, some \co{tryAppend(A,U$_k$)} has performed line \ref{head-cas}, but no 
\co{tryAppend(U$_k$,*)} has made another node greater then $\U_k$ active.  
Since \B\ is an active node with $\B>_c\A$, \B\ must be \U$_k$.
Thus, after $s$ updates \co{A->right} to \B,
we have $\A \leftrightarrow \B$, or $\U_{k-1} \leftrightarrow \U_k$, satisfying claim \ref{done}.

This step does not affect the other claims of the invariant.

\item
\label{append-proof2}
Suppose $s$ is a successful CAS at line \ref{mr-cas2}  in \co{tryAppend(B,C)}.
Then, $s$ changes \co{B->right} from \co{null} to \C.
Prior to $s$, \co{Head} was changed from \B\ to \C\ at line \ref{head-cas}, 
so \B\ and \C are active at $s$.

By Observation \ref{head-increases}, there is no ABA problem on the \co{Head} pointer.
So the \co{tryAppend} that performs $s$ is the only \co{tryAppend(B,*)} that could return true.
By Assumption \ref{list-assumptions}.\ref{append-before-remove}, no \co{remove(B)} has been invoked before $s$.
By Lemma \ref{remove-safe}, \B\ is not finalized when $s$ occurs.

Since $\co{Head} >_c \B$ when $s$ occurs and \co{Head} is not finalized, by Lemma \ref{head-unfinalized},
\B\ is not the last active, unfinalized node with respect to $<_c$.
In the configuration before $s$, $\B\rightarrow \co{null}$, so the induction hypothesis implies
that $\B = \U_{k-1}$ and case \ref{exception} of induction hypothesis \ref{well-formed} is satisfied.
Thus, some \co{tryAppend(B,U$_k$)} has performed line \ref{head-cas}, but no 
\co{tryAppend(U$_k$,*)} has made another node greater then $\U_k$ active.  
Since \C\ is an active node with $\C >_c \B$, we must have $\C = \U_k$.
Thus, after $s$ updates \co{B->right} to \C,
we have $\B\leftrightarrow \C$, or $\U_{k-1} \leftrightarrow \U_k$, satisfying claim \ref{done}.

This step does not affect the other claims of the invariant.

\item
\label{invoke-splice-proof}
Suppose $s$ is an invocation of \co{splice(X,Y,Z)} or a CAS at line \ref{store-right-desc} or \ref{store-left-desc} that successfully stores a Descriptor containing \co{(X,Y,Z)}.
This step does not affect the truth of claim
\ref{well-formed}, \ref{well-formed-left}, \ref{well-formed-right} or \ref{no-skip}.

We first prove that $s$ preserves claim \ref{no-overlap}.

If $s$ is an invocation of \spl\ at line \ref{help-splice} of \co{help(desc)},
it gets its arguments from the Descriptor \co{desc}.
Prior to the call to \co{help(desc)} at line \ref{remove-help}, \ref{sur-help-old}, \ref{sur-help-new},
\ref{sul-help-old} or \ref{sur-help-new}, \co{desc} was in the \co{rightDesc} or \co{leftDesc} field
of a node at line \ref{read-desc1}, \ref{read-desc2}, \ref{store-right-desc}, \ref{read-desc3} or \ref{store-left-desc}, respectively.
Thus, Claim \ref{no-overlap} follows from induction hypothesis \ref{no-overlap}.

So it remains to show that splice triples that arise from calls to splice on 
line \ref{call-splice} and from Descriptors satisfy the claim.
To derive a contradiction, assume there are two such overlapping triples \co{(W,X,Y)} and \co{(X,Y,Z)}.
We consider all possible cases and show that each leads to a contradiction.

\begin{enumerate}[{C{a}se i}]
\item
Suppose both \co{splice(W,X,Y)} and \co{splice(X,Y,Z)} have been called from line \ref{call-splice}.
Before \co{splice(W,X,Y)}, the test on line \ref{test-priority} ensures that $\X >_p \Y$.
Similarly if \co{splice(X,Y,Z)} is called, then $\X <_p \Y$.  A contradiction.

\item
\label{sul-sul}
Suppose there are two Descriptors \co{descW} containing \co{(W,X,Y)} and \co{descX} containing \co{(X,Y,Z)} that have each been stored in a node.
We consider the case where $\X >_p \Y$.
The case where $\X <_p \Y$ is symmetric (and \X\ and \Y\ cannot have the same priority, by Lemma \ref{desc-properties}).
\Eric{The symmetric proof is actually done in detail in the .tex file inside an ignore statement, just to be sure it works.}
By Lemma~\ref{desc-properties}, \co{descW} and \co{descX} were stored in the
\co{rightDesc} fields of \W\ and \X\ by calls to \co{\sul(W,X,Y)} and \co{\sul(X,Y,Z)}, respectively.
Let these calls be $sul_{\X}$ and $sul_{\Y}$, and let $rr_{\X}$ be the call to \co{removeRec(X)} that invokes $sul_{\X}$.
By Lemma \ref{frozen-before-rr}, \co{descX} is removed
from \co{X->rightDesc} at some point before $rr_{\X}$ performs line \ref{read-right}.
By induction hypothesis \ref{descriptors-work}, this means that, at all times during $rr_{\X}$, $\X\not\rightarrow \Y$.
This contradicts the fact that $rr_{\X}$ read \Y\ from \co{X->right} at line~\ref{read-right}.

\ignore{
For completeness, here is the symmetric version of the proof above for the case where $\X <_p \Y$.
Suppose there are two descriptors \co{descY} containing \co{(W,X,Y)} and \co{descZ} containing \co{(X,Y,Z)}.

By Lemma \ref{desc-properties}, \co{descY} and \co{descZ} were stored in the
\co{leftDesc} fields of \Y\ and \Z\ by calls to \co{\sur(W,X,Y)} and \co{\sur(X,Y,Z)}, respectively.
Let these calls be $sul_{\X}$ and $sul_{\Y}$, and let $rr_{\Y}$ be the call to \co{removeRec(Y)} that invokes $sul_{\Y}$.

By Lemma \ref{frozen-before-rr}, \co{descY} is removed
from \co{Y->leftDesc} at some point before $rr_{\Y}$ performs line \ref{read-left}.
By induction hypothesis \ref{descriptors-work}, this means that, at all times during $rr_{\Y}$, $\X\not\leftarrow \Y$.
This contradicts the fact that $rr_{\Y}$ reads \X\ from \co{Y->left} at line \ref{read-left}.
}

\item
\label{regular-sul}
Suppose \co{splice(W,X,Y)} has been called from line \ref{call-splice}
and a Descriptor \co{desc} containing \co{(X,Y,Z)} has been stored in a node.
Let $rr$ be the \co{removeRec(X)} operation that calls \co{splice(W,X,Y)} at line \ref{call-splice}.
By the test at line \ref{test-priority}, $\X >_p \Y$.
By Lemma~\ref{desc-properties}, \co{desc} must have been stored in
\co{X->rightDesc}.
By Lemma \ref{frozen-before-rr}, \co{desc} is removed from \co{X->rightDesc} before $rr$
is called.
By induction hypothesis \ref{descriptors-work}, this means that, at all times during $rr$, $\X\not\rightarrow \Y$.
This contradicts the fact that $rr$ read \Y\ from \co{X->right} at line \ref{read-right}.

\item
Suppose a Descriptor \co{desc} containing \co{(W,X,Y)} has been stored in a node and
\co{splice(X,Y,Z)} is called from line \ref{call-splice}.  The proof is symmetric to Case \ref{regular-sul}.
\Eric{The symmetric proof is actually done in detail in the .tex file inside an ignore statement, just to be sure it works.}

\ignore{
Let $rr$ be the \co{removeRec(Y)} operation that calls \co{splice(X,Y,Z)} at line \ref{call-splice}.
By the test at line \ref{test-priority}, $\X <_p \Y$.
By Lemma \ref{desc-properties}, \co{desc} must have been stored in
\co{Y->leftDesc}.
By Lemma \ref{frozen-before-rr}, \co{desc} is removed from \co{Y->leftDesc} before $rr$
is called.
By induction hypothesis \ref{descriptors-work}, this means that, at all times during $rr$, $\X\not\leftarrow \Y$.
This contradicts the fact that $rr$ reads \X\ from \co{Y->left} at line \ref{read-left}.
}
\end{enumerate}
This completes the proof of Claim \ref{no-overlap}.

Next, we verify that $\X\leftarrow \Y\rightarrow \Z$ to ensure claim \ref{splice-args} is satisfied.
We show that $\X\leftarrow\Y$ in the configuration after $s$; the proof that $\Y\rightarrow\Z$ is symmetric.
By Lemma \ref{splice-pre}, there is a time before $s$ when $\X\leftarrow \Y$.
To derive a contradiction, suppose there is some step $s'$ that changed \co{Y->left} from \X\ to some other value \W\ before $s$.
Step $s'$ must be an execution of line \ref{left-cas} in a \co{splice(W,X,Y)} operation, which
contradicts Claim \ref{no-overlap}, proved above.

If $s$ is a CAS at line \ref{store-right-desc} or \ref{store-left-desc}, we verify that it preserves
claim \ref{descriptors-work} in Case \ref{remove-desc}, below.

\item
\label{fail-left-proof}
Suppose $s$ is a failed CAS on line \ref{left-cas} inside a \co{splice(A,B,C)}.
This step is a failed attempt to change \co{C->left} from \B\ to \A.
If this is not the first time line \ref{left-cas} is performed by a \co{splice(A,B,C)}, then there is nothing to check.
If it is the first, then we must verify that \A, \B\ and \C\ do not satisfy case \ref{not-done} of claim
\ref{well-formed} or case \ref{not-done-leftend} of claim \ref{well-formed-left}.
But this is obvious, since if they did, the CAS would succeed.

\item
\label{fail-right-proof}
Suppose $s$ is a failed CAS on line \ref{right-cas} inside a \co{splice(A,B,C)}.
This step is a failed attempt to change \co{A->right} from \B\ to \C.
If this is not the first time line \ref{right-cas} is performed by a \co{splice(A,B,C)}, then there is nothing to check.
If it is the first, then we must verify that \A, \B\ and \C\ do not satisfy case \ref{left-done} of claim
\ref{well-formed} or case \ref{not-done-rightend} of claim \ref{well-formed-right}.
But this is obvious, since if they did, the CAS would succeed.

\item
\label{fail-mrcas}
Suppose $s$ is an execution of the CAS on line \ref{mr-cas1} or \ref{mr-cas2} of \co{tryAppend} that fails.
We must show that $s$ does not cause claim \ref{exception} to become false.
Suppose step $s$ attempts to CAS \co{X->right} from \co{null} to a node \Y.
The CAS fails, so $\X \not\rightarrow \co{null}$ in the configuration before $s$.
Thus, \X\ cannot be $\U_{k-1}$ in claim \ref{exception}, so $s$ cannot cause that claim to become false.

\item
\label{remove-desc}
Suppose $s$ is a successful CAS on line \ref{freeze-desc}, \ref{store-right-desc} or \ref{store-left-desc}
that removes a Descriptor \co{desc} containing \co{(X,Y,Z)} from the \co{leftDesc} or \co{rightDesc} field of a node.
We must show that claim \ref{descriptors-work} holds in the configuration after $s$.

To prove claim \ref{descriptors-work-right}, suppose that $\X \neq \co{null}$.
Before \co{desc} was created at line \ref{create-right-desc} of \co{\sul(X,Y,Z)} or line \ref{create-left-desc} of \co{\sur(X,Y,Z)},
we had $\X\rightarrow \Y$ at line \ref{sul-test-right} or \ref{sur-test}.
We show that, at some time between the creation of \co{desc} and $s$, $\X \not\rightarrow \Y$.
Prior to $s$, \co{help(desc)} is called at line \ref{remove-help}, \ref{sul-help-old} or \ref{sur-help-old}, respectively.
That call to \co{help} completes a \co{splice(X,Y,Z)} operation.
If the \co{splice} terminates on line \ref{test-right}, then $\X\not\rightarrow \Y$.
Otherwise, the \co{splice} performs line \ref{right-cas}.
Either $\X\not\rightarrow \Y$ before that CAS,
or $\X\rightarrow\Z$ (where $\Z \neq \Y$) after the CAS.
In all cases, $\X\not\rightarrow \Y$ at some point during the \co{help(desc)}, prior to the CAS that removes \co{desc} from \co{X->rightDesc}.
Thus, \co{X->right} has changed from \Y\ to some other value before $s$.
It follows from Lemma \ref{order} that $\co{X->right} >_c \Y$ in the configuration after $s$.

To prove claim \ref{descriptors-work-left}, suppose that $\Z \neq \co{null}$.
The Descriptor \co{desc} was created at line \ref{create-right-desc} 
of \co{\sul(X,Y,Z)} or line \ref{create-left-desc} of \co{\sur(X,Y,Z)},
each called by an instance $rr_{\Y}$ of \co{removeRec(Y)}.
Consider the configuration $C'$ before $rr_{\Y}$ executes line \ref{read-right}, 
when $\Y \rightarrow \Z$.
We consider two cases.

If \Y\ is finalized in $C'$, then it was finalized by a \co{splice($\X'$,Y,$\Z'$)}.
Since, in the configuration before $s$, there has been a call to \co{splice($\X'$,Y,$\Z'$)} and \co{desc} containing \co{(X,Y,Z)} has been installed in a node, we have $\X' = \X$ and $\Z' = \Z$ by induction hypothesis \ref{splice-args}.
By Lemma \ref{desc-properties}, the \co{splice} was not called at line \ref{call-splice}
since either $\X <_p \Y <_p \Z$ or $\X >_p \Y >_p \Z$.
So, the \co{splice} was called by \co{help(desc$'$)} for some Descriptor $\co{desc}'$, also containing \co{(X,Y,Z)},
that was installed in the same node as \co{desc} prior to the creation of \co{desc}.
Thus, $\co{desc}'$ has been removed from a node before~$s$.
By induction hypothesis \ref{descriptors-work}, the claim is true in the configuration before $s$, and therefore
also in the configuration after $s$.

If \Y\ is not finalized in $C'$, then by induction hypotheses \ref{well-formed} and \ref{well-formed-right},
we have $\Y\leftarrow \Z$ in $C'$, when $rr_{\Y}$ executes line \ref{read-right}.
Later, at line \ref{sul-test-right} or \ref{sur-test} before \co{desc} is created,
we have that either $\X =\co{null}$ or $\X\rightarrow \Y$.
Prior to $s$, \co{help(desc)} is called at line \ref{remove-help}, \ref{sul-help-old} or \ref{sur-help-old}, respectively.
That call to \co{help} completes a \co{splice(X,Y,Z)} operation.
If it performs line \ref{left-cas}, then $\Y\not\leftarrow \Z$ after that CAS.
Otherwise, it terminated at line \ref{test-right}.
This means that $\X\neq \co{null}$ and $\X\not\rightarrow \Y$.
We know that, in an earlier configuration, $\X\rightarrow \Y$, so \co{X->right} has changed
from \Y\ to another value.
This could only have been done by line \ref{right-cas} of a \co{splice(X,Y,Z$'$)} operation.
By induction hypothesis \ref{splice-args}, $\Z=\Z'$.
By Lemma \ref{desc-properties}, this \co{splice} could not have been called at line \ref{call-splice},
so it was called by \co{help(desc$'$)}.
This \co{splice} also performed a \co{CAS(Z->left,Y,X)} at line \ref{left-cas}.
After this step, $\Y\not\leftarrow \Z$.
If this CAS was performed after $C'$ (when $\Y\leftarrow \Z$), then it follows from Lemma \ref{order}
that $\co{Z->left} <_c \Y$ at all times after the CAS.
If this CAS was performed before $C'$, then $\co{desc}'\neq \co{desc}$ (since \co{desc} had 
not yet been created at~$C'$) and $\co{desc}'$, which also contains \co{(X,Y,Z)} 
has been removed from a node before~$s$.
In this case, the claim follows from the induction hypothesis.

\item
An execution of line \ref{mark} cannot change the \co{status} field of a node
from \co{finalized} to \co{marked}, by Lemma \ref{finalize-works}.  Hence, this step cannot
cause any of the claims to become false. \qedhere
\end{enumerate}
\end{proof}

\listcorrectness*

\begin{proof}
We choose linearization points as follows.
\begin{itemize}
\item
Each \co{getHead} is linearized when it reads \co{Head}.
\item
Each \co{tryAppend} is linearized when it performs the CAS on \co{Head} at line \ref{head-cas}.
\item
Each \co{remove} and \co{find} can be linearized at any time during its execution.
\end{itemize}
Results returned by \co{getHead} and \co{tryAppend} are clearly consistent with this
choice of linearization points.  The \co{remove} operation does not return any result.
So, it remains to show that \co{find} operations return results consistent with the linearization.

Recall that $\V_1, \V_2, \ldots$ are the active nodes in the order they are appended to the list,
and $\ts_1, \ts_2, \ldots$ are their timestamps.
Consider a \co{find($\V_j$,t)} operation.
By assumption \ref{list-assumptions}.\ref{append-pre}, $V_j$ is active, so the local variable \co{cur} is initially set
to an active node.  Since \co{cur} is updated by following \co{left} pointers, \co{cur}
is always either \co{null} or an active node, by Lemma \ref{order}.

We next prove the invariant that if $\co{cur} = \V_i$ and $i<j$, then $\ts_{i+1} > \co{t}$.
When \co{cur} is initialized at line \ref{find-init}, $i=j$, so the claim holds.
Suppose \co{cur} is updated from node $\V_i$ to node $\V_{i'}$ at line \ref{find-advance} 
(because $\ts_i > \co{t}$).
Since $\V_{i'}$ is read from \co{$\V_i$->left} we have $i'<i$, by Lemma \ref{order}.
If $i'+1 = i$, then we know that $\ts_{i'+1} > \co{t}$.
If $i' < i-1$, then each of the nodes $\V_k$ with $i'<k<i$ are finalized, by Invariant \ref{inv}.
By Corollary \ref{remove-safe}, \co{remove} has been called on each of these nodes $\V_k$.
By Assumption \ref{list-assumptions}.\ref{removals-safe}, this means that $\co{t} \notin [\ts_k, \ts_{k+1})$ for each such $k$.
So, $\co{t} \notin [\ts_{i'+1},\ts_i)$ and we know that $\ts_i > \co{t}$.
Hence, $\co{t} < \ts_{i'+1}$, as required.

It follows that if the \co{find($\V_j$,t)} returns a node $\V_i$ where $i<j$, then $\V_i$
is the rightmost node ever appended to the list that is to the left of $\V_j$ and has $\ts_i \leq \co{t}$.
Moreover, by assumption \ref{list-assumptions}.\ref{append-pre}, no \co{remove($\V_i$)} has been invoked prior to the \co{find}'s linearization point.
Thus, the value returned by the \co{find} is consistent with its linearization point.

If the \co{find($\V_j$,t)} returns $\V_j$, then $\ts_j \leq \co{t}$ by the test on line \ref{find-read}.  So, it suffices to show that no \co{remove($\V_j$)} is linearized before the \co{find}.
If, at the linearization point of the \co{find}, $\V_j$ is the last node that has been appended 
to the list then Assumption \ref{list-assumptions}.\ref{remove-pre} ensures that 
\co{remove($\V_j$)} has not been invoked. 
If a successful \co{tryAppend($\V_j, \V_{j+1}$)} has been linearized before the \co{find},
then $\ts_{j+1} > t$ by Assumption \ref{list-assumptions}.\ref{find-append}.
By Assumption \ref{list-assumptions}.\ref{removals-safe}, this implies that \co{remove($\V_j$)}
has not been invoked before the linearization point of the \co{find}.

Finally, we consider the case where \co{find} returns \co{null}.
Consider the last execution of line \ref{find-advance} that changes \co{cur} from $\V_i$ to \co{null}.
Then $\ts_i > \co{t}$ and $\V_i\co{->left} = \co{null}$.
By Invariant \ref{inv}, for all $k<i$, $\V_k$ is finalized.
By Corollary \ref{remove-safe}, \co{remove} has been called on each of these nodes $\V_k$.
By Assumption \ref{list-assumptions}.\ref{removals-safe}, this means that $\co{t} \notin [\ts_k, \ts_{k+1})$ for each such $k$.
So, $\co{t} \notin [\ts_1,\ts_i)$.
Since we know $\co{t} < \ts_i$, it follows that $\co{t} < \ts_1$, so no node has ever been appended
to the list with a timestamp less than or equal to \co{t}.
Thus, it is consistent with the linearization to return the result \co{null}.
\end{proof}


\er{The freedom to linearize \co{remove} and \co{find} operations at any time during their execution
results from Assumption \ref{list-assumptions}.\ref{removals-safe}, which essentially
ensures there are no dangerous races between \co{remove} and \co{find} operations.}

\future{Perhaps add a remark that if you want to weaken assumption about not removing a node
when a pending or future find is looking for it, you could modify the code to check if a node is marked.}

\subsection{Amortized Time Bounds for Version Lists}


The \co{tryAppend} operation performs $O(1)$ steps.  Some clever tricks from \cite{And05} can be used to compute the function \co{p} using $O(1)$ machine instructions.

We consider a \co{tryAppend} operation to be {\it successful} if its CAS at line \ref{head-cas}
successfully updates the \co{Head} pointer.
By Theorem \ref{remove-wait-free}, the worst-case number of steps performed by a \co{remove} operation
in an execution with $L$ successful \co{tryAppend} operations is $O(\log L)$.
We now prove that the amortized number of steps per \co{remove} operation is constant.
To prove this, we first prove a lemma that implies at most one call to \co{removeRec(Y)} can recurse.
(Interestingly, this lemma will also be useful later in bounding the amount of space used.)

\begin{lemma}
\label{bound-recursion}
Among all calls to \co{splice(*,Y,*)} at line \ref{call-splice}, 
\co{\sul(*,Y,*)} and \co{\sur(*,Y,*)} in an execution, 
at most one performs a successful CAS at line \ref{finalize-cas}, \ref{store-right-desc}
or \ref{store-left-desc}, and hence at most one can return true.
\end{lemma}
\begin{proof}
It follows from Invariant \ref{inv}.\ref{splice-args} that any two such calls that perform 
a successful CAS at line \ref{finalize-cas}, \ref{store-right-desc}
or \ref{store-left-desc} must have the same triple of arguments \co{(X,Y,Z)}.
If \co{splice(X,Y,Z)} is called at line \ref{call-splice}, then $\X <_p \Y >_p \Z$, by the test at
line \ref{test-priority}.
It follows from Lemma \ref{desc-properties}, that there cannot be calls to two different routines that both return true.
That is, the successful calls are all to \co{splice}, all to \sul, or all to \sur.

By Lemma \ref{finalize-works}, only the first \co{splice(X,Y,Z)} to perform the CAS at line
\ref{finalize-cas} can succeed in finalizing \Y.

To derive a contradiction, suppose two calls to \co{\sul(X,Y,Z)} perform a successful CAS at line \ref{store-right-desc}.
Let these two calls be $sul_1$ and $sul_2$, in the order that they perform their successful CAS.
Let $\co{desc}$ be the Descriptor installed in \co{X->rightDesc} by $sul_1$.
We consider two cases.

First, suppose \co{desc} is replaced by another Descriptor before
$sul_2$ reads \co{X->rightDesc} at line \ref{read-desc2}.
By Invariant \ref{inv}.\ref{descriptors-work},
$\X \not\rightarrow \Y$ when $sul_2$ performs line \ref{sul-test-right} so $sul_2$ cannot
perform the CAS at line \ref{store-right-desc}, a contradiction.

Now, suppose \co{desc} is still in \co{X->rightDesc} when $sul_2$ reads that 
field at line \ref{read-desc2}.
Then, $sul_2$ calls \co{help(desc)} at line \ref{sul-help-old}, which calls \co{splice(X,Y,Z)}.
Prior to this call to \co{splice}, we had $\X\rightarrow \Y$ when $sul_1$ performed 
line \ref{sul-test-right}.
Then, during $sul_2$'s call to \co{splice}, we have $\X\not\rightarrow\Y$ at line \ref{test-right}
or before or after the CAS at line \ref{right-cas} is performed.
By Lemma \ref{order}, $\X\not\rightarrow \Y$ at all times after this \co{splice} returns.
Thus, when $sul_2$ performs the test at line \ref{sul-test-right}, it returns false
and does not perform the CAS at line \ref{store-right-desc}, a contradiction.

The proof that no two calls to \co{\sur(X,Y,Z)} can perform a successful CAS at 
line \ref{store-left-desc} is symmetric to the proof for \sul.  
\Eric{The symmetric case is actually done in the .tex file inside an ignore environment.}
\ignore{
To derive a contradiction, suppose two calls to \co{\sur(X,Y,Z)} perform a successful CAS at line \ref{store-left-desc}.
Let these two calls be $sul_1$ and $sul_2$, in the order that they perform their successful CAS.
Let $\co{desc}$ be the Descriptor installed in \co{X->rightDesc} by $sul_1$.
We consider two cases.

First, suppose \co{desc} is replaced by another Descriptor before
$sul_2$ reads \co{Z->leftDesc} at line \ref{read-desc3}.
By Invariant \ref{inv}.\ref{descriptors-work},
$\Y \not\leftarrow \Z$ when $sul_2$ performs line \ref{sur-test} so $sul_2$ cannot
perform the CAS at line \ref{store-right-desc}, a contradiction.

Now, suppose \co{desc} is still in \co{Z->leftDesc} when $sul_2$ reads that 
field at line \ref{read-desc3}.
Then, $sul_2$ calls \co{help(desc)} at line \ref{sur-help-old}, which calls \co{splice(X,Y,Z)}.
Prior to this call to \co{splice}, we had $\Y\leftarrow \Z$ when $sul_1$ performed 
line \ref{sur-test}.
Then, during $sul_2$'s call to \co{splice}, we have $\Y\not\leftarrow\Z$ 
before or after the CAS at line \ref{left-cas} is performed.
By Lemma \ref{order}, $\Y\not\leftarrow \Z$ at all times after this \co{splice} returns.
Thus, when $sul_2$ performs the test at line \ref{sur-test}, it returns false
and does not perform the CAS at line \ref{store-left-desc}, a contradiction.
}
\end{proof}

\removeAmortized* 
\begin{proof}
As argued in the proof of Theorem \ref{remove-wait-free}, the total number of steps performed
by \co{remove} operations is proportional to the number of calls to \co{remove} and \co{removeRec}.
So, it suffices to show that there are at most $O(R)$ calls to \co{removeRec}.

There are at most $R$ calls to \co{removeRec} at line \ref{call-rr1} of the \co{remove} operation.
Each recursive call to \co{removeRec} is preceded by a (distinct) call to \co{splice} at 
line \ref{call-splice}, \sul\ or \sur\ that returns true.
Thus, it suffices to show that at most $O(R)$ such calls can return true.
By Lemma \ref{splice-marked}, each call to \co{splice(*,Y,*)}, \co{\sul(*,Y,*)} or \co{\sur(*,Y,*)}
is on a marked node \Y.  Since there are at most $R$ marked nodes, the claim follows from 
Lemma \ref{bound-recursion}.
\end{proof}

\subsection{Space Bound for Version Lists}
In this section, we prove Theorem \ref{space-bound}, which bounds the number of nodes that remain in the list.
The proof of this theorem requires several lemmas.
\er{
We say that a node \X\ is {\it lr-reachable} in a configuration if it can be reached from an active, unfinalized node by following a sequence of \co{left} and \co{right} pointers.
}

Consider a finite execution.  Define the priority tree $T$ for the execution as follows.
Consider the sequence of all nodes that become active during the execution, with the
priorities assigned to them by the algorithm.
Arrange the sequence into a priority tree by choosing the minimum priority node
as the root and recursively constructing the left and right subtrees from the subsequences
to the left and right of this node.
As remarked earlier, our scheme for assigning priorities ensures that the choice of the
root in this recursive construction is always unique.
Furthermore, the sequence of active nodes ordered by $<_c$ is an in-order traversal of $T$.
See Figure \ref{priority-tree} for an example.
$T$ is a single tree defined for the entire execution.
Whenever we refer to ancestors or descendants of nodes, we are referring to relationships
in the tree $T$.  By convention, we consider a node to be its own ancestor and descendant.

We say that a node is {\it removable} in a configuration if \co{remove} has been invoked on it before that configuration.
As the execution proceeds, certain removable nodes within $T$ will become marked and finalized.
At any given configuration during the execution, we wish to bound the number of nodes
in the tree that are removable but not finalized.

\begin{lemma}
\label{splice-priorities}
If \co{splice(X,Y,Z)} is called, then $\X <_p \Y >_p \Z$ or $\X <_p \Y <_p \Z$ or $\X >_p \Y >_p \Z$.
\end{lemma}
\begin{proof}
If \co{splice(X,Y,Z)} is called on line \ref{call-splice}, then $\X <_p \Y >_p \Z$ by the test
at line \ref{test-priority}.
Otherwise, \co{splice(X,Y,Z)} is called at line \ref{help-splice} of a call to \co{help(desc)} on some Descriptor \co{desc} containing \X,\Y,\Z.
Since \co{desc} was created by a call to either \co{\sul(X,Y,Z)} or \co{\sur(X,Y,Z)}, we must have
either $\X >_p \Y >_p \Z$ or $\X <_p \Y <_p \Z$, by the tests preceding those calls at line \ref{call-sul}
and \ref{call-sur}.
\end{proof}

\begin{lemma}
\label{anc-or-desc}
In any configuration, if \X\ is an active node then \co{X->left} and \co{X->right} each point to a proper ancestor of \X\ in $T$, a proper descendant of \X\ in $T$, or \co{null}.
\end{lemma}
\begin{proof}
Suppose the claim is false.  Consider the first time it is violated.  We consider all steps that can change \co{left} or \co{right} pointers.

Suppose the violation is caused by a CAS at line \ref{left-cas} or \ref{right-cas} of some \co{splice(A,B,C)} operation.  This CAS either updates \co{A->right} to \C, or \co{C->left} to \A.
Since this causes a violation of the lemma,
\A\ and \C\ must be non-\co{null} and neither \A\ nor \C\ is an ancestor of the other.
When the CAS occurs, $\A\leftarrow \B \rightarrow \C$, by Invariant \ref{inv}.\ref{splice-args}.
Since no violation of the lemma occurred before this step, one of
\A\ and \B\ must be a proper ancestor of the other, and one of \B\ and \C\ must be a proper ancestor of the other.
The only way this can happen is if \B\ is a proper ancestor of both \A\ and \C.
Thus, $\A >_p \B <_p \C$.
This contradicts Lemma~\ref{splice-priorities}.

Suppose the violation is caused by line \ref{mr-cas1} which sets $\co{A->right}$ to \B{}.
Since we have $\A\leftarrow \B$ earlier at line \ref{mr-read-left}, this would violate
the assumption that no violation occurred before this step.

The violation cannot be caused by line \ref{set-left} or line \ref{mr-cas2} because
they set a \co{left} or \co{right} pointer to an adjacent node in the $<_c$ ordering,
and two adjacent nodes in an in-order traversal of $T$ must be ancestors or descendants of each other.
\end{proof}

\Eric{Following corol doesn't seem to be used in step 2 proof.  Cut it out if it isn't needed in Step 3}
\begin{corollary}
\label{recurse-on-anc}
If a call to \co{removeRec(X)} calls \co{removeRec(Y)} recursively, then \Y\ is an ancestor of \X.
\end{corollary}
\begin{proof}
\Y\ is read from \co{X->left} or \co{X->right}.  
The \co{validAndFrozen} test before any recursive call ensures that \Y\ is not \co{null}.
Thus, by Lemma \ref{anc-or-desc}, \Y\ is either an ancestor or descendant of \X.  
Since \co{removeRec(Y)} is called only after testing that $\Y <_p \X$, \Y\ must be an ancestor of~\X.
\end{proof}

Incidentally, it also follows from Lemma \ref{anc-or-desc} that we never compare the priorities of two nodes
with the same priority during the \co{removeRec} operation.

\begin{lemma}
\label{one-subtree-fin}
When a node is finalized, either all nodes in its left subtree in $T$ or all nodes in its right subtree in $T$ are finalized.
\end{lemma}
\begin{proof}
Consider the call to \co{splice(A,B,C)} that finalizes node \B.
By Lemma \ref{splice-priorities}, either $\A <_p \B$ or $\C <_p \B$.
We prove the lemma for the case where $\A <_p \B$; the other case is symmetric.
By Lemma \ref{splice-pre}, there was an earlier configuration $C$ when $\A \leftarrow \B$.
We consider two cases.

If \A\ is \co{null}, then by Invariant \ref{inv} in the configuration when $\co{null} \leftarrow \B$,
there are no non-finalized nodes \D\ such that $\D <_c \B$.
In other words, \B\ is the first non-finalized node in an in-order traversal of $T$.
So, all nodes in the left subtree of
\B\ in $T$ are finalized.

Now, suppose \A\ is not \co{null}.
By Lemma \ref{anc-or-desc}, \B\ is a descendant of \A\ in $T$ since $\A <_p \B$.
By Lemma \ref{order}, $\A <_c \B$, so \B\ is in the right subtree of \A.
By Invariant \ref{inv}, there is no non-finalized node \D\ such that $\A <_c \D <_c \B$ in configuration $C$.
In other words, there is no non-finalized node between \A\ and \B\ in an in-order traversal
of $T$.  Thus, every node in \B's left subtree in $T$ is finalized in $C$.
\end{proof}

We want to bound the number of unfinalized nodes.
We divide them into three types as follows.
\begin{itemize}
\item
Type 0:  neither of the node's two subtrees contains an unfinalized node.
\item
Type 1:  exactly one of the node's two subtrees contains an unfinalized node.
\item
Type 2:  both of the node's two subtrees contain an unfinalized node.
\end{itemize}

The following lemma shows that it will be sufficient to bound the number of unfinalized nodes of Types 0 and 1.

\begin{lemma}
\label{type-2-bound}
For any configuration and any subtree of $T$ that contains unfinalized nodes,
the number of unfinalized nodes of Type 2 in the subtree is less than the number of unfinalized nodes of Type 0 in the subtree.
\end{lemma}
\begin{proof}
We prove the claim by induction on the height of the subtree.

Base case:  Consider a subtree of $T$ consisting of a single leaf node.
If the node is finalized, the claim holds trivially.
If the node is not finalized, then it is of type 0, so the claim holds.

Induction step:  Assume the claim holds for subtrees of height less than $h$.
Consider a subtree of height $h$ with root $r$.  We consider several cases.

Case 1:  Suppose $r$ is finalized.
By Lemma \ref{one-subtree-fin}, either the left or right subtree of $r$ contains only finalized nodes.
The claim follows by applying the induction hypothesis to the other subtree.

Case 2:  Suppose $r$ is unfinalized and of Type 0.  Then $r$ is the
only unfinalized node in the subtree, so the number of unfinalized nodes of Type 2 is 0.

Case 3:  Suppose $r$ is unfinalized and of Type 1.
Then, either the left or right subtree $r$ contains only finalized nodes.
The claim follows by applying the induction hypothesis to the other subtree (which does contain unfinalized nodes).

Case 4:  Suppose $r$ is unfinalized and of Type 2.
Let $T_L$ and $T_R$ be the left and right subtrees of $r$.
By definition, both contain unfinalized nodes.  So by the induction hypothesis,
\begin{eqnarray*}
&&\mbox{\# Type-2 unfinalized nodes in the subtree}\\
&=& 1 + (\mbox{\# Type-2 unfinalized nodes in $T_L$}) + (\mbox{\# Type-2 unfinalized nodes in $T_R$})\\
&\leq& 1 + (\mbox{\# Type-0 unfinalized nodes in $T_L$}) -1 +  (\mbox{\# Type-0 unfinalized nodes in $T_R$}) -1\\
&=& (\mbox{\# Type-0 unfinalized nodes in subtree})-1\\
&<& \mbox{\# Type-0 unfinalized nodes in subtree}.\hfill \qedhere 
\end{eqnarray*}
\end{proof}

\Eric{Following lemma seems useful only for step 3, so might move it later.}
\begin{lemma}
\label{unfinalized-reachable}
In each configuration, every active, unfinalized node is lr-reachable.
\end{lemma}
\begin{proof}
Assume the configuration has some active, unfinalized nodes.
Let $\U_1, \ldots, \U_k$ be the sequence of active, unfinalized nodes in the configuration, ordered by $<_c$.

We first show that $\U_k$ is reachable.
We define a sequence of nodes $\W_1, \ldots, \W_\ell$ inductively.
Let $\W_1$ be the value of \co{Head}.
If $\W_i$ is not finalized or if $\co{null} \leftarrow \W_i$, then $\ell=i$ and $\W_i$ is the last node in the sequence.
Otherwise, $\W_i$ is finalized and we define $\W_{i+1} = \W_i\co{->left}$.
We show that $\U_k=\W_\ell$.

We prove by induction on $i$ that $\U_k\leq_c \W_i$.
For $i=1$, ince $\U_k$ is active, $\U_k \leq_c \W_1$, by Observation \ref{head-increases}.
For $1\leq i < \ell$ suppose $\U_k \leq_c \W_i$ to prove that $\U_k \leq_c \W_{i+1}$.
Since $i<\ell$, $\W_i$ is finalized, so $\U_k \neq \W_i$.
Thus, $\U_k <_c \W_i$.
It follows from Invariant \ref{inv}.\ref{no-skip-left} that $\U_k \leq_c \W_{i+1}$.

Thus, we have $\U_k\leq_c W_\ell$.  
If $\W_\ell$ were finalized, then $\co{null}\leftarrow W_\ell$, which would contradict Invariant \ref{inv}.\ref{no-skip-left}.
Thus, $\W_\ell$ is not finalized, and we have $\U_k=\W_\ell$
since $\U_k$ is the largest unfinalized node with respect to $<_c$.
Thus, $\U_k$ is lr-reachable via the sequence of nodes $\W_1, \ldots, \W_\ell = \U_k$.

By Invariant \ref{inv}.\ref{well-formed}, if $\U_{i+1}$ is lr-reachable, then so is $\U_i$.
The claim follows by induction.
\end{proof}

\Eric{Maybe move this lemma closer to where it is used}
\begin{lemma}
\label{remains-unreachable}
If an active node \X\ ceases being lr-reachable, it never subsequently becomes lr-reachable again.
\end{lemma}
\begin{proof}
It follows from Invariant \ref{inv}.\ref{splice-args} that when a
node's \co{left} or \co{right} pointer is modified, the node is not finalized.
By Invariant \ref{inv}, a \co{left} or \co{right} pointer can only be changed to an unfinalized node or \co{null}.
(For example, Invariant \ref{inv}.\ref{well-formed} implies that \co{$\U_i$->right} can only be changed to  $\U_{i+1}$ or \co{null}.)
If the new value of the pointer is not \co{null}, it is an unfinalized node, which was already lr-reachable by definition, so this change does not cause any node to become lr-reachable.
\end{proof}

\begin{lemma}
\label{full-splice-works}
After an instance of \co{splice(*,Y,*)} executes line \ref{right-cas}, \Y\ is not lr-reachable.
\end{lemma}
\begin{proof}
Let $\U_i$ be the largest unfinalized node that is less than \Y\ (with respect to $<_c$), or \co{null}
if there is no such node.
Similarly, let $\U_{i+1}$ be the smallest unfinalized node that is greater than \Y, or \co{null} if there is no such node.
By Invariant \ref{inv}, if an unfinalized node points
to a finalized node, the finalized node points only to unfinalized ones.
Thus if there is a path of \co{left} or \co{right} pointers from an unfinalized node
to a finalized node, there cannot be two finalized nodes in a row on the path.
So, if \Y\ is lr-reachable, there is some unfinalized node that points directly to it.
Thus, by Invariant \ref{inv}, \Y\ can only be lr-reachable if either $\U_i\rightarrow \Y$ or $\Y\leftarrow \U_{i+1}$.

Since the \co{splice(*,Y,*)} has executed line \ref{right-cas}, 
and therefore also line \ref{finalize-cas}, 
\Y\ is finalized, by Lemma \ref{finalize-works}.
Since the \co{splice} has also executed lines \ref{left-cas} and \ref{right-cas}, it follows from Invariant \ref{inv} that $\U_i \not\rightarrow \Y$ and $\Y\not\leftarrow\U_{i+1}$.
So, by the claim proved in the first paragraph, \Y\ is not lr-reachable.
\end{proof}

\ignore{IT SEEMS THIS LEMMA IS NO LONGER NEEDED----------------------------------
\begin{lemma}
\label{splice-works}
After an instance of \co{splice(X,Y,Z)}, \co{\sul(X,Y,Z)} or\linebreak
\co{\sur(X,Y,Z)} returns true, \Y\ is not lr-reachable.
\end{lemma}
\begin{proof}
If a \co{splice(X,Y,Z)} has returned true, then it performed line \ref{right-cas}, so \Y\ is not lr-reachable, by Lemma \ref{full-splice-works}.

Suppose an instance $s$ of \co{\sul(X,Y,Z)} or
\co{\sur(X,Y,Z)} has returned true.
Then $s$ installed a Descriptor \co{desc} containing \co{(X,Y,Z)} into \co{X->rightDesc} or \co{Z->leftDesc}.
Before installing \co{desc}, we have $\X=\co{null}$ or $\X\rightarrow\Y$ at line \ref{sul-test-right} or \ref{sur-test}.
After installing \co{desc}, $s$ completed a call to \co{help(desc)} at line \ref{sul-help-new} or \ref{sur-help-new}, which in turn completed a call
to \co{splice(X,Y,Z)}.
If this \co{splice} executed line \ref{right-cas}, then \Y\ would not
be lr-reachable, by Lemma \ref{full-splice-works}.
Otherwise, this splice terminated on line \ref{test-right} with $\X\neq \co{null}$ and $\X\not\rightarrow \Y$.
Thus \co{X->right} changed from \Y\ to some other value, and that change must have been done at line 
\ref{right-cas} of a \co{splice(X,Y,*)}.  By Lemma \ref{full-splice-works}, \Y\ is not lr-reachable.
\end{proof}
----------------------------------------------}

\begin{lemma}
\label{nbrs-when-finalized}
At the time a node \Y\ is finalized by a \co{splice(X,Y,Z)}, the following hold with respect to the ordering $<_c$.\\
If $\X \neq \co{null}$ then \X\ is the greatest unfinalized node that is less than \Y\ and $\X\leftrightarrow\Y$.\\
If $\X = \co{null}$ then \Y\ is the smallest unfinalized node.\\
If $\Z \neq \co{null}$ then \Z\ is the least unfinalized node that is greater than \Y\ and $\Y\leftrightarrow\Z$.
\end{lemma}
\begin{proof}
Consider the configuration before the step that finalizes \Y.
By Invariant \ref{inv}.\ref{splice-args}, we have $\X\leftarrow \Y \rightarrow \Z$.

Since \Y\ is not finalized, if \Z\ is a finalized node, then a 
\co{splice(Y,Z,*)} must have been invoked by Invariant \ref{inv}.\ref{well-formed}, 
but this would violate Invariant \ref{inv}.\ref{no-overlap}.  
So, \Z\ is not a finalized node.  By Invariant \ref{inv}.\ref{well-formed},
if \Z\ is not \co{null}, then it is the smallest unfinalized node greater than \Y\ and $\Y\leftarrow\Z$.

Suppose \X\ is not \co{null}.
If \X\ is finalized, then a \co{splice(*,X,Y)} must have been invoked, violating 
Invariant \ref{inv}.\ref{no-overlap}.
So, \X\ is not finalized.  
By Invariant \ref{inv}.\ref{well-formed} \X\ is the greatest
unfinalized node that is less than \Y.
We had $\X\rightarrow\Y$ at line \ref{test-right} before \Y\ is finalized, so
when \Y\ is finalized, we have $\co{X->right} \geq_c \Y$, by Lemma \ref{order}.
By Invariant \ref{inv}.\ref{well-formed}, $\co{X->right} \leq_c \Y$ since neither \X\ nor \Y\ 
is finalized.
Thus, $\X\rightarrow\Y$.

If \X\ is \co{null}, it follows from Invariant \ref{inv}.\ref{well-formed-left} that
\Y\ is the smallest unfinalized node.
\end{proof}

\begin{lemma}
\label{frozen-left}
If, in some configuration $C$, \X\ is frozen
and \co{X->left} is not a descendant of \X, then \co{X->left} never changes after $C$.
\end{lemma}
\begin{proof}
Let \V\ be \co{X->left} in $C$.
To derive a contradiction, suppose \co{X->left} changes from \V\ to some other value 
after $C$.
This change must be performed by a \co{splice(*,V,X)} operation.
By Lemma \ref{splice-pre}, \V\ is not \co{null}.
By assumption, \V\ is not a descendant of \X\, so by Lemma \ref{anc-or-desc}, \V\ is an ancestor of \X.
Since $\V <_p \X$, the \co{splice(*,V,X)} operation
must have been called by a \co{help(desc)}, where \co{desc} is a Descriptor containing \co{(*,V,X)}
that was stored in \co{X->leftDesc}.
Since \co{X->leftDesc} was \co{frozen} before $C$,
\co{desc} was removed from \co{X->leftDesc} before $C$.
By Invariant \ref{inv}.\ref{descriptors-work}, $\co{X->left} <_c \V$ at some time before $C$.
By Lemma \ref{order}, $\co{X->left} <_c \V$ at $C$, a contradiction.
\end{proof}

We say a \co{removeRec} {\it is at \X} if an instance of \co{removeRec(X)} has been
called and it has neither terminated nor made a recursive call.

If \X\ is a finalized node,
we wish to define the instance of \co{removeRec(X)} that {\it caused} \X\ to become finalized.
\X\ was finalized by a call to \co{splice(*,X,*)}.
If this \co{splice} is called from line \ref{call-splice} of an instance of \co{removeRec(X)},
then that is the instance that caused \X\ to become finalized.
Otherwise, the \co{splice} must have been called by \co{help(desc)}, where
\co{desc} is a Desciptor that was installed in a node
by a call to \co{\sul(*,X,*)} or \co{\sur(*,X,*)} at line \ref{call-sul} or \ref{call-sur}
of an instance of \co{removeRec(X)}.  In this case, that instance is the one that cause 
\X\ to become finalized.
By Lemma \ref{bound-recursion}, there is a unique \co{removeRec(X)} that caused \X\ to be finalized.

\begin{lemma}
\label{type-0-1-bound}
If, in a configuration $C$, a node \X\ is of type 0 or type 1 and is frozen but not finalized,
then there is a \co{remove} or \co{removeRec} at a descendant of \X\ in $C$.
\end{lemma}
\begin{proof}
We prove the claim by induction on the height of the subtree rooted at \X.

\begin{description}
\item[Base case]
(\X\ is a leaf of $T$):
To derive a contradiction, suppose the claim is false for \X.
Then, the call to \co{remove(X)} that 
froze \X\ called \co{removeRec(X)} and that call either terminated or recursed to another node.
Since \X\ is not finalized in $C$, the test at line \ref{test-unfinalized} did not 
cause \co{removeRec(X)} to terminate.
By Lemma \ref{anc-or-desc}, the nodes $\X_L$ and $\X_R$ read at line \ref{read-left} and \ref{read-right}
are either \co{null} or ancestors of \X, so the test at line \ref{test-priority} succeeded and 
\co{splice($\X_L$,X,$\X_R$)} was called at line \ref{call-splice}.

Consider the test at line \ref{test-right} of \co{splice($\X_L$,X,$\X_R$)}.
By Invariant \ref{inv}, we cannot have $\co{$\X_L$->right} >_c \X$, since \X\ is not finalized.
If $\co{$\X_L$->right} = \W <_c \X$, then by Invariant \ref{inv}, 
\W\ is finalized and a \co{splice($\X_L$,W,X)} has been called.
By Lemma \ref{anc-or-desc}, this means $\W <_p \X$.
So, the \co{splice($\X_L$,W,X)} was called by \co{help} on a Descriptor stored in \co{X->leftDesc}.
By Lemma \ref{desc-frozen}, this Descriptor was removed before the call to \co{removeRec(X)}.
By Invariant \ref{inv}.\ref{descriptors-work}, $\X_L \not\rightarrow \W$ when line \ref{test-right} is executed,
a contradiction.

Thus, the test at line \ref{test-right} of \co{splice($\X_L$,X,$\X_R$)} cannot cause the splice to terminate.
Since the \co{removeRec(X)} is no longer at \X, the \co{splice($\X_L$,X,$\X_R$)} must execute 
line~\ref{finalize-cas} before~$C$.  By Lemma \ref{finalize-works}, this means \X\ is finalized before $C$, a contradiction.

\item[Induction step]
(\X\ is an internal node of $T$):  We assume the claim holds for
proper descendants of \X\ and prove that it holds for \X.
To derive a contradiction, assume the claim does not hold at \X.
Thus, at configuration $C$, \X\ is frozen but not finalized and there is no \co{remove} or
\co{removeRec} at a descendant of \X.

Let $S$ be the set of proper descendants \Y\ of \X\ such that in $C$, \Y\ is finalized and there are
no unfinalized nodes between \Y\ and \X\ with respect to the ordering $<_c$.
Let $S_L$ and $S_R$ be the subsets of $S$ in \X's left and right subtree, respectively.

\begin{claim*}
There is a configuration $C_L$ before $C$, such that
\begin{enumerate}[(a)]
\item
\label{SL-unreachable-left}
in $C_L$, \co{X->left} does not point to any node in $S_L$,
\item
\label{SL-unreachable-right}
in $C_L$, if $\co{X->left} \neq \co{null}$ then $\co{X->left->right} = \X$, and 
\item
\label{late-call}
a \co{removeRec(X)} is invoked between $C_L$ and $C$.
\end{enumerate}
\end{claim*}

\begin{claimproof}
If $S_L$ is empty, let $C_L$ be the configuration when the \co{right} pointer of \X's predecessor in $<_c$ is updated to point to \X. 
(If \X\ is the minimum node with respect to $<_c$, then let $C_L$ be the configuration when the
head pointer is changed to \X.)
The first part of the claim is vacuously true.  In $C_L$, if $\co{X->left} \neq \co{null}$, then $\co{X->left->right} = \X$.
Since \X\ is frozen in $C$, there must have been a call to \co{removeRec(X)} between $C_L$ and $C$.

If $S_L$ is non-empty, let \W\ be the node in $S_L$ that was finalized last.
Consider the \co{removeRec(W)} $rr_{\W}$ that causes \W\ to be finalized.
Let $\W_L$ and $\W_R$ be the values $rr_{\W}$ reads from \co{W->left} and \co{W->right} 
at line \ref{read-left} and \ref{read-right}.
By Invariant \ref{inv}.\ref{splice-args}, we have $\W_L \leftarrow \W \rightarrow \W_R$ when 
\W\ becomes finalized.
By Lemma \ref{nbrs-when-finalized}, $\W_L$ and $\W_R$ are the closest unfinalized nodes 
on either side of \W\ in the $<_c$ ordering when \W\ becomes finalized 
(or \co{null} if there is no such node).
So, by definition $\W_R = \X$.  By Lemma \ref{anc-or-desc}, $\W_L$ is either an ancestor of \X\ (or \co{null})
if $S_L$ contains all nodes \X's left subtree, 
or a descendant of \X, otherwise.

We consider two cases to show that some call to \co{splice(*,W,*)}
executes line \ref{right-cas} before~$C$.

Suppose $rr_{\W}$ calls \co{splice($\W_L$,W,$\W_R$)} at line \ref{call-splice} and
successfully performs line \ref{finalize-cas} to finalize \W.
Then, $rr_{\W}$ must also perform line \ref{right-cas}
before $C$, since there is no \co{removeRec} at \W\ in $C$, by assumption.

Suppose $rr_{\W}$ calls \sul\ or \sur\ and installs a Descriptor containing \co{($\W_L$,W,$\W_R$)}.
Then, at line \ref{sul-test-right} or \ref{sur-test}  
we have $\W_L\rightarrow \W$ if $\W_L$ is not \co{null}.
After installing the Descriptor, $rr_{\W}$ calls 
\co{splice($\W_L$,W,$\W_R$)} via line \ref{sul-help-new} or \ref{sur-help-new}.
This call must complete before $C$, since there is no \co{removeRec} at \W\ in $C$, by assumption.
If this \co{splice} sees that $\W_L$ is not \co{null} and $\W_L\not\rightarrow \W$ at line \ref{test-right},
then \co{$\W_L$->right} must have been changed from \W\ to another value.  This can only be done
by some \co{splice(*,W,*)} performing line \ref{right-cas}.
(Since \W\ is not \co{null}, this change was not by a \co{tryAppend}.)
Otherwise, $rr_{\W}$'s call to \co{splice($\W_L$,W,$\W_R$)} will perform line \ref{right-cas} itself.

As argued above, $rr_{\W}$ calls \co{splice($\W_L$,W,X)}.
By Invariant \ref{inv}.\ref{splice-args}, {\it all} calls to \co{splice(*,W,*)}
are of the form \co{splice($\W_L$,W,X)}.
As argued above, some \co{splice($\W_L$,W,X)} executes line \ref{right-cas} before $C$.
Let $C_L$ be the configuration after the 
first execution of line \ref{right-cas} by any call to \co{splice($\W_L$,W,X)}.
Prior to $C_L$, a \co{splice($\W_L$,W,X)} executes a \co{CAS(X->left,W,$\W_L$)} at line \ref{left-cas}.

By Lemma \ref{nbrs-when-finalized},
$\W \leftrightarrow \X$ in the configuration before \W\ was finalized.
Thus, when the first CAS(\co{X->left,W,$\W_L$}) is performed, it must succeed,
since no other step can change \co{X->left}.
By Lemma \ref{order}, this means that in all configurations after $C_L$,
$\co{X->left} \leq_c \W_L$.
Since $\W_L$ is not a finalized node when \W\ is finalized,
it must be either \co{null} or less than all nodes in $S_L$ with respect to the order $<_c$.
Thus, at all times after $C_L$, either \co{X->left} is \co{null} or $\co{X->left} <_c \Y$ for all $\Y\in S_L$.
This completes the proof of Claim~\ref{SL-unreachable-left}.

We now prove Claim \ref{SL-unreachable-right}.  Assume \co{$\X$->left} is not \co{null} in $C_L$.

We first show that if $W_L$ is not \co{null}, then the execution of \co{CAS($\W_L$->right,W,X)}
immediately before $C_L$ is successful.
Before \W\ was finalized by line \ref{finalize-cas} of a \co{splice($\W_L$,W,X)}
we have $\W_L\rightarrow \W$ at line \ref{test-right}.
The only way that \co{$\W_L$->right} can change is by line \ref{right-cas} of a call
to \co{splice($\W_L$,W,X)}. 
(Since \W\ is not \co{null}, this change cannot be by a \co{tryAppend}.) 
Thus, the first such execution succeeds.
So, if $\W_L \neq \co{null}$, then we have $\W_L\rightarrow \X$ in $C_L$.

As argued above, after the first execution of line \ref{left-cas} by a call to \co{splice($\W_L$,W,X)},
we have $\W_L \leftarrow \X$.
To derive a contradiction, suppose $\W_L \not\leftarrow \X$ at $C_L$.
Then \co{X->left} has been changed by line \ref{left-cas} of a \co{splice(*,$\W_L$,X)},
which means $\W_L \neq \co{null}$.
When this change occurs, we have $\W_L\rightarrow \X$, by Invariant \ref{inv}.\ref{splice-args}.
But this cannot be true before the step preceding $C_L$, as argued in the previous paragraph.
This contradiction shows that $\W_L \leftarrow \X$ in $C_L$.
So, in $C_L$, \co{X->left} is $\W_L$ (which is assumed to be non-\co{null}) and \co{$\W_L$->right} is \X,
which establishes claim \ref{SL-unreachable-right}.

We now prove Claim \ref{late-call}.  
The routine $rr_{\W}$ either calls \co{splice($\W_L$,\W,$\W_R$)} directly at 
line \ref{call-splice} or via \co{help} after successfully storing a Descriptor
containing \co{($\W_L$,\W,$\W_R$)}.  Thus, $rr_{\W}$ terminates or recurses
between $C_L$ and $C$.
We consider two cases.

If \X\ is frozen in $C_L$, then $rr_{\W}$ cannot terminate; it must recurse.
If $\W_L$ is \co{null}, then $rr_{\W}$ will recurse to \X.
If $\W_L$ is an ancestor of \X, then $rr_{\W}$ will recurse to \X, because $\X >_p \W_L$.
If $\W_L$ is a descendant of \X, then $\W_L$ is unfinalized at $C$ (by definition of $S$) and
it is not of type 2, since all nodes in $\W_L$'s right subtree are finalized at $C$,
and hence by the induction hypothesis, $\W_L$ is not frozen before $C$.
Thus, $rr_{\W}$
will also recurse to \X\ in this case.
In all cases, there is a call to \co{removeRec(X)} between $C_L$ and $C$.

If \X\ is not frozen in $C_L$, then it becomes frozen by a \co{remove(X)} between $C_L$ and $C$.  
By assumption, there is no \co{remove(X)} at \X\ in $C$, 
so the \co{remove(X)} that froze \X\ must call \co{removeRec(X)} between $C_L$ and $C$.

This completes the proof of the claim.
\end{claimproof}

\begin{claim*}
There is a configuration $C_R$ before $C$, such that
\begin{enumerate}[(a)]
\item
\label{SR-unreachable-left}
in $C_R$, \co{X->right} does not point to any node in $S_R$,
\item
\label{SR-unreachable-right}
in $C_R$, if $\co{X->right} \neq \co{null}$ then $\co{X->right->left} = \X$, and 
\item
\label{SR-late-call}
a \co{removeRec(X)} is invoked between $C_R$ and $C$.
\end{enumerate}
\end{claim*}

The proof of this claim is symmetric to the previous one.

\medskip


Combining the two claims, there is an instance $rr_{\X}$ of \co{removeRec(X)} that is invoked 
after both $C_L$ and $C_R$
and before $C$.  By assumption, $rr_{\X}$ must terminate or recurse before $C$.
It cannot terminate at line \ref{test-unfinalized} before $C$, since \X\ is not finalized in $C$.
Let $\X_L$ and $\X_R$ be the nodes that $rr_{\X}$ reads at line \ref{read-left} and \ref{read-right}.
We consider three cases to show that \X\ is finalized before $C$, which is the desired contradiction.

\begin{enumerate}[C{a}se 1]
\item
Suppose \X\ is of type 0 in $C$.  Then, $S$ contains all of \X's proper descendants.
By the claims proved above and Lemma \ref{order}, $\X_L$ and $\X_R$ are not descendants of \X.
By Lemma \ref{anc-or-desc} they must each be either \co{null} or an ancestor of \X.
Thus, $rr_{\X}$ calls \co{splice($\X_L$,\X,$\X_R$)} at line \ref{call-splice}.

Let \V\ be the value of \co{X->left} when $rr_{\X}$ is invoked.
In the configuration $C_L$, which is before $rr_{\X}$ is invoked,
\co{X->left} does not point to any node in \X's left subtree.
By Lemma \ref{order}, \co{X->left} cannot point to any node in \X's left subtree after $C_L$.
Thus, \V\ is not a descendant of \X.
By Lemma \ref{frozen-before-rr}, \X\ is frozen when $rr_{\X}$ is invoked.
By Lemma \ref{frozen-left}, $\V\leftarrow \X$ at all times after $rr_{\X}$ begins.
Thus, $\V=\X_L$ and $\X_L \leftarrow \X$ at all times after $rr_{\X}$ begins.

If $\X_L = \co{null}$, then line \ref{finalize-cas} of the \co{splice($\X_L$,X,$\X_R$)} called at line
\ref{call-splice} of $rr_{\X}$ finalizes \X\ before $C$, by Lemma \ref{finalize-works}.

Otherwise, let \U\ be the value of \co{$\X_L$->right} read at line \ref{test-right}.
To derive a contradiction, suppose $\U <_c \X$.
At line \ref{test-right}, we have $\X_L \leftarrow \X$ and $\X_L\rightarrow \U$ and \X\ is unfinalized.
By Invariant \ref{inv}, this means a \co{splice($\X_L$,U,X)} has been called and $\U\rightarrow \X$.
By Lemma \ref{anc-or-desc}, \U\ is either an ancestor or descendant of \X.
Since \X\ is unfinalized,
$\X_L$ and \U\ are lr-reachable at line \ref{test-right}.
Since all nodes in \X's left subtree are not lr-reachable at the earlier configuration
$C_L$, it follows from the Lemma \ref{remains-unreachable} that \U\ cannot be a descendant of \X.
Thus, \U\ is an ancestor of \X.
By Lemma \ref{desc-properties}, the \co{splice($\X_L$,U,X)} must have 
been called by \co{help} on a Descriptor containing \co{($\X_L$,U,X)} that was stored in 
\co{X->leftDesc}.
Since \X\ is frozen before $rr_{\X}$ begins, the fact that $\X_L\rightarrow \U$ during $rr_{\X}$
contradicts Invariant \ref{inv}.\ref{descriptors-work}.
Thus, $\U\geq_c\X$.
If $\U=\X$ then by Lemma \ref{finalize-works}, line \ref{finalize-cas} finalizes \X\ before~$C$.
If $\U>_c\X$ then \X\ is finalized by Invariant \ref{inv}.

\item
Suppose that in $C$, \X\ is of type 1 and all nodes of its right subtree are finalized.
Then, by Lemma \ref{anc-or-desc} and the first claim above, $\X_R$ is either \co{null} or an ancestor of \X.
In $C_L$, \co{X->left} does not point to any element of $S_L$ by the claim above,
and it cannot point past an unfinalized node by Invariant \ref{inv}.  Thus,
$\X_L$ is the maximum node (with respect to $<_c$) in \X's left subtree 
that is not finalized in $C_L$.  Since all nodes in $S_L$ are finalized in $C_L$,
$\X_L$ is not finalized in $C$ either.
It follows from Invariant \ref{inv}.\ref{well-formed} that throughout the period 
between $C_L$ and $C$, $\X_L \leftrightarrow \X$.

Since all nodes in $\X_L$'s right subtree are finalized in $C$, $\X_L$ is not of type 2 in $C$.
It follows from the induction hypothesis that $X_L$ is not frozen in $C$.
Since there is no \co{remove($\X_L$)} at $\X_L$ in $C$, $\X_L$ is unmarked in $C$.

So, we have $\X_L >_p \X >_p \X_R$ and $rr_{\X}$ completes a call $sul_{\X}$ to \co{\sul($\X_L$,X,$\X_R$)} before $C$.
Since $X_L$ is unmarked at $C$, $sul_{\X}$ does not terminate at line \ref{sul-test-marked}.
Since $\X_L\rightarrow \X$ throughout its execution, $sul_{\X}$ does not terminate at line \ref{sul-test-right}.
Thus, $sul_{\X}$ performs the CAS at line \ref{store-right-desc}.

To derive a contradiction, assume that this CAS fails.  Then some other call to \sul\ installed
a Descriptor \co{($\X_L$,$\X'$,$\X_R'$)} in \co{$\X_L$->rightDesc} between $sul_{\X}$'s executions of line \ref{read-desc2} and \ref{store-right-desc}.
If $\X' >_c \X$ then there was an earlier time when $\X_L\leftarrow \X'$ with $\X_L <_c \X <_c \X'$
which means that \X\ was finalized, a contradiction.
If $\X' = \X$ then the \co{removeRec(X)} that installed the Descriptor \co{($\X_L$,$\X'$,$\X_R'$)} will ensure that \X\ is finalized
before $C$.  
\future{more explanation needed}%
Otherwise, $\X' <_c \X$, so $\X' \in S_L$.
So, $\X'$ is finalized before $C_L$.
Hence, there was a \co{removeRec($\X'$)} that performed a successful CAS at line
\ref{finalize-cas}, \ref{store-right-desc} or \ref{store-left-desc}
that caused $\X'$ to be finalized before $C_L$.
This contradicts Lemma \ref{bound-recursion}.

So, $sul_{\X}$ performs the CAS at line \ref{store-right-desc} successfully, and then calls
\co{splice($\X_L$,X,$\X_R$)} via line \ref{sul-help-new} and that call to \co{splice} must complete before $C$.
It cannot terminate at line \ref{test-right} because $\X_L\rightarrow \X$ throughout its execution.
So it performs line \ref{finalize-cas}, which finalizes $\X$ before $C$, by Lemma \ref{finalize-works}.

\item
Suppose that in $C$, \X\ is of type 1 and all nodes of \X's left subtree are finalized.
Although this case is nearly symmetric to the previous one, there are some differences due to 
asymmetries in the code.

\future{following two parags are very similar to reasoning in case 1; can we pull out a common claim instead?}
Let \V\ be the value of \co{X->left} when $rr_{\X}$ is invoked.
In the configuration $C_L$, which is before $rr_{\X}$ is invoked,
\co{X->left} does not point to any node in \X's left subtree.
By Lemma \ref{order}, \co{X->left} cannot point to any node in \X's left subtree after $C_L$.
Thus, \V\ is not a descendant of \X.
By Lemma \ref{frozen-before-rr}, \X\ is frozen when $rr_{\X}$ is invoked.
By Lemma \ref{frozen-left}, $\V\leftarrow \X$ at all times after $rr_{\X}$ begins.
Thus, $\V=\X_L$ and $\X_L \leftarrow \X$ at all times after $rr_{\X}$ begins.

We next show that if $\X_L \neq \co{null}$ then $\X_L \rightarrow \X$ at all times
between the invocation of $rr_{\X}$ and $C$.
To derive a contradiction, assume that, 
at some configuration $C'$ between the invocation of $rr_{\X}$ and $C$,
$\co{$\X_L$->right} \neq \X$.
Let \U\ be the value of \co{$\X_L$->right} at $C'$.
By Invariant \ref{inv}.\ref{well-formed}, we have $\X_L \leftrightarrow \U \rightarrow \X$
at $C'$ and a \co{splice($\X_L$,U,X)} has been called.
By Lemma \ref{unfinalized-reachable}, \X\ and therefore $\X_L$ and \U\ are lr-reachable
at line \ref{test-right}.
Since all nodes in \X's left subtree are not lr-reachable at the earlier configuration
$C_L$, it follows from the Lemma \ref{remains-unreachable} that \U\ is not a descendant of \X.
Since we have $\U\rightarrow \X$, it follows from Lemma \ref{anc-or-desc} that \U\ is an ancestor of \X, so $\U <_p \X$.
Thus, the call to \co{splice($\X_L$,U,X)} must have 
been called by \co{help} on a Descriptor containing \co{($\X_L$,U,X)} that was stored in 
\co{X->leftDesc}.
Since \X\ is frozen before $rr_{\X}$ begins, the fact that $\X_L\rightarrow \U$ during $rr_{\X}$
contradicts Invariant \ref{inv}.\ref{descriptors-work}.

Let \Z\ be the minimum (with respect to $<_c$) node in \X's right subtree that is not finalized at $C$.
By the second claim above, $\co{X->right} \geq_c \Z$ at $C_R$.
Since \Z\ is not finalized at $C$, it follows from Lemma \ref{order} and Invariant \ref{inv}
that $\X \rightarrow \Z$ at all times between $C_R$ and $C$.
Thus $\Z = \X_R$ and $\X \rightarrow \X_R$ at all times between $C_R$ and $C$.
By the second claim above, $\X \leftarrow \X_R$ at $C_R$.
Since \X\ is not finalized at $C$, it follows from Lemma \ref{order} and Invariant \ref{inv}
that $\X \leftarrow \X_R$ at all times between $C_R$ and~$C$.

At $C$, all nodes in $\X_R$'s left subtree are finalized, so $\X_R$ is not of type~2.  
It follows from the induction hypothesis that $\X_R$ is not frozen in $C$.
Since there is no \co{remove($\X_R$)} at $\X_R$ in $C$, $\X_R$ is unmarked in $C$.

Since $\X_L <_p \X <_p \X_R$, $rr_{\X}$ calls \co{\sur($\X_L$,X,$\X_R$)} at line \ref{call-sur}.
Let $sur_{\X}$ be this call, which
must terminate or recurse before $C$.
It does not terminate at line \ref{sur-test-marked}, since $\X_R$ is unmarked at $C$.
We have already shown that at all times between the invocation of $rr_{\X}$ and $C$,
$\X\leftarrow \X_R$ and if $\X_L\neq\co{null}$ then $\X_L\rightarrow \X$.
Thus, $sur_{\X}$ cannot terminate at line \ref{sur-test}.
So, it performs the CAS at line \ref{store-left-desc}.

\future{Following parag is similar to one in Case 2; pull out a lemma?}
To derive a contradiction, assume that this CAS fails.  Then some other call to \sur\ installed
a Descriptor \co{($\X_L'$,$\X'$,$\X_R$)} in \co{$\X_R$->leftDesc} between $sur_{\X}$'s executions of line \ref{read-desc3} and \ref{store-left-desc}.
If $\X' <_c \X$ then there was an earlier time when $\X'\rightarrow\X_R'$ with $\X' <_c \X <_c \X_R$
which means that \X\ was finalized (by Invariant \ref{inv}), a contradiction.
If $\X' = \X$ then the \co{removeRec(X)} that installed the Descriptor \co{($\X_L'$,$\X'$,$\X_R$)} will ensure that \X\ is finalized
before $C$.  
\future{more explanation needed}%
Otherwise, $\X' >_c \X$, so $\X' \in S_R$.
So, $\X'$ is finalized before $C_R$.
Hence, there was a \co{removeRec($\X'$)} that performed a successful CAS at line
\ref{finalize-cas}, \ref{store-right-desc} or \ref{store-left-desc}
that caused $\X'$ to be finalized before $C_R$.
This contradicts Lemma \ref{bound-recursion}.
\qedhere
\end{enumerate}
\end{description}
\end{proof}

Finally, we can present the proof of the overall space bound.

\removeSpace* 

\begin{proof}
We divide the lr-reachable nodes into several categories in order to bound the number of them.
There are $L-R$ nodes that are not removable.
If a removable node \X\ is not frozen, then the \co{remove(X)} operation is still at \X,
so there are at most $P$ such nodes.
If a node \X\ is finalized but still lr-reachable, then there is a \co{splice(*,X,*)} that
performed line \ref{finalize-cas}, but not line \ref{right-cas}, by Lemma \ref{full-splice-works}, so there are at most $P$ such nodes.
It remains to bound the number of nodes that are frozen but not finalized.
By Lemma \ref{type-0-1-bound}, there are $O(P\log L)$ such nodes of types
0 and 1.  
By Lemma \ref{type-2-bound}, the total number of unfinalized nodes of type 2 is at most $L-R +O(P\log L)$.
Thus, the number of lr-reachable nodes is at most $2(L-R) + O(P\log L)$.
\end{proof}


\section{Details of Memory Reclamation for Version Lists}
\label{sec:smr-list-proof}


\subsection{Pseudo-code}

As described in Section \ref{sec:smr-list}, we must make some adjustments to the list implementation
of Section \ref{sec:vlists} to allow for efficient reclamation of list nodes.
The majority of the VersionList methods are implemented the same way as in Figure \ref{fig:list-alg-full} except with raw pointers replaced by reference-counted pointers.
Some methods require more involved changes; only these methods are shown in Figure \ref{fig:list-alg-full-gc},
with the changes shown in blue.

Using the reference counting scheme from~\cite{ABW21} involves adding a reference count field to each Node and Descriptor object and replacing raw pointers to \y{Node}s or Descriptors with reference-counted pointers (of type \rcptr{} and \arcptr{}), which automatically manage the reference counts of the objects they point to.
The \rcptr{} type supports read, write and dereferencing just like raw pointers,
and it is typically used for pointers stored in local memory.
The \arcptr{} (a.k.a. \texttt{atomic\_rc\_ptr} in~\cite{ABW21}) type additionally supports CAS,
 and is typically used for pointers stored in shared memory.


\begin{figure*}[!thp]\small
  \hspace{-10mm}
  \begin{minipage}[t]{.55\textwidth}
  	\StartLineAt{200}
    \begin{lstlisting}[linewidth=.99\columnwidth, numbers=left,frame=none]
class Node {
  arc_ptr<Node> left, right; // initially null
  enum status {unmarked,marked,finalized};
    // initially unmarked
  int counter; // used to define priority
  int priority; // defines implicit tree
  int ts; // timestamp
  arc_ptr<Descriptor> leftDesc, rightDesc;  
    // initially null
  @\hil{int refCount;}@ };

class Descriptor {
  arc_ptr<Node> A, B, C;
  @\hil{int refCount;}@ };

class VersionList {
  arc_ptr<Node> Head;

  rc_ptr<Node> find(rc_ptr<Node> start, 
                      int ts) {
    rc_ptr<Node> cur = start;
    while(cur != null && cur->ts > ts) {
      @\hil{rc\_ptr<Node> next = cur->left;}@
      @\hil{if(next == $\top$) next = cur->right;}@
      @\hil{cur = next;}@ }
    return cur; } };
  \end{lstlisting}
\end{minipage}\hspace{.3in}
\begin{minipage}[t]{.65\textwidth}
  \StartLineAt{226}
  \begin{lstlisting}[linewidth=.99\textwidth, numbers=left, frame=none]
  bool tryAppend(rc_ptr<Node> B, rc_ptr<Node> C) {
    if(B != null) {
      C->counter = B->counter+1;
      rc_ptr<Node> A = B->left;             @\label{mr-read-left-gc}@
      @\hil{if(A == $\top$) return false;}@   @\label{check-top-gc}@
      if(A != null) A->right.CAS(null, B); @\label{mr-cas1-gc}@
    } else C->counter = 2;
    C->priority = p(C->counter);          @\label{set-priority-gc}@
    C->left = B;                  @\label{set-left-gc}@
    if(Head.CAS(B, C)) {              @\label{head-cas-gc}@
      if(B != null) B->right.CAS(null, C); @\label{mr-cas2-gc}@
      return true;
    } else return false; }

  bool splice(rc_ptr<Node> A, rc_ptr<Node> B,
                                 rc_ptr<Node> C) {
    // B cannot be null
    if(A != null && A->right != B) return false;              @\label{test-right-gc}@
    bool result = CAS(&(B->status), marked, finalized);    @\label{finalize-cas-gc}@
    if(C != null) C->left.CAS(B, A);                    @\label{left-cas-gc}@
    if(A != null) A->right.CAS(B, C);                   @\label{right-cas-gc}@
    @\hil{if(A != null \&\& A->priority > B->priority)}@
        @\hil{B->left = $\top$;}@     @\label{clear-left-gc}@
    @\hil{if(C != null \&\& C->priority > B->priority)}@
        @\hil{B->right = $\top$;}@    @\label{clear-right-gc}@
    return result; } };
  \end{lstlisting}
\end{minipage}
\vspace{-.1in}
\caption{Changes to the implementation of version lists to accommodate memory management.
}
\label{fig:list-alg-full-gc}
\end{figure*}

\subsection{Time Bounds for Memory-Managed Version Lists}

In this section, we prove Theorem \ref{thm:step3-time}.
We divide the proof into the following pieces.
\begin{itemize}
	\item Lemma~\ref{lem:ref-count} shows that the amortized number of steps for reference counting operations is constant.
	\item Lemma~\ref{lem:remove-with-gc} shows that the changes to \op{remove} do not affect its constant amortized time bound.
	\item Lemma~\ref{lem:traversal} shows that the modified \op{find} operation does not go back to the same node multiple times.
\end{itemize}

We assume that memory can be allocated and de-allocated in constant time.
Since nodes and Descriptor objects have a fixed size, we can make use of recent work on implementing concurrent fixed-sized \op{alloc} and \op{free} in constant time~\cite{blelloch2020brief}.

\begin{lemma}
	\label{lem:ref-count}
	Read-only operations on \rcptr{} and \arcptr{} (including reading and dereferencing) take O(1) steps.
	If each 
	object contains $O(1)$ reference-counted pointers, the expected amortized number of steps for each update operation (store and CAS) on \rcptr{} 
	and \arcptr{} is $O(1)$.
\end{lemma}

\begin{proof}
	Theorem 1 from~\cite{ABW21} states that read-only operations perform $O(1)$ steps, and that the expected number of steps for update operations, ignoring the call to \op{delete}, is constant.
	The \op{delete} operation takes as input a pointer and it frees the object being pointed to. If the object contains any reference counted pointers, the reference count of those pointers are decremented. If these decrements cause any objects' reference counts to hit 0, then delete is recursively called on those objects.

	It suffices to argue that the cost of \op{delete} can be amortized to earlier update operations.
	When the \op{delete} operation is called on an object $O$, for each reference counted pointer $ptr$ that belongs to $O$, \op{delete} decrements the reference count of the object pointed to by $ptr$.
	After decrementing the reference counts, the \op{delete} operation frees the memory occupied by $O$, which takes constant time.
	If an object's reference count hits 0 due to one of the earlier decrements, \op{delete} is recursively called on that object.
	Since each object contains a constant number of reference counted pointers, the step complexity of a \op{delete} operation is proportional to the number of decrements performed by the \op{delete} operation.
	We can amortize these reference count decrements as follows. Suppose the \op{delete} operation decrements the reference count of object $Y$ because it was pointed to by field $F$ of object $X$. We charge the cost of this decrement to the update operation that caused $X.F$ to point to $Y$.
	In this way, we guarantee that each update operation is charged with at most one decrement from a \op{delete} operation.
	Therefore, each update operation on a reference-counted pointer takes an amortized expected constant number of steps.
\end{proof}

The following lemma proves both the correctness of the \op{remove} algorithm in Section~\ref{sec:smr-list} as well as its time bounds.
It is convenient to combine these two arguments because they both boil down to arguing that the \op{remove} algorithm in Section~\ref{sec:smr-list} behaves the same way as in Section~\ref{sec:vlists}.
We show that $R$ remove operations takes $O(R)$ steps plus the cost of $O(R)$ operations on reference-counted pointers. 

\begin{lemma}
	\label{lem:remove-with-gc}
	Theorem~\ref{space-bound} holds for the \op{remove} algorithm in Section~\ref{sec:smr-list}. 
	Similar to Theorem~\ref{remove-amortized-time}, $R$ \op{remove} operations take $O(R)$ steps plus the cost of $O(R)$ operations on reference-counted pointers.
\end{lemma}

\begin{proof}
	We begin by comparing Figure~\ref{fig:list-alg-full-gc} with Figure~\ref{fig:list-alg-full} and arguing that the changes made in Figure~\ref{fig:list-alg-full-gc} do not affect the behavior of the original algorithm.
	The major change is Lines \ref{clear-left-gc} to \ref{clear-right-gc} of \op{splice}, which set the 
	\co{left} and \co{right} pointers of the spliced out node to $\top$ if they point to descendants, 
	as well as the additional if statement on Line \ref{check-top-gc}.
	This change only adds constant time overhead to \op{splice} and \op{tryAppend}, but we need to argue that it does not change the time complexity and correctness of other parts of \op{remove} or \op{tryAppend}.
	To do this, we need to go through every read and CAS on \co{left} or \co{right} pointers that could have potentially seen the value $\top$ and argue that the operation would have behaved the same way had it seen the value that was overwritten by $\top$.
	
	We begin by reasoning about CAS operations performed on \co{left} and \co{right} fields by the algorithm of Figure \ref{fig:list-alg-full}. 
	It follows from Lemma \ref{splice-pre} and Lemma \ref{order} that the CAS steps on line \ref{left-cas} and \ref{right-cas} of \co{splice} use distinct values for the expected and new value.
	By Invariant \ref{inv}.\ref{splice-args}, \co{B->left} and \co{B->right} never change after \op{splice(*,B,*)} is invoked.  Thus, CAS operations on the \co{left} and \co{right} pointers of \op{B} can never succeed. 
	Therefore setting \op{B}'s \co{left} or \co{right} pointer to $\top$ during a \op{splice(*,B,*)} does not impact the return value of any other CAS operation.
	
	Consider the reads on Lines \ref{read-left} and \ref{read-right}. We check if \B\ is finalized immediately after the two reads. If either of them were $\top$, then the check will fail, so the operation returns without making use of the invalid pointer.
	
	Consider the read of \op{C->left} on line \ref{sur-test} of \co{\sur(A,B,C)}. If it returns $\top$, then \C\ was marked and frozen, so even if the previous value of \op{C->left} was \B\, the \op{spliceUnmarkedRight} would have returned false anyways after failing the next CAS on Line \ref{store-left-desc}, by Lemma \ref{desc-frozen}.
	The same argument applies to the read of \op{A->right} at line \ref{sul-test-right} of \co{\sul(A,B,C)}.
	\er{Now consider a read of \op{A->right} on line \ref{sur-test} of
	\co{\sur(A,B,C)} or on line \ref{test-right} of \co{splice(A,B,C)}.	
	If it returns \B, then there cannot have been a call to \co{splice(*,A,*)} prior to this test,
	because by Invariant \ref{inv}.\ref{splice-args}, it would have to be a
	\co{splice(*,A,B)},
	which would violate Invariant \ref{inv}.\ref{no-overlap}.
	Thus, in the modified algorithm, \op{A->right} cannot have been changed to $\top$ before the read
	of \co{A->right} occurs.
	So the tests on lines \ref{sur-test}, \ref{sul-test-right} and \ref{test-right} 
	evaluate to the same value
	with or without the modification to the \co{splice} routine.
	Therefore, the return value and behavior
	of \co{splice}, \sul\ and \sur\ do not change. }
	
	
	For the read on Line \ref{mr-read-left-gc}, if \op{B->left} is $\top$, then \B\ has already been finalized, so a newer node has already been installed, in this case the \op{tryAppend} is guaranteed to fail and it will return false on Line \ref{check-top-gc}. Conversely, if \op{tryAppend} returns false on Line \ref{check-top-gc}, then \B\ is no longer at the head of the version list so the CAS on Line \ref{head-cas-gc} is guaranteed to fail.
	
	Therefore, adding Lines \ref{clear-left-gc}--\ref{clear-right-gc} and \ref{check-top-gc} to Figure \ref{fig:list-alg-full} does not affect its
	\Eric{What is "it"?} 
	behavior and this change just potentially causes some operations to return sooner, so the same arguments for correctness and time bounds apply (Theorems \ref{remove-amortized-time} and~\ref{space-bound}).
	
	The only remaining changes is that some raw pointers in Figure~\ref{fig:list-alg-full} are replaced with reference-counted pointers in Figure~\ref{fig:list-alg-full-gc}.
	This does not affect the behavior (and therefore the correctness) of the algorithm as long as safety is ensured (i.e. a node or Descriptor is not accessed after it is freed).
	Our algorithm ensures safety because all access to nodes and Descriptors are done through reference counted pointers and these pointers prevent it from being freed.	
	In terms of time bounds, \op{tryAppend} performs a constant number of reference-counted pointer operations and $R$ \op{remove} operations perform at most $O(R)$ reference-counted pointer operations, so the lemma holds.
\end{proof}

\future{Think of a better name for forward/upward traversals, since upward doesn't really mean up the tree}

\er{
\begin{definition}
\label{def:up-down-traversal}
	An update to \co{cur} in the \co{find} routine made by following a \co{left} or \co{right} pointer is called a \emph{traversal}.
	An \emph{upward traversal} is one that follows a link from an finalized node to another node.
	A \emph{forward traversal} is one that follows a link from a unfinalized node to another node. 
\end{definition}

A forward traversal always follows a \co{left} pointer.
We use the term upward traversal because
most upward traversals move from a node to an ancestor. However, this is not necessarily the case:  if a node has been finalized but its \co{left} pointer has not yet been set to $\top$, then an upward traversal from the node might go to a descendant.
}

\begin{lemma}
	\label{lem:traversal}
	The source nodes of upward traversals in a \co{find} operation are all distinct.
	Similarly, the destination nodes of forward traversals in a \co{find} are all distinct.
\end{lemma}

\begin{proof}
	We first prove the claim about upward traversals. 
	\er{Suppose there is an upward traversal from $\X_0$ to $\X_1$.
	Let $\X_2, \X_3, \ldots$ be the sequence of nodes visited after $\X_1$.
	Assume there is another upward traversal after the one from $\X_0$ to $\X_1$, and let
	the first such upward traversal be from $\X_k$ to $\X_{k+1}$.
	We argue that $\X_k$ becomes finalized after $\X_0$ becomes finalized by considering two cases.
	
	First, suppose $k=1$.
	Consider the time when the \co{find} traverses from $\X_0$ to $\X_1$.
	Either $\X_0\rightarrow \X_1$ or $\X_1\leftarrow \X_0$.
	$\X_0$ was finalized by line \ref{finalize-cas} of a \co{splice(*,$\X_0$,*)}.
	By Invariant \ref{inv}.\ref{splice-args}, that \co{splice} 
	was of the form \co{splice($\X_2$,$\X_0$,$\X_1$)} or \co{splice($\X_1$,$\X_0$,$\X_2$)} for some $\X_2$.
	By Invariant \ref{inv}, $\X_1$ was unfinalized when that pointer was swung to make $\X_0$ \lrunreachable.
	
	Now, suppose $k>1$.  By definition, the traversal from $\X_{k-1}$ to $\X_k$ is a forward
	traversal.  Thus, at the time the traversal is performed, 
	$\X_{k-1}$, and therefore $\X_k$, is \lrreachable.
	Thus, $\X_k$ becomes \lrunreachable\ after $\X_0$.
	
	If we now consider the sequence of source nodes of the upward traversals in the \co{find},
	it follows that each becomes \lrunreachable\ earlier than the next one in the sequence,
	so the same node can never appear as the source of two upward traversals.
	}
		
	
	Next, we prove the claim about forward traversals.
	Suppose there is a forward traversal $f$ from \Y\ to \Z.
	To derive a contradiction, assume that at some later configuration, there is another forward traversal $f'$ from \X\ to \Z.   (\X\ might be equal to \Y.)
	Consider the path of nodes $\pi = \Z, \U_1, \U_2, ..., \U_k = \X, \Z$ that the traversal visits between $f$ and $f'$.
	When the forward traversal $f'$ from \X\ to \Z\ occurs, \X\ is \lrreachable\ and therfore \Z\ is \lrreachable.  Thus, \Z\ is \lrreachable\ when the earlier traversal from \Z\ to $\U_1$ occurs since once \Z\ becomes \lrunreachable{}, it cannot become \lrreachable{} again, by Lemma \ref{remains-unreachable}.
	Therefore, the traversal from \Z\ to $\U_1$ is a forward traversal.
	By Lemma \ref{order}, we know that $\U_1 <_c \Z$ because $\U_1$ was a left neighbor of \Z.
	By the same lemma, we know that $\X >_c \Z$.
	Therefore, there are two consecutive nodes, \A\ and \B\, in $\pi$ such that $\A \leq_c \Z$
	and $\B>_c \Z$.
	By Lemma \ref{order}, the traversal from \A\ to \B\ must follow a \co{right} pointer and is therefore an upward traversal.
	$\A \neq \Z$ because only forward traversals can be performed from \Z\ and these always go left.
	Therefore, $\A <_c \Z$.
	Since $\pi$ performs an upward traversal from \A\ to \B, \A\ must have been \lrunreachable{} (and therefore finalized) at the time of the traversal and $\A \rightarrow \B$ so by Invariant \ref{inv}.\ref{no-skip-right}, all nodes \D\ such that $\A \leq_c \D \leq_c \B$ are unreachable.
	That would imply that \Z\ is unreachable, which is a contradiction.
	Therefore, destination nodes of forward traversals are always distinct.
%
\end{proof}

\begin{lemma}
	\label{lem:up-traversal}
	If \op{find} performs an upward traversal from \B\ to \C, then \C\ is a proper ancestor of \B.
\end{lemma}

\begin{proof}
	If \op{find} performs an upward traversal from \B\ to \C, then \co{B->left} is $\top$.
	So, \B\ must have already been finalized by some \op{splice(A,B,C$'$)} operation.
	By Invariant \ref{inv}.\ref{splice-args}, $\B \rightarrow \C'$ at the time of the upward traversal from \B\ to \C, so $\C = \C'$.
	To complete the proof, we argue that \C\ is a proper ancestor of \B.
	Since \co{B->left} was set to $\top$, we know that $\A \geq_p \B$ so by Lemma \ref{splice-priorities}, we know that $\B \geq_p \C$.
	By Lemma \ref{anc-or-desc}, \co{B->right} is either a proper ancestor or a proper descendant of \B, so \C\ is a proper ancestor of \B.
\end{proof}

\Eric{I changed "\A\ is to the left of \B" to "$\A <_c \B$" in following lemma, since "to the left/right" is a bit unclear.}

\begin{lemma}
	\label{lem:common-ancestor}
	Suppose \A\ and \B\ are two nodes in a list such that $\A <_c \B$. If \C\ is an ancestor of \A\ and $\C >_c \B$ then \C\ is also an ancestor of \B.
\end{lemma}

\begin{proof}
	Let \C\ be an ancestor of \A\ such that $\C >_c \B$.
	Our goal is to show that \C\ is also an ancestor of \B.
	Let \D\ be the lowest common ancestor of \A\ and \B.
	Either $\A = \D$ or \A\ is in the left subtree of \D.
	Similarly, either $\B = \D$ or \B\ is in the right subtree of \D.
	So, each node \E\ along the path from \A\ to \D\ satisfies $\E \leq_c \D \leq_c \B$.
	Since $\C >_c \B$, \C\ is not one of the nodes on the path from \A\ to \D,
	so \C\ is a proper ancestor of \D, and therefore of \B.
\end{proof}

\begin{lemma}
\label{lem:new-nodes-ancestors}
	If the traversal of a \op{find} operation starts at node \V, then the nodes visited by the traversal to the right of \V\ are ancestors of \V.
\end{lemma}

\begin{proof}
	Let $\X_1, \X_2, \ldots$ 
	be the sequence of distinct nodes visited whose counter values are greater than \V, in the order in which they were first visited.
	We proceed by induction on this set of nodes.
	Consider the node $\X_1$.
	The traversal must have arrived at this node by  leaving a node \U, where $\U \leq_c \V$.
	By Lemma \ref{order} this traversal followed \co{U->right}, so it must have been an upward traversal.
	By Lemma \ref{lem:up-traversal}, $\X_1$ is an ancestor of \U.
	By Lemma \ref{lem:common-ancestor}, $\X_1$ must also be an ancestor of \V.
	For the inductive step, suppose $\X_i$ is an ancestor of \V\ for all $i < j$.
	It suffices to show that $\X_j$ is also an ancestor of \V.
	There are two ways the traversal could have arrived at $\X_j$: (1) from some node $\X_i$ where $i < j$, (2) from some node \U\ with $\U \leq_C \V$.
	Case (2) is the same as the base case, so we consider case (1).
	Since an upward traversal was performed from $\X_i$ to $\X_j$, by Lemma \ref{lem:up-traversal}, $\X_j$ must be an ancestor of $\X_i$.
	Since $\X_i$ is an ancestor of \V, $\X_j$ is also an ancestor of \V.
	Therefore, $\X_1, \X_2, \ldots$ are all ancestors of \V.
\end{proof}

\Eric{Note change in wording of following lemma; the old wording said c is the number of tryAppends concurrent with the find.  But its possible to read V from the head, then 1000 nodes are appended to list, then find(V,t) is called.  Those 1000 tryAppends could contribute to c parameter used to bound the time for the find even though they are not concurrent with the find.}
\begin{lemma}
	\label{lem:new-nodes-visited}
	If the traversal of a \op{find} operation starts at node \V, the number of distinct nodes visited to the right of \V\ is bounded above by both the depth of \V\ in the implicit tree and $O(\log c)$, where $c$ is the number of successful \op{tryAppend} operations \er{from the time \V\ was the head of the list and the end of the} \op{find} (on the same object).
\end{lemma}

\begin{proof}	
	Let $\X_1, \X_2, ..., \X_k$ be the $k$ distinct nodes visited by the \op{find} operation whose counter values are larger than \V's, sorted so that $\X_1 <_p \X_2 <_p \cdots <_p \X_k$.
	Each of these $k$ nodes must be an ancestor of \V, so $k$ is at most the depth of \V.
	To prove the second part of the claim, it suffices show that $c$
	is at least $2^{k-1}$.
	Each node $\X_i$ was added by a \co{tryAppend} between the time \V\ was the head of the list and the end of the \op{find}.
	Since $\V <_c \X_i$, \V\ must be in the left subtree of $\X_i$ for all $i$.
	Thus, all of $\X_2, \ldots, \X_k$ are in the left subtree of $\X_1$.
	This subtree (like all left subtrees) is a complete binary tree.
	Hence, $\X_2$ is the root of a complete binary tree that contains
	all the nodes $\X_2, \X_3, ..., \X_k$, so its height is at least $k-1$. This means that the right subtree of $\X_2$ has $2^{k-1}$ nodes and all of these nodes have counter values greater than \V,
	which means that $c \geq 2^{k-1}$.
\end{proof}

\memoryReclamationTime* 

\begin{proof}
\Hao{requires editing}
Since \op{tryAppend} and \op{getHead} have no loops, they perform $O(1)$ steps, plus $O(1)$ operations on reference-counted pointers.
By Lemma \ref{lem:remove-with-gc}, we know that \co{remove} takes amortized $O(1)$ time plus amortized $O(1)$ operations on reference counted pointers.

Next, we  show that the number of steps and reference-counted pointer operations
performed by a \op{find} is $O(n + \min(d, \log{} c))$.
It suffices to show that the number of distinct nodes visited by the \op{find} operation is $O(n + \min(d, \log{} c))$ because it follows from Lemma \ref{lem:traversal} that the number of upward and forward traversals is at most the number of distinct nodes visited.
Each forward and upward traversal contributes $O(1)$ steps plus the cost of $O(1)$ reference counted pointer operations.

By Lemma \ref{lem:new-nodes-visited}, we know that the number of distinct nodes visited whose counter values are greater than \V's is at most $O(\min(d, \log{} c))$.
Let $S$ be the set of nodes visited whose counter values are less than \V's.
Next, we argue that $|S|$ is at most $O(n)$.
This is because whenever we do a forward or upward traversal, we know that the timestamp of the source node is larger than the timestamp of the \op{find} operation.
Thus, all nodes in $S$ have timestamp greater than \ts\ except one.
All nodes in $S$ must have also been reachable from \V\ at the start of the \op{find} operation. \Hao{needs some justification}
Therefore, there are at most $O(n)$ nodes in $S$.

The final step is to argue that the reference counting operation performed by the version list operations all take $O(1)$ amortized expected time.
This is a little tricky to show because there could be operations on reference counted pointers outside of the version list operations.
However, if we example the proof of Lemma~\ref{lem:ref-count}, no decrement operations will ever be charged to operations on reference counted pointers outside of the version list operations.
Therefore, the amortized expected time complexity of \op{tryAppend}, \op{getHead}, \op{remove}, and creating a new version list is $O(1)$, and the amortized expected time complexity of \op{find(V, ts)} is $O(n + \min(d, \log c))$.
%
%
\end{proof}

\subsection{Space Bounds for Memory-Managed Version Lists}

\Hao{When we say reference counted pointer, we mean either \rcptr{} or \arcptr{}.}

\er{The following observation follows from the fact that the RC scheme of \cite{ABW21} 
has $O(P^2)$ delayed decrements, as discussed in Section \ref{sec:smr-list}.}

\begin{observation}
	\label{obs:rc-ptr-space}
	If there are a total of $O(K)$ reference-counted pointers to a particular type of object, then the number of objects of that type that have been allocated, but not freed is $O(K+P^2)$.
\end{observation}
\memoryReclamationSpace* 

\begin{proof}
	Let $T$ be the total number of \lrreachable{} nodes across all the version lists.
	By Lemma~\ref{lem:remove-with-gc}, we know that Theorem~\ref{space-bound} holds for the modified version list algorithm in Section~\ref{sec:smr-list}, so $T \in O((L-R) + P \log L)$.
	To prove this theorem, it suffices to show that the number of \vnode{}s and Descriptors is at most $O(T + (P^2 + K) \log L)$.
	We begin by defining some useful terminology.
	An \emph{\upptr{}} is a \er{\co{left} or \co{right}} pointer from a \vnode{} \A\ to another \vnode{} \B\ such that \B\ is an ancestor of \A\ in the implicit tree.
	Similarly, a \emph{\downptr{}} is a pointer from \A\ to \B\ such that \B\ is a descendant of \A\ in the implicit tree.
	By Lemma \ref{anc-or-desc}, every pointer between two active \vnode{}s is either an \upptr{} or a \downptr{}.
	A \emph{source node} is an \lrunreachable{} node that is not pointed to by any \upptr{}s.
	
	We begin by proving the claim that all \lrunreachable{} nodes are reachable by following \upptr{}s starting from some source node.
	Starting from an \lrunreachable{} node \A, consider the following process: if \A\ is a source node, then terminate; otherwise there exists an \lrunreachable{} node \B\ with an \upptr{} to \A\ (\B\ cannot be \lrreachable{} because that would mean \A\ is \lrreachable{} as well), so we repeat this process starting from \B.
	This process is guaranteed to terminate because \upptr{}s do not form a cycle, and this proves the claim.
	
	Next, we show that there are at most $O(T + P^2)$ Descriptors that have been allocated but not freed.
	This is because there are at most $O(P)$ Descriptors pointed to by local pointers (these pointers are temporarily held by a process when it is in the middle of a version list operation), and $O(T)$ Descriptors pointed to by \lrreachable{} nodes.
	Nodes that are \lrunreachable{} do not point to any non-\co{frozen} Descriptor objects.
	This is because a \vnode{} cannot point to any Descriptor object until it is active, and if a node is active and \lrunreachable{}, then it must have been finalized, by Lemma \ref{unfinalized-reachable}, so its descriptor fields have already been frozen, by Lemma \ref{frozen-before-rr}.
	Therefore, there are at most $O(T+P)$ pointers to Descriptor objects, and by Observation \ref{obs:rc-ptr-space}, there are at most $O(T+P^2)$ Descriptor objects that have been allocated, but not freed.
	
	It remains to show that there are $O(T + (P^2 + K) \log L)$ \vnode{}s and we prove this by bounding the number of \lrunreachable{} \vnode{}s by the same amount.
	
	Next, we show that there are  $O(T + P^2 + K)$ source nodes. 
	This is because source nodes are either pointed to by (1) Descriptor objects, (2) a local pointer held during version list operations, (3) \downptr{}s or (4) user pointers. There are  $O(T + P^2)$ pointers of type (1), $O(P)$ of type (2) and (3), and $O(K)$ of type (4).
	To show the bounds for type (3), there are only two ways for a \lrunreachable{} node to be pointed to by a \downptr{}. 
	First, it could be that the node pointing to it has been spliced out and its \co{left} or \co{right} pointer has not yet been set to $\top$. 
	Second, it could be that the node pointing to it was newly allocated for a \op{tryAppend} that is now guaranteed to fail. 
	There can only be a constant number of such cases per process, which proves the bound.
	The \co{Head} pointer of a version list can never point to a source node, since
	we set the \co{Head} pointer to \op{null} before calling \co{remove} on the head of a version list, \Hao{I think we should change the code to do this.} which ensures \co{Head} always points to an \lrreachable{} node.
	Therefore, there is a total of $O(T + P^2 + K)$ reference-counted pointers to source nodes, which results in $O(T + P^2 + K)$ source nodes that have been allocated and not freed, by Observation \ref{obs:rc-ptr-space}.
	
	From each source node, there are at most $O(\log L)$ nodes reachable by following \upptr{}s. 
	This is because all the reachable nodes are ancestors of the source node, and each node can only have $O(\log L)$ ancestors.
	However, some of these $O(\log L)$ nodes are \lrreachable{} and we only want to count the \lrunreachable{} ones.
	It turns out that whenever a source node \C\ is pointed to by a Descriptor installed in a \vnode{}, the left and right neighbors of \C\ are both \lrreachable{}, so all of the $O(\log L)$ nodes reachable from \C\ by following \upptr{}s are \lrreachable{}. \Hao{Ideally requires some more justification if we have time before the deadline.}\Eric{Seems plausible (though I'm not 100\% convinced) but non-trivial to prove}
	There are $O(T)$ source nodes that are pointed to by a Descriptor installed in a VNode and $O(P^2 + K)$ sources nodes that are not.
	For each source node in the latter category, there could be $O(\log{} L)$ \lrunreachable{} nodes that are reachable from that node by following \upptr{}s. 
	Therefore, the number of \lrunreachable{} nodes is bounded by $O(T + (P^2 + K) \log{} L)$.
\end{proof}


\section{Details of Snapshottable Data Structure Application}
\label{sec:application-details}


\subsection{Pseudo-code for Memory-Managed VersionedCAS Objects}
\label{vcas-description}

\begin{figure*}[!th]\small
	\hspace{-10mm}
	\begin{minipage}[t]{.55\textwidth}
		\begin{lstlisting}[linewidth=.99\columnwidth, numbers=left,frame=none]
class Camera {
	int timestamp;
	@\hil{RangeTracker<VNode*> rt;}@
	Camera() { timestamp = 0; }   @\label{line:ss-con}@
	
	int takeSnapshot() {
		int ts = @\hil{rt.Announce(\&timestamp);}@ @\label{line:read-time}@
		CAS(&timestamp, ts, ts+1);    @\label{line:inc-time}@
		return ts; }
	
	@\hil{void unreserve()}@ {
		@\hil{rt.\Unannounce();}@ } };

class VNode {
	Value val; int ts;
	@\hil{arc\_ptr}@<VNode> left, right; 
	// the remaining variables are only
	// used by remove() and tryAppend()
	@\hil{int priority; int counter;}@
	@\hil{enum status; int refCount;}@
	VNode(Value v) {
		val = v; ts = TBD; } };

class VersionedCAS {
	@\hil{VersionList VList;}@
	Camera* C;
	VersionedCAS(Value v, Camera* c) {           @\label{line:vcas-con}@
		C = c;
		@\hil{rc\_ptr}@<VNode> node = new VNode(v);
		@\hil{VList.tryAppend(null, node);}@
		initTS(node); }	
				\end{lstlisting}
			\end{minipage}\hspace{.3in}
			\begin{minipage}[t]{.55\textwidth}
				\StartLineAt{32}
				\begin{lstlisting}[linewidth=.99\textwidth, numbers=left, frame=none]
void initTS(@\hil{rc\_ptr}@<VNode> n) {
	if(n->ts == TBD) {               @\label{line:init1-gc}@
		int curTS = C->timestamp;       @\label{line:readTS-gc}@
		CAS(&(n->ts), TBD, curTS); } }   @\label{line:casTS-gc}@

Value readVersion(int ts) {
	@\hil{rc\_ptr}@<VNode> node = @\hil{VList.getHead()}@;     @\label{line:readss-head-gc}@
	initTS(node);                     @\label{line:readss-initTS-gc}@
	@\hil{node = VList.find(node, ts);}@
	return node->val; } // node cannot be null       @\label{line:readss-ret-gc}@

Value vRead() {
	@\hil{rc\_ptr}@<VNode> head = @\hil{VList.getHead()}@;         @\label{line:readHead-gc}@
	initTS(head);                         @\label{line:read-initTS-gc}@
	return head->val; }

bool vCAS(Value oldV, Value newV) {
	@\hil{rc\_ptr<VNode>}@ head = @\hil{VList.getHead()}@;         @\label{line:vcas-head-gc}@
	initTS(head);                         @\label{line:vcas-initTS-gc}@
	if(head->val != oldV) return false;   @\label{line:vcas-false1-gc}@
	if(newV == oldV) return true;         @\label{line:vcas-true1-gc}@
	rc_ptr<VNode> newN = new VNode(newV);        @\label{line:newnode-gc}@
	if(@\hil{VList.tryAppend(head, newN)}@) {        @\label{line:append-gc}@
		initTS(newN);                       @\label{line:initTSnewN-gc}@
		@\hil{List<VNode*> redundant = C.rt.\retireInterval(}@
		@\hil{toRawPtr(head), head.ts, newN.ts);}@
		@\hil{forevery(node in redundant)}@ {
			@\hil{// \textit{convert from raw pointer to rc\_ptr}}@
			@\hil{VersionList.remove(rc\_ptr<VNode>(node));}@ } @\label{line:to-rc-ptr}@
		return true; }                      @\label{line:vcas-true2-gc}@
	else {
		delete newN;                        @\label{line:scdelete-gc}@
		initTS(@\hil{VList.getHead()}@);                      @\label{line:vcas-initTSHead-gc}@
		return false; } } };                   @\label{line:vcas-false2-gc}@
			\end{lstlisting}
		\end{minipage}
		\vspace{-.1in}
		\caption{Indirect \vCAS{} implementation (black text) using the range-tracking and version list maintenance interface (blue text). Functions \texttt{announce}, \texttt{unannounce} and \retireInterval\ are defined in the range-tracking object (see Section \ref{sec:identify}). Functions \texttt{tryAppend} and \texttt{remove} are defined in the version list maintenance interface (see Section \ref{sec:vlists}).
		}
		\label{fig:vcas-alg-youla-gc}
	\end{figure*}

Wei et al.'s implementation~\cite{WBBFRS21a} of VersionedCAS objects appears as the black text of
Figure~\ref{fig:vcas-alg-youla-gc}.
Additions we make to the algorithm for memory management are shown in blue.  
A \co{takeSnapshot} operation
attempts to increment the counter of the Camera object and returns the old value of the counter as a snapshot handle.
Each vnode contains a value
that has been stored in the VersionedCAS object, a pointer to the next older vnode and a timestamp
value read from the associated Camera object.
A \co{vRead} obtains the current value of the VersionedCAS object from the head of the list.
A $\co{vCAS}(old,new)$ operation checks that the head of the list represents the value $old$ and, if so,
attempts to add a new vnode to the head of the list to represent the value $new$.
A \co{versionedRead} operation using a timestamp $s$ returned by a \co{takeSnapshot}
traverses the version list until it finds a vnode whose timestamp
is less than or equal to $s$, and then returns the value stored in that vnode.
To ensure linearizability, a number of helping mechanisms are incorporated into the
implementation.  See \cite{WBBFRS21a} for details.

To facilitate garbage collection, we add an \op{unreserve} operation to the Camera object, which a process can use
to release a particular snapshot handle, guaranteeing that the process will not subsequently use it
as an argument to \co{readVersion}.

\subsection{Proof of time bounds}

\begin{observation}
	\label{obs:deprecate-remove}
	\item An execution with $M$ calls to \op{vCAS} contains $M$ calls to \op{deprecate} and at most $M$ calls to \op{remove}.
\end{observation}

\applicationTime* 

\begin{proof}
%
Each snapshot query begins with a \op{Camera.takeSnapshot} returning some timestamp \op{ts}, ends with a \op{Camera.unreserve} and performs some \op{readVersion(ts)} operations in between.
By inspection of the pseudo-code and by Theorem~\ref{step1-time}, we can see that the \op{takeSnapshot} and \op{unreserve} operations take $O(1)$.
If each \op{readVersion} operation took $O(1)$ time as well, then time complexity of the snapshot query would be the same as its sequential complexity.
Unfortunately, \op{readVersion} does not take constant time because it calls \op{find}.
We will prove the following claim: a \op{readVersion(ts)} on VersionedCAS object $O$ has amortized expected time complexity at most the number of successful \op{vCAS} operations on object $O$ concurrent with the snapshot query.
Assuming that the snapshot query calls \op{readVersion} on each object at most once (which is true in~\cite{WBBFRS21a} because values returned by \op{readVersion}s are cached so there's no need to read the same object twice), this means that the amortized expected time complexity of the \op{readVersion} is proportional to its sequential complexity plus the number of successful \op{vCAS} operations it is concurrent with.
Therefore, to prove the desired time bound for \op{readVersion}, it suffices to prove the previous claim.

To prove the claim, we first note that a \op{readVersion(ts)} operation starts by performing \op{getHead} and this takes amortized expected constant time by Theorem~\ref{thm:step3-time}.
Then it calls \op{initTS} on $V$ (which also takes constant time) followed by a \op{find(*, ts)}.
We know that all nodes visited by the \op{find(*, ts)} operation, except for the last one, have timestamp greater than \op{ts}.
The value of the global timestamp at the start of \op{readVersion} is less than or equal to \op{ts}, so the vCAS operations that added these nodes must have been linearized after the start of the \op{readVersion}.
Therefore, the number of distinct nodes visited by the \op{find} operation is bounded by one plus the number of successful vCAS operations on the same object that were concurrent with the \op{readVersion}.
By Lemma \ref{lem:traversal}, we know that the total number of forward and upward traversals performed by a \op{find} is proportional the number of distinct nodes visited.
Each forward or upward traversal takes amortized expected constant time, so the \op{readVersion} has amortized expected time complexity $O(C)$ where $C$ is the number of \op{vCAS} operations on the same object that are concurrent with the snapshot query.
Thus the claim holds.

Finally, the amortized expected time complexity of frontier operations remains unchanged because both \op{vRead} and \op{vCAS} take amortized expected constant time.
\end{proof}

\subsection{Proof of space bounds}


As in Section~\ref{sec:application}, we consider a \dnode{} to be \emph{necessary} if it is either not freed by the original data structure's memory reclamation scheme or if there exists an announced timestamp in between its birth (inclusive) and retire (exclusive) timestamp.
We say that a \vnode{} is \emph{necessary} if it is pointed to by a \dnode{} that has not yet been deprecated (i.e., freed by the original data structure's memory reclamation scheme) or if it has been deprecated but the interval it was assigned contains an announced timestamp.
Let $D_{max}$ and $V_{max}$ be the maximum number of necessary \dnode{}s and \vnode{}s \er{at any one time during an execution}.

\begin{theorem}
	\label{thm:dnode-space}
	In any configuration, there are $O(D_{max} + P^2 \log P)$ \dnode{}s that have been allocated but not freed.
\end{theorem}

\begin{proof}
	We begin by noting that a \dnode{} is deprecated as soon as it is freed by the original data structure's memory reclamation scheme.
	The set of \dnode{}s that have been allocated but not freed can be partitioned into three sets, one set $A$ for \dnode{}s that have not been freed by the original data structure's memory reclamation scheme, another set $B$ for \dnode{}s that have been passed to \retireInterval\ but not returned, and another set $C$ for \dnode{}s that have been returned by some \retireInterval\ but not freed.
	Since any \dnode{} that is not yet deprecated is necessary, $|A| < D_{max}$.
	The nodes returned by \retireInterval\ are freed before the next call to \retireInterval\ by the same process, so by Lemma \ref{lem:deprecateRetVal}, each process contributes at most $4 P \log P$ nodes to the set $C$, so $|C| \leq 4P^2 \log P$.
	
	To complete the proof, it remains to bound $|B|$.
	Recall that a deprecated object is said to be \emph{needed} if its interval contains an active announcement.
	This means that all needed deprecated \vnode{}s are also necessary.
	Therefore $D_{max}$ is an upper bound on the number of needed deprecated \vnode{}s across all configurations of the execution.
	By Theorem \ref{step1-space}, $|B| \leq 2D_{max} + 25 P^2\log P$, which completes the proof.
\end{proof}

Assuming each \dnode{} is constant-sized, we can prove the following theorem.

\applicationSpace* 

\begin{proof}
Theorem \ref{thm:dnode-space} bounds the total number of \dnode{}s by $O(D_{max} + P^2 \log P)$.
To bound the number of \vnode{} and Descriptor objects, we make use of Theorem \ref{thm:step3-space}.
In Figure~\ref{fig:vcas-alg-youla-gc}, we can see that each process holds on to a constant number of \vnode{} pointers so we set the parameter $K$ in Theorem~\ref{thm:step3-space} to $O(P)$.
To show that there are not too many \vnode{}s and Descriptors, we need to bound the number of successful \op{tryAppend}s ($L$) minus the number of \op{remove}s ($R$) by $O(D_{max} + V_{max} + P^2 \log P)$.
This is equivalent to bounding the number of \vnode{}s that have been appended, but not passed to remove.
We divide these \vnode{}s into disjoint sets and show that there are:

\begin{enumerate}
	\item $O(D_{max})$ \vnode{}s that have been appended and are still at the head of the version list,
	\item $O(P)$ \vnode{}s that have been replaced as the head of the version list but are not yet deprecated,
	\item $O(V_{max} + P^2 \log P)$ \vnode{}s that have been deprecated, but not yet returned, and
	\item $O(P^2 \log P)$ \vnode{}s that have been returned, but not yet passed to remove.
\end{enumerate}

The first bound holds because each \dnode{} contains a constant number of versionedCAS objects, which are essentially pointers to heads of version lists.
The second bound holds because a \vnode{} is deprecated immediately after a new head is appended and there can only be $P$ processes stalled in between these two steps.
The third bound holds because every needed deprecated \vnode{} is also a necessary \vnode{}, so $V_{max}$ is an upper bound on $H_{max}$ from Theorem~\ref{step1-space}. Therefore, by Theorem~\ref{step1-space}, at most $O(V_{max} + P^2 \log P)$ \vnode{}s have been deprecated, but not yet returned.
The final bound holds because each call to \retireInterval\ returns at most $4P \log P$ \vnode{}s (Lemma~\ref{lem:deprecateRetVal}) and \vnode{}s are passed to remove as soon as they are returned by some \op{deprecate}.

Summing up over all the cases, we see that $(L-R) \in O(D_{max} + V_{max} + P^2 \log P)$. 
By Theorem~\ref{thm:step3-space}, the number of \vnode{}s and Descriptors that have been allocated and not freed is at most $O(D_{max} + V_{max} + P^2 \log P + P^2\log L)$.
Since each \dnode{}, \vnode{} and Descriptor takes a constant amount of space, the overall space usage of these objects is $O(D_{max} + V_{max} + P^2 \log P + P^2 \log L)$.
\end{proof}

\end{document}